\pdfoutput=1
\documentclass[12pt,letter]{article}
\usepackage{bbm}
\usepackage[english]{babel}
\usepackage[utf8]{inputenc}
\usepackage{amssymb}
\usepackage{amsfonts}
\usepackage{amsmath}
\usepackage{verbatim}
\usepackage{chngcntr}
\usepackage{tikz}
\usepackage{color,soul}
\usepackage{import}
\usepackage{float}
\usepackage{rotating}
\usepackage{supertabular}
\usepackage{longtable}
\usepackage[flushleft]{threeparttable}
\usepackage{indentfirst}
\usepackage{booktabs}
\usepackage{setspace}
\usepackage{amsmath}
\usepackage{dsfont}

\graphicspath{{draft_2_graphics/}{draft_2_tcache/}{draft_2_gcache/}}
\DeclareGraphicsExtensions{.pdf,.eps,.ps,.png,.jpg,.jpeg}

\usepackage{amsthm}
\usepackage{lipsum}

\usepackage{caption}
\usepackage{subcaption}

\newtheoremstyle{custom}
                {\topsep}
                {\topsep}
                {\itshape}
                {}
                {\bfseries}
                {}
                {\newline}
                {}
\theoremstyle{custom}

\newtheorem{theorem}{Theorem}

\newtheorem{definition}{Definition}

\newtheorem{lemma}{Lemma}

\newtheorem{proposition}{Proposition}

\newtheorem{assumption}[theorem]{Assumption}


\let\oldproofname=\proofname
\renewcommand{\proofname}{\rm\bf{\oldproofname}}

\usepackage{natbib}
\usepackage{tikz}
\usetikzlibrary{shapes,arrows}
\usepackage[font={large,bf}]{caption}
\usepackage{memhfixc}
\usepackage{eso-pic}
\usepackage{lscape}
\usepackage{pdflscape}
\usepackage{eepic}
\usepackage{array}
\newcolumntype{P}[1]{>{\centering\arraybackslash}p{#1}}

\usepackage[hmargin=0.95in,vmargin=1in]{geometry}      
\usepackage{longtable}
\usepackage{fancyhdr}
\usepackage{multirow}
\usepackage{hyperref}
\usepackage{hhline}

\usepackage[section]{placeins}

\usepackage{tikz}
\usetikzlibrary{decorations.pathreplacing,angles,quotes}
\usetikzlibrary{shapes,arrows}
\usetikzlibrary{calc}
\tikzset{
  font={\fontsize{11pt}{11}\selectfont}}

\graphicspath{{Graphs/}}

\graphicspath{{FiguresTables/Figures/}}

\usepackage{hyperref}
\hypersetup{
    colorlinks = true,
    linkcolor={red!80!black},
    citecolor={blue!50!black}
}
\renewenvironment{abstract}
 {\small
  \begin{center}
  \bfseries \abstractname\vspace{-.5em}\vspace{0pt}
  \end{center}
  \list{}{%
    \setlength{\leftmargin}{4mm}
    \setlength{\rightmargin}{\leftmargin}%
  }%
  \item\relax}
 {\endlist}

\setcitestyle{square}

\title{The Optimal Size and Progressivity \\ of Old-Age Social Security}
 
\author{Francisco Cabezon\footnote{Email: jfcabezon@princeton.edu, Address: 20 Washington Rd, Princeton, NJ 08540} \\ Princeton University}

\begin{document}

\maketitle

\begin{abstract}
Almost every public pension system shares two attributes: earning deductions to finance benefits, and benefits that depend on earnings. This paper analyzes theoretically and empirically the trade-off between social insurance and incentive provision faced by reforms to these two attributes. First, I combine the social insurance and the optimal linear-income literature to build a model with a flexible pension contribution rate and benefits' progressivity that incorporates inter-temporal and inter-worker types of redistribution and incentive distortion. The model is general, allowing workers to be heterogeneous on productivity and retirement preparedness, and they exhibit present-focused bias. I then estimate the model by leveraging three quasi-experimental variations on the design of the Chilean pension system and administrative data merged with a panel survey. I find that taxable earnings respond to changes in the benefit-earnings link, future pension payments, and net-of-tax rate, which increases the costs of reforms. I also find that lifetime payroll earnings have a strong positive relationship with productivity and retirement preparedness, and that pension transfers are effective in increasing retirement consumption. Therefore, there is a large inter-worker redistribution value through the pension system. Overall, there are significant social gains from marginal reforms: a 1\% increase in the contribution rate and in the benefit progressivity generates social gains of 0.08\% and 0.29\% of the GDP, respectively. The optimal design has a pension contribution rate of 17\% and focuses 42\% of pension public spending on workers below the median of lifetime earnings.

\end{abstract}
\vspace{1cm}

\newpage

\section{Introduction}

At the core of the contract between the government and workers in a public pension system is how much workers are forced to contribute from their active earnings and how much they receive as a function of those earnings when they retire. Consequently, almost every public pension system shares two attributes: earning deductions to finance benefits, and benefits that depend on earnings. Population aging is forcing governments to revisit pension system designs making the welfare analysis to reforms of these two characteristics a crucial input of public policy. Is it socially desirable to increase public pension size by increasing earning deductions? Is it socially desirable to change the relationship between benefits and lifetime earnings, to make them more or less progressive?

The trade-off between incentive provision and income redistribution is at the core of the answer to these questions (\cite{DIAMOND_framework_ss_analysis}). An increase in the pension contribution rate transfers income from active to retired periods, while a change in benefits' progressivity transfers retirement income across workers with different history of lifetime earnings. Therefore, reforms to these two parameters do income redistribution across time and across workers at the cost of distorting workers' incentives to generate earnings, which is costly. The social desirability of these reforms must weigh the distributional gains with the behavioral distortion costs. 
[empirical challenge: they have not been estimated before]
Incorporating these elements in the analysis is theoretically and empirically challenging. For this reason, the social desirability of increasing the pension contribution rate and benefits' progressivity remains important, but open questions.

In this paper, I overcome these challenges and evaluate the social desirability of reforming benefits' progressivity and pension contribution rate from payroll earnings. First, I leverage in the social insurance (\cite{baily1978some} and \cite{chetty2006general}) and the optimal linear-income literature (\cite{sheshinski1972optimal} among others) to build a model that incorporates inter-temporal and inter-worker types of redistribution, and    distortion. From the model, I build expressions for the social gains of reforming the pension contribution rate and benefits' progressivity, which I then estimate empirically using administrative and survey data on the Chilean pension system.

To estimate these expressions empirically, it is necessary to first estimate the behavioral responses to the reforms. Thus, I first decompose the responses into three elements: the response of taxable earnings to future pension benefit, the earnings-benefit link, and the net-of-tax rate. Then, I estimate each of these elements separately using three sources of quasi-experimental variation. 

Once I estimate the responses, I proceed to estimate the inter-temporal and inter-worker gains of reform to the earning deductions rates and the progressivity of the pension system. To measure these effects, I first  show that the inter-temporal and inter-worker gains of reform can be estimated as function of the covariance between lifetime earnings and active consumption and the consumption drop at retirement, and the marginal propensity to consume from pension benefit. Then, I estimate these statistics empirically by using a panel survey merged with the administrative data. 

My estimations provide two sets of results. First, the behavioral responses to the pension contribution and benefits' progressivity are economically significant and increase the costs of the reforms. Second, I find that workers are bad-prepared for retirement, retirement preparedness has a strong correlation with lifetime earnings, and an increase in pension benefits increases retirement consumption. This second set of findings shows a positive value for inter and intra-worker redistribution through the pension system. When these two sets of results are put together, the redistribution and behavioral distortion trade-off is solved in favor of the redistributional gains of reforms, making it socially desirable for an utilitarian central planner to increase benefits' progressivity and the pension contribution rate.

My paper makes three contributions. The first contribution is to provide a framework to analyze the welfare effects of changes in the pension contribution rate and benefits' progressivity, while incorporating behavioral bias, and heterogeneity on productivity and retirement preparation. The framework is simple enough to distinguish between the redistribution value generated by redistribution across workers and time, and the cost of externalities and internalities generated by incentive distortion to generate taxable income (\cite{farhi_behavioral_earnings}). From this general framework, I provide a lower bound on the social gains of reforms as a function of moments that are estimable in the data.

The second contribution is to provide empirical evidence that workers are forward-looking to incentives, responding before retirement to incentives generated by retirement benefits (income and substitution). Specifically, I find that, in response to an unconditional increase in future pension benefits, workers reduce their pre-retirement taxable earnings. Similarly, workers reduce their taxable earnings if the relationship between earnings and future pension is reduced. These findings are novel and add to the incipient literature on the link between taxable earnings and pension benefit (\cite{french_pen_cont_link}). In a similar analysis, I also find that earnings deductions, like payroll taxes, also have a negative effect on gross taxable earnings. Overall, I provide evidence that the pension system is distorting, which is relevant to its design.

The third contribution is to measure the redistributional role that the pension system serves. I find that consumption at retirement is strongly correlated with lifetime payroll earnings. Furthermore, this strong relationship is driven by two forces. First, active life consumption increases strongly with lifetime payroll earnings, making the latter a good tagging of workers’ productivity. Second, the consumption drop at retirement strongly decreases with lifetime earnings, i.e., low-income workers are less prepared for retirement, suffering more from it. These findings add to the literature on heterogeneous retirement preparedness and pension systems design (\cite{kolsrud2021retirement}). Additionally, I provide causal evidence that retirement consumption increases with retirement benefits, with a marginal propensity to consume from retirement earnings of around 0.8. Overall, these findings support the effectiveness of providing social insurance against lifetime productivity and retirement preparedness by doing interworker redistribution through the pension system.

I do my analysis in the context of the Chilean pension system, whose simplicity allows me to isolate the different components of pension system's income redistribution and incentive distortion trade-off. In Chile, old-age pensions have two sources: a forced defined contribution plan and a government PAYGO defined benefit plan. Workers are forced to save 10\% of their income which is invested until retirement, and the government complements that self-funded pension with a subsidy that phases out with the self-funded pension. The whole system is defined by three parameters, the mandatory saving rate and two parameters that govern the design of the subsidy. Figure \ref{chilean_design} shows the subsidy design, which is defined by the minimal pension and largest pension with subsidy. This simplicity facilitates the empirical estimation of the behavioral responses to reforms of benefits' progressivity and the pension contribution rate. 

I structure the paper in three parts: in the first part, I introduce the analytical framework, then I estimate its elements empirically, and finally, I put together the different estimated elements to evaluate the welfare gains of reforming the pension contribution rate and benefits' progressivity.   

In the first part of the paper, I embed a pension system with flexible pension contribution rate and benefits' progressivity into a two-period life-cycle model, where workers exhibit present-focused bias, and they have heterogeneous productivity and access to saving instruments. These three workers' characteristics play different roles in the model. The heterogeneous productivity generates dispersion in lifetime earnings and active life consumption (\cite{saez2001elasticities}), while heterogeneous access to saving instruments rationalizes the heterogeneity on consumption drop at retirement, i.e, retirement preparation (\cite{kolsrud2021retirement}). These two characteristics creates the inter-temporal and inter-worker value for redistribution of reforms. 

The present-focused bias, which is fundamental for the existence of pension systems, has two important consequences on the results. First, in addition to fiscal externality, common to welfare program analysis, behavioral responses also generate an internality on workers welfare (\cite{farhi_behavioral_earnings}). Second, the timing of retirement benefits consumption has an impact on welfare. For example, if a worker consumes all the retirement benefit while active (through borrowing), then to increase retirement benefit does not have an inter-temporal redistribution value. 

From the model I derive two optimality conditions, one for each pension system parameter (contribution rate and progressivity). This conditions can be interpreted as the incentive distortion and income redistribution trade-off. I build a lower bound for these conditions, which I estimate empirically in the second part of the paper.

In the empirical part of the paper, I first focus on the cost of incentive distortion generated by reforms. Reforms distort workers' incentive to produce taxable earnings. An increase in the pension contribution rate delays the reception of taxable earnings, while an increase in benefits progressivity reduces the benefit-contribution link. With the model, I show that these responses are defined by three elasticities: that of taxable earnings to payroll taxes, to taxes on mandatory pension savings, and future pension payment. I use three quasi-experimental variation and admin data to estimate them. 

To estimate the response of taxable earnings to future pension payments, I use the fact that forced pension savings are invested in the capital market, being exposed to idiosyncratic market returns. The investment strategy of these savings depends on age. Specifically, as worker ages, their savings are discretely transferred from more risky investments (stock heavy) to safer ones (bond heavy). I leverage this age-dependent investment strategy and the shock to stocks' returns during the Global Financial Crisis to build a shock to pension savings that depends on date of birth by month. I find that, due to the financial crisis and this age-dependent strategy, workers born less than a year apart have over 25\% difference in their pension savings return obtained in the last ten years before retirement. Then, I use this shock to pension savings as an instrument for future pension payments. In a second stage, I use the instrumented future pension benefit to estimate the effect of future pension benefit on pre-retirement taxable earnings. I found that the elasticity of taxable earnings with respect to future pension payments is around 0.11. 

To estimate the response of taxable earning response to the benefit-contribution link, I use the introduction of the subsidy to self-funded pension in 2008. This subsidy phases-out with mandatory pension savings, thus created an implicit tax on the pension savings and, thus, its introduction reduced the benefit-contribution link for future recipients. To estimate the elasticity, I use a diff-diff design, where the first difference is the given by the timing of the subsidy introduction, and the second one is given the recipient status. To exogenous assign the treatment status, I use the savings accumulated by the worker before the subsidy introduction jointly with a discontinuity on the subsidy design in a first-stage. I find that workers respond the the benefit-earnings link with an elasticity of around 0.22, which is similar, but smaller than the one estimated in the context of a reform to the benefits-earnings link in Poland estimated by \cite{french_pen_cont_link}.

To estimate the final elasticity, that of taxable earnings to a payroll tax, I use the fact that pension fund administrators charge a management fee to workers that act equivalently to a payroll tax and is set independently by the pension fund administrators (PFA). In 2014 one of the PFA changed fees, while the rest kept theirs constant. I use a difference-in-differences design, using workers affiliated with the other PFAs as controls, to estimate the effect that change in the net=of-tax rate has on taxable income. I find that the earnings elasticity to net-of-tax rate is around 0.38, somewhat larger than estimates for Sweden (\cite{saez2019payroll} and \cite{EGEBARK2018163}), smaller than those for Hungary (\cite{payroll_hungary}) and similar to those of Argentina (\cite{payroll_argetina}). I find smaller elasticities for high income workers, which is in line with the findings of \cite{saez2012earnings}. 

Then, I empirically estimate the redistributional gains of reforms. I show that, conditional on relative risk aversion and a retirement preference parameter, redistributional gains of reforms can be decomposed on three statistics: the covariance between lifetime earnings and active life consumption, the covariance of lifetime earnings and consumption drop at retirement, and the marginal propensity to consume from retirement income. The first two covariances capture, respectively, the heterogeneity on productivity and retirement preparation and how good are transfers based on lifetime earnings to hedge against them. I find that lifetime earnings exhibit a strong and positive relationship with active life consumption and a strong negative one with consumption drop at retirement: workers with high lifetime payroll earnings are more productive and are better prepared for retirement. Therefore, there is a large value for inter-worker redistribution that can be done with the pension system. 

The third necessary statistics is the marginal propensity to consume (MPC) at retirement from pension income. Given that workers are bad-prepared for retirement, consumption has a larger social value at retirement than in active periods. However, present-focused workers may consume pension benefit at active periods, which reduces the welfare gains of transfers at retirement. Using the shock to pension savings generated by the Global Financial Crisis, I find that retirement income is effective to raise retirement consumption with a MPC at retirement from pension income of around 0.8.

In the third part of the paper, I put together the estimated redistribution value with the estimated behavioral cost of reforms. That is, I solve the incentive distortion and income redistribution trade-off empirically.   

I find that even in the presence of these significant behavioral costs, it is socially desirable to increase the pension contribution rate and benefits' progressivity. The low preparedness of workers for retirement and its strong positive correlation with lifetime earnings makes the redistributional gains larger than the distortion costs. The lower bound for the social gains is \$0.11 for every \$1 mechanically transferred to retirement by an increase in pension contributions rate, and the social gains for every \$1 transferred by increasing benefits' progressivity are \$0.45. These social gains are large. To put them in context, an increase in pension contribution rate of 1\% will generate social gains equivalent to 0.08\% of the GDP, and an increase in 1\% of benefits' progressivity will generate social gains of 0.29\% of the GDP. The optimal system has a progressivity on pension benefits similar to that of Canada, Japan and Israel, where 42\% of public pension benefits go to workers with earnings beloew the median, and a pension contribution rate of 17\%, similar to that of Israel, Japan and Netherlands.

These results are in line with the findings of \cite{odea_design_pensions} and \cite{braun2017old}. Both find, using structural estimations, that is welfare improving to increase income floors at retirement, and thus to focus transfers on those in the lowest part of the elderly income distribution. 

The rest of the paper proceeds as follows. Section \ref{frame_work} presents the model and the social gains of reforms. Section \ref{context_data} presents the context of the Chilean pension system and the data used in the estimation. Section \ref{estimation} estimates the elements of the model, where Sections \ref{pension_benefit_estimation}-\ref{payroll_tax_estimation} estimate the behavioral response to reforms, Section \ref{earnings_consumption_section} estimates the relationship between consumption and lifetime earnings, and Section \ref{calibration_estimaion} calibrates the free parameters. Section \ref{results} presents the results. The final Section concludes.

\section{Conceptual Framework}

\label{frame_work}

This section presents the model. \\

\noindent \textbf{Setup.--} There are a continuum of workers indexed by $i$ that live for two periods ($t=\{1,2\}$). In the first period, workers actively generate taxable earnings ($z_i$) and in the second, they retire and receive a pension benefit. Worker $i$ chooses her consumption when active and retired $\{c_{i1}, c_{i2}\}$, her taxable income ($z_{i1}$), savings for retirement ($s_i$) and other unobservable decisions like other sources of income, captured by the reduced form variable $\chi_i$. From taxable earnings ($z_i$), the worker pays a pension contribution with at $\kappa$ and a linear income tax at rate $\tau$. Workers have heterogeneous productivity, access to savings instruments, and other sources of income. This heterogeneity is captured by the heterogenous utility cost that taxable earnings $z_i$, $s_i$ and $\chi_i$ generates. Workers may exhibit behavioral biases.   \\

\noindent \textbf{Pension system.--} At retirement, worker $i$ receives a pension $p_i$ that is comprised of two components: a self-funded pension $a_i$ that depends on her previous earnings $z_i$ and a lump-sum transfer $b$. The component $a_i$ is given by the worker pension contribution ($\kappa z_i$) compounded by the returns on her pension investment ($R$) minus a linear tax at rate $\phi$ (i.e., $a(z_i)=k z_i R(1-\phi)$). The lump-sum transfer $b$ is funded by the resources collected by the linear pension tax ($\int_i \phi z R k di$) plus another government pension spending financed by other public revenue sources $E$. Therefore, the lump-sum transfer $b$ is given by $\phi \overline{z}Rk+\overline{E}$, where the over-line indicates population average. \\

\noindent \textbf{Workers' problem.--} Combining the setup with the pension system, the workers' problem is to maximize:

$$\hat{U}_i(c_{i1}, c_{i2}, z_i, s_i, \chi+i)$$

subject to:
\begin{align*}
    &c_{i1} \in B_{i1}(z,\chi)= z (1-\kappa-\tau) + s+g_i(\chi) \\
    &c_{i2} \in B_{i2}(z,\chi)= a(z) +b + f(s) \\
    &a(Z)= (1-\phi) \kappa z R^P \\
    &b=\phi \kappa \overline{z} R^P +\overline{E}
\end{align*}
where $c_{i1}$ and $c_{i2}$ are consumption when active and retired, respectively; $z_i$ is active life taxable earnings; $s_i$ is pension savings; $\chi_i$ are other unobservable worker decision that affect her budget constraint; $\kappa$ is the pension contribution rate; $\kappa$ is linear tax on taxable earnings; $\phi$ is linear tax at retirement on the pension contribution; $a_i$ is self-funded pension; $b$ is the lump-sum transfer at retirement; $R$ is the return on pension contribution investment; and $E$ is other government spending on pensions. $g_i(\chi)$ and  $h_i(\chi)$ are arbitrary functions that capture the effect over the budget constraint that $\chi$ has. Figure \ref{graph_pension_system} shows the design of the pension system by showing the relationship between lifetime earnings and pension benefit at retirement. 

With this pension design, it is possible to control inter-worker and inter-temporal redistribution by modifying two parameters: $\phi$ and $\kappa$. On the one hand, the allocation of benefits at retirement are shaped, in the spirit of \cite{sheshinski1972optimal}, as a linear income tax on the pension contribution, where taxed income is given back as a lump-sum transfer to every worker. Thus, $\phi$ controls the benefits' progressivity with respect to active life earnings. For example, an increase in $\phi$ transfers retirement benefit from workers with high lifetime earnings to those with low. In the polar case when $\phi=1$, every worker receives the same pension, and if $\phi=0$ the pension contribution becomes a mandatory saving rate ---a defined contribution plan--. Overall, reforming to $\phi$ has two effects on the pension design: it changes the lump-sum transfer, and it reduces the slope of the relationship between earnings and benefits (i.e. the benefit-earnings link). The effects of an increase in $\phi$ on retirement benefits are illustrated in Figure \ref{graph_reform} panel (a). 

On the other hand, $\kappa$ controls the inter-temporal redistribution. An increase in the pension contribution rate ($\kappa$) transfers income from the active to the retired period. This has three effects: it reduces the net-of-tax rate of taxable earnings while the worker is active, it increases the lump-sum transfer at retirement, and it increases the slope of the relationship between lifetime earnings and benefit at retirement (benefit-earnings link). Figure \ref{graph_reform} panel (b) illustrates the effects of an increase in $\kappa$ on retirement benefits. \\  

\noindent\textbf{Government problem.--} The government maximizes a utilitarian social welfare function with respect to the two parameters of the pension system:

$$W=\int_i U_i(\kappa,\phi) di$$
 
where  $U_i(\cdot)$ can differ from workers self-perceived utility $\hat{U}_i(\cdot)$. This captures, in a paternalistic way, the bias that workers may exhibit in their choices. 

I further assume that any effect of the reforms on public revenue through the income tax ($\tau$) is transfer back to workers through a reduction (or an increase) of $\overline{E}$. Therefore, by construction, reforms to $\kappa$ and $\phi$ are intra-cohort budget balanced through the lump-sum transfers at retirement ($b$). This allows to isolate the analysis from the extensive literature that analyzes inter-cohort transfers and the fiscal burden of public pension systems, and it focuses on the optimal allocation of resources inside a cohort, across workers and time. Nonetheless, my results are informative on, for example, the optimal way of reducing a  pension system fiscal burden. 

I make two assumptions that simplify government first order conditions: 

\begin{assumption}[\textbf{Separable preferences}]
Preferences are separable:
$\frac{\partial^2 U}{\partial k \partial l}=0 \ \ for \ k,l \in \{c_1,c_2,z,\chi\} \ and \  k\neq l$
\end{assumption}

This assumption is standard in the social insurance literature (see \cite{landais2021value} for a recent example). 

\begin{assumption}
Workers exhibit present-focused bias. While they are active, they discount retired period consumption by the factor $\delta$: $$\frac{\partial \widehat{U}}{\partial c_2}= \beta \frac{\partial U}{\partial c_2}$$
where $\widehat{U}$ is worker perceived preferences and $U$ is the central planner workers' preference used in the social welfare function.
\end{assumption}

Under this assumption workers may disagree with central planner but in in a specific way: different valuation of marginal utility of consumption at retirement (when $\delta \neq 1$). The consequences of this bias are equivalent to hyperbolic discounting, and the central planner approach to it is paternalistic. 

This behavioral bias generates two distortions which have an impact on the welfare effects of reforms. First, given that part of workers taxable earnings are given back to them in retirement and they do not value that future consumption correctly, they do not generate taxable earnings optimally while active. Therefore, when a policy reform discourages taxable earnings, there is an internality on worker's well-being. The literature calls this internality as bias correction effect. 

Second, given a profile of consumption ($c_{i1},c_{i2}$), a worker's  inter-temporal marginal rate of substitution is different than for the central planner. Therefore, worker's allocation of extra retirement income on consumption when  active and retired may be different than that of the central planner. In other words, a worker may unsave when active in order to consume the additional future pension benefit even if she is under-consuming, from the perspective of the central planner, when she is retired. This inter-temporal consumption allocation depends on worker's preferences and her access to saving instruments. For example, a completely myopic worker ($\beta=0$) with perfect access to capital markets will consume when active the whole retirement income by borrowing against it. Thus, in this case, an increase in retirement income will have no effect on retirement consumption. Under present-focused workers, not only matters \emph{when the transfer is done}, but also \emph{when the worker consumes it}, i.e., the marginal propensity to consume out of retirement income.

I make two definitions that simplify the notation of the problem: 

\begin{definition}[Retirement preparedness] 
\label{def_retirement_prepa}
    Given a consumption pair ($c_{i1}, c_{i2})$, let $d_i=U_{c_{i1}}(c_{i1})-U_{c_{i2}}(c_{i2})R$ be the distance to the Euler equation under the central planner's preferences and given the return to pension savings.
\end{definition}

\begin{definition}[Retirement MPC]
\label{def_mpc}
Let $\mu$ be the marginal propensity to consume out of retirement disposable income: 

$$\mu=\frac{\partial c_{i2}}{\partial y_{i2}}$$
where $y_{i2}$ is disposable income at retirement.
\end{definition}

Finally, I make three assumptions driven by data limitations used to estimate the behavioral responses: 

\begin{assumption}    
   \noindent (i) Behavioral responses to reforms-- elasticities and marginal propensity to consume--are constant across income and the value of the pension system parameters.

 \noindent   (ii) Income sources other than taxable earnings do not respond to reforms on the pension system parameters.

 \noindent   (iii) Workers' marginal saving return rate is the same than the return of pension funds $R$. 
\end{assumption}

The first two assumptions are driven by the nature of the variation used to identify behavioral responses. I either, do not have the sample power or the variation is local to a subset of the population and to the current design of the pension system. The third assumption simplifies the expression for social gains of reforms and provides a lower bound for the social gains of reforms if the pension system is more efficient in investing savings than individual workers. However, the expression can be extended to incorporate heterogeneous return rate on savings returns. 

With these assumptions and definitions, Proposition 1 states the social gains of reforming the pension system. 

\begin{proposition}
The welfare effects of a marginal reform to the benefits progressivity ($\phi$) is given by:

\begin{align}
    \frac{dW}{d\phi}=-&R\kappa\left(\mu Cov\left[d_i,z_i\right]+Cov\left[U_{c1}(c_{i1}),z_i\right]\right)  +(\tau +\phi\kappa)R \overline{z}[-\varepsilon_{m\overline{z}}+\varepsilon_{b\overline{z}}] \int_i (\mu d_i+U_{c1}(c_{i1})) di \notag  \\ &+   \mu(1-\beta)(1-\phi)\kappa \int_i(\mu d_i+U_{c1}(c_{i1}))[-\varepsilon_{mz_i}+\varepsilon_{bz_i}] z_i  di   \label{foc_progressivity}
\end{align}
and the social gains of a marginal reform to the pension contribution rate ($\kappa$) is given by:

\begin{align}
    \frac{dW}{d\kappa}= &  \mu Cov\left[d_i,zi \right]+  \mu \overline{d}\overline{z}-R\kappa\left( \mu Cov\left[d_i,z_i\right]+Cov\left[U_{c1}(c_{i1}),z_i\right]\right) \notag \\
     & +(\tau +\phi\kappa)R \overline{z} \left[\frac{1-\kappa}{1-\tau}\varepsilon_{(1-\tau)\overline{z}}-\frac{1-\phi}{\phi}\varepsilon_{m\overline{z}}+\varepsilon_{b\overline{z}}\right] \int_i (\mu d_i+U_{c1}(c_{i1})) di \notag  \\
    &+\mu(1-\beta)(1-\phi)\kappa  \int_i \left[\frac{1-\kappa}{1-\tau}\varepsilon_{(1-\tau)z_i}-\frac{1-\phi}{\phi}\varepsilon_{mz_i}+\varepsilon_{bz_i}\right]( d_i+U_{c1}(c_{i1}))  z_i  di \label{foc_contribution_rate}
\end{align}
where overline means the population average.   
\end{proposition}
\begin{proof}
In the appendix. 
\end{proof}

This proposition defines the social gains of marginal reforms to the pension system. The expressions are simple and can be interpreted as the trade-off between the welfare gains of income redistribution, across time and workers, and the welfare cost of behavioral distortions. 

The pension systems generates redistributional gains from two sources. First, it provides social insurance against income drop in retirement. The value of this social insurance is larger for workers more exposed to the retirement risk, i.e., those that are farther away from the Euler equation ($d_i$ from Definition \ref{def_retirement_prepa}). Then, the value of the social insurance is captured by $Cov[d_i,z_i]$. However, as discussed, the income transfer at retirement only has social insurance value if it is consumed at retirement. Thus, the social insurance term is multiplied by the MPC out of retirement income ($\mu$ from Definition \ref{def_mpc}). The second source of redistributional gains is the redistribution of income across workers with different productivities, which is captured by $Cov[U_{c_1}(c_{i1}), z_i]$. In sum, the pension systems serves double rol; social insurance and income redistribution. 

The welfare cost of the distortion generated by the reforms are driven by the response of taxable earnings. As shown in Figure \ref{graph_reform}, the response of taxable earnings are a function of how taxable earnings respond to future benefit payment ($\varepsilon_{bz}$), to benefits-earnings link ($\varepsilon_{mz}$), and to  net-of-tax rate of active life taxable earnings ($\varepsilon_{\tau z}$). The response of taxable earnings generates welfare cost through two margins. First, it reduces government revenue which, in turn, reduces the lump-sum transfer at retirement ($b$). This has a welfare cost, given by welfare value that the lump-sum transfer has. This is the second term on equation (\ref{foc_progressivity}) and fourth term in equation (\ref{foc_contribution_rate}), and is known as the fiscal externality (\cite{kleven2021sufficient}). Second, the response of taxable earnings to reforms moves workers further away from their optimal choice of taxable earnings. This has a welfare cost because workers are present-focused bias and do not fully-internalize the effect that taxable earnings generates in their retirement consumption. This is the last term of each equation and is known as bias correction (\cite{farhi_behavioral_earnings}). 

These expression--(\ref{foc_progressivity}) and (\ref{foc_contribution_rate})-- are generalizations of the optimal  social insurance formula (\cite{baily1978some} and \cite{chetty2006general}), and of the optimal linear-income tax formula derived by the extensive literature that analyzes this type of taxes (\cite{sheshinski1972optimal}, \cite{atkinson1995public}, \cite{itsumi1974distributional}, \cite{stern1976specification}, \cite{dixit1977some}, \cite{helpman1978optimal}, \cite{deaton1983explicit}, and \cite{tuomala1985simplified}, among others). The extensions are to incorporate  heterogeneity on workers' productivity, ability to prepare for retirement, and present-focused bias. In Appendix \ref{appendix_B}, I show that by shutting down these extensions, I  obtain the literature's standard formulas.

I then make two assumptions about consumption preferences. These assumptions simplify the relationship between consumption and marginal utility of consumption. 

\begin{assumption}[\textbf{State dependence}] 
Workers preferences are time separable, with time discount factor of $\beta^{-1}=R^P$, and state-dependent preferences with respect to retirement such that:
$$\frac{\partial U(c)}{\partial c_2}=\beta \theta \frac{\partial U(c)}{\partial c_1}$$
for any consumption $c>0$.
\end{assumption}

\begin{assumption}[\textbf{CRRA}]
Preferences for consumption are CRRA with relative risk aversion parameter $\gamma$. 
\end{assumption}

These two assumptions are standard in the social insurance literature. With them, I can approximate the social value of the marginal income transfers as simple functions of active and retired consumption. 

\begin{lemma}
The distance to the Euler equation and the marginal utility of active consumption 
 are given by: 
\begin{align}
    &d_i (c_{i1}, c_{i2}) \approx U_{c_{1}} (c_{i1})\left[(1-\theta)+\theta \gamma\frac{c_{i2}-c_{i1}}{c_{i2}} \right]  \\
    &U_{c_1}(c_{i1})\approx U_{c_1}(\overline{c})\left[1-\gamma\frac{c_1-\overline{c}}{\overline{c}} \right]
\end{align}
where $\overline{c}$ is the certain equivalent active consumption,i.e., $\overline{c}$ is such that $\int_i U_{c_1}(c_{i1})di=U_{c_1}(\overline{c})$.
\end{lemma}
\begin{proof}
    In the appendix.
\end{proof}

Lemma 1 states that the distance to the Euler equation ($d_i$ of definition 1) can be approximated as a linear function of the consumption drop in retirement. Equivalently, marginal value of consumption of a worker when active can be approximated as a linear function of the ratio between active life consumption of that workers and the certain equivalent of the population. The linear functions are useful to plug them in the co-variances that drive the redistributional gains of reforms ($Cov[d_i,z_i]$ and $Cov[U_{c_1}(c_{i1}),z_i])$, helping to shed light on what forces are driven the social gains of reforms. However, the approximation error is larger for some observations so I use the exact definition for the estimation of social gains of reforms in Section \ref{results}. 

The welfare effect of reforms to the pension contribution rate ($\kappa$) and benefits progressivity ($\phi$) can be computed as function of a few parameters. I divided the parameters in two sets: (i) those that I estimate empirically $\left(\varepsilon_{\tau z}, \varepsilon_{m z}, \varepsilon_{b z}, \mu, Cov[d_i,z_i], Cov[U_{c_1}(c_{i1})], \overline{d}, \overline{z}, R\right)$ in Section \ref{estimation}, and (ii) those that I calibrate using literature estimates ($\gamma, \theta, \beta$). 

\section{Context and Data}

\label{context_data}

\subsection{Context: Chilean Pension System}

The Chilean pension system consist of a mandatory defined (DC) contribution plan that is complemented by a government subsidy. In the mandatory DC part, workers are required to save 10\% of their payroll income, until a cap, in a personal retirement account which is invested until retirement. The personal retirement account is invested by a Pension Fund Administrator (PFA) selected by the worker. The PFAs are for profit and privately owned firms, authorized by the government to be pension fund administrators. In the relevant period of the analysis there were six PFAs. In return for the investment service, each PFA charges their affiliated workers a management fee, which is deducted from payroll income every month. The managment fee rates are set independently by each PFA.

Workers can choose between administrators. Since the inception of the system in 1981 until 2008, workers had to choose a PFA at the moment of receiving their first income. At the moment of election, workers were shown information about each PFA current fee and past returns. After the initial election, workers can actively switch across PFAs by making in-person visits to both the old and new PFA offices. In practice, switches between PFAs have been not common.\footnote{Between 2009-2019 less than 10\% of the workers switched between PFAs} This stickiness to the initially chosen PFAs motivated the government to introduce an auction of new workers from 2008. In this auction, conducted every two years. the PFA with lowest management fee rate receives all the workers earning their first paycheck income in the following two years.

Workers' pension savings are held in personal accounts, invested in five funds (A, B, C, D, E). The investment strategy of each funds is defined by the PFA, but with many restrictions on the portfolios. In particular, different funds have different limits on the share of stocks and bonds that can be held, thus regulating their riskiness. Fund A is the riskiest one and fund E the safest one, and the other ones are in between. By default, the allocation of worker savings between funds is done depending worker's age, going from a more risky to a safer portfolio as the worker ages.  The worker's account is invested in fund B until the worker turns 36 years old, when a 20\% of his account is moved to fund C. Every year thereafter, another 20\% is moved until the total of his account is in fund C when he turns 40 years old. Similarly, when the worker is 10 years from the retirement (50 for females and 55 for males), his account is move from fund C to fund D, in the same fashion with 20\% increments every year. After that, personal account is invested in fund D for the rest of the worker's life. Instead of the default investment allocation, the worker can actively choose investment funds for his savings, but this was not a widespread practice in the past. For example, less than 5\% of the workers actively choose funds between 2007-2014.    

Workers can choose to retire after they turn 65 years old for males, and 60 for females.\footnote{Or sooner if their savings are sufficient to to fund a pension with a replacement rate of 70\% of their average income in the last 10 years.} Here, the retirement concept refers to be able to receive an income flow from the personal retirement account (self-funded pension), but the worker can keep working with no penalties. At retirement, workers convert the stock saved in their retirement account in a flow by buying an annuity to an insurance company or by keeping the savings in the PFA and making withdrawals from it under conditions defined by the law. 

In addition to the self-funded pension, there is a government subsidy that acts as a complement to it financed from general government revenue and it was introduced in 2008. The subsidy is means-tested, aimed at the 60\% poorer and phases out with the self-funded pension. There are two parameters that define the subsidy: the minimum pension (PBS) and the largest pension with subsidy (PMAS). The PBS defines a minimal pension that someone with \$0 saving receives, while the PMAS defines the self-funded pension value at which the subsidy becomes \$0. The ratio at which the subsidy phases-out with self-funded pension is defined by the ratio between these two parameters, and it creates an implicit tax on pension savings. For every \$1 that a workers saves, he looses $\$\frac{PBS}{PMAS}$ of subsidy at retirement. Overall, the subsidy is an important part of the system: 68\% of the elderly receive it and, for recipients, the subsidy represents 62\% of their total pension payment.

\subsection{Data}

I use administrative data from the Chilean pension system. In this data I can observe, for the whole population since 2002, each worker's monthly payroll income, monthly pension savings account, voluntary contributions, voluntary savings account, pension fund administrator, the investment fund, demographics, retirement date, and pension subsidy received. 

A sub-sample of this administrative data is matched to a panel survey (Encuesta de Prevision Social) that asks about family composition, income, wealth, consumption, education, and other demographics. Several rounds of this survey were held in 2002, 2004, 2006, 2009, 2012, 2015, and 2020.   

In my analysis I focus on the period between 2004-2019.  This period encompasses the variation that used in my empirical identification of behavioral responses to changes in the pension system. Table \ref{sum_admin} shows the descriptive statistics of pension administrative data for this period. Table \ref{sum_survey} shows the summary statistics of the EPS survey variables used in the analysis.

\section{Estimation}

\label{estimation}

In this section I estimate the elements that govern the welfare effects of reforms. I start with the elements that govern the behavioral response to reforms: three elasticities (taxable earning response to future pension benefit, benefit-contribution link, and payroll tax) and the marginal propensity to consume from retirement income. I then compute the moments that govern the redistributive value of reforms that are the relationship between lifetime earnings and active-life consumption, and the relationship between lifetime earnings and consumption drop at retirement. In the last section I calibrate the relative risk aversion, consumption at retirement preference, hyperbolic discounting, and pension system return.

\subsection{Taxable earnings elasticity to future pension benefit}

\label{pension_benefit_estimation}

To estimate the elasticity of pre-retirement income to future pension payment, I use in the heterogeneous exposure to market returns that workers faced during the Global financial crisis, which affected their future pension benefit. 

Pension savings are invested in capital markets, thus, exposed to idiosyncratic returns. The savings are invested in three relevant funds (B, C, and D), defined by the share of the fund invested in variable income securities: 60\% of fund B, 40\% of fund C and 20\% of fund D are invested in variable income securities, while the rest is invested in fixed income securities. The aim of this regulation on portfolio is to have investment alternatives with different levels of risk.  

By default, workers savings are allocated in these funds depending on his age. At the beginning of a worker active life, all his savings are invested in fund B and then, they are switched at specific points to the less risky funds. Specifically, the month the worker turns 36, 20\% of his savings are switched from fund B to fund C, and another 20\% each year after, until all the savings are in fund C when the worker turns 40. Similarly, the month the worker turns 10 years before retirement age (50 for females and 55 for males), 20\% is switched form fund C to fund D, and 20\% each year after. This design of savings allocation has two useful characteristics. First, it is discrete. Second, the switch happens in the month when the worker turns certain ages. Therefore, the investment allocation and risk exposure in a specific time depends on the month of birth of the worker. In practice, this age-dependent investment strategy is binding in the data. Figure \ref{fund_allocation} shows the empirical allocation of workers savings across the funds at each age (defined in months). 

I use this heterogeneous exposure to investment risk jointly with the Global Financial Crisis (GFC) shock to security returns to build an instrument for future pension payments that depend on date of birth. Then, I use the instrumented future pension payment in a second stage, with pre-retirement taxable earnings as the dependent variable.

The instrument is the money-metric return obtained during the GFC if the worker $i$ followed the default age-dependent investment strategy:  

\begin{equation} \label{gfc_return}
    \rho_{i}=\prod_{t=B}^{E} (\alpha_{it}^C(age_{it})\cdot R_{t}^C+\alpha_{it}^D(age_{it}) \cdot R_{t}^D)*S_{i}^{pre}
\end{equation}

where $t$ is time in months; $B$ and $E$ mark the beginning and end of the Global Financial Crisis (GFC), respectively; $\alpha_{it}^j$ is the share of worker $i$ savings in fund $j$ defined by his age at time $t$; $R_{t}^j$ is the average return of the system for fund $j$ at time $t$, and $S_{i}^{pre}$ are worker $i$ pension savings on the month before the start of the GFC. Therefore, $\rho_{i}$ captures the return in dollars obtained during the GFC by the savings accumulated before the crisis. 

For the beginning ($B$) and ending ($E$)  points I use two alternative definitions, one long and one short, commonly used in the literature. The ending of the long one is when the US labor supply recovered, while the ending of the short is when the S\&P500 returned to pre-crisis levels. Figure \ref{pen_sav_shock} shows $\rho_i/S_{i}^{pre}$ during the financial crisis for workers with different born date, and Figure \ref{last10y} shows the return during workers' last 10 years by cohort for those turning the retirement age between 2009 and 2020. In these two figures, one can see the large and heterogeneous effect that the Global Financial Crisis had on the pension savings of workers born a few months apart. 

As robustness check, I use two sub-samples of workers: one wide and one narrow. The wide sample is given by the cohorts that turn the retirement age between 2008 and 2020. These are the cohorts that made switches between funds during the GFC, and therefore had heterogeneous exposure to it.   The narrow sample is given by those that turn the retirement age between 2016 and 2019. These cohorts, given their dynamic of fund switching during the GFC, are the cohorts that had the largest heterogeneity on returns during the GFC. Figure \ref{last10y} shows the return on the last 10 years of active life for both samples.

Before doing the two-stage estimation, I first provide evidence of the causal interpretation of my estimation by showing the relationship between taxable earnings and return during the GFC, i.e., a reduced form estimation. I use two samples: a treated sample and a placebo one. The treated sample is given by the narrow sample, those that turn the retirement age between 2016-2019 and had the largest heterogeneity on returns during the GFC. The placebo sample is given by those cohorts that will turn the retirement age between 2021-2024. These cohorts were young enough during the GFC so that their pension savings were invested in the same way, there is no heterogeinity on their return. 

I split the cohorts of each sample in two, high and low exposure to GFC shock. For the treated sample, I split them by the return they obtained during the GFC: those below the median of return are high-exposed and those above are low-exposed. By construction, every cohort of the placebo sample obtained the same return during the GFC. Thus, I split them in two by the order of cohorts of the treated sample. I order them from younger to older and assign them to high and low exposure by the equivalent age rank of the treated sample.  

Figure \ref{gfc_earnings_reducedform} shows the annual taxable earnings relative to 2005, in the left axis, and the average differential in return of pension savings since 2005, in the right axis, for high and low-exposed cohorts. Taxable earnings are controlled by age fixed effect and gender. In Panel (a), I show it for the treated sample, and in Panel (b) for the placebo sample. We can see, first, that there are significant heterogeneity in exposure to the GFC shock for the treated sample, with an average difference in returns of around 12.5\%, while there is no heterogeneity in return among the cohorts of the placebo sample. Second, more negatively exposed cohorts – that had a larger shock to their pension savings—significantly increased their pre-retirement earnings after the GFC, while there was no difference before the GFC. Also, there is no significant difference in taxable earnings before or after the GFC for cohorts of the placebo sample. This provides evidence that the difference in taxable earnings is not driven by the relationship between GFC exposure and cohort age rank among the sample.

I then estimate the structural equation. I use $\rho_i$ in a first stage as an instrument for pension payment at retirement. Specifically:
\begin{equation} \label{first_stage_eq}
    log(p_{i})=X_{i}\beta_1+\beta_2 log(\rho_{i})+\epsilon_{i}
\end{equation}
where $p_{i}$ is pension payment received by worker $i$, $X_{i}$ are controls (age fixed-effects, gender and pre-GFC pension savings), and $\rho_i$ is the returned obtained during the GFC by following the default age-dependent investment strategy, defined in equation (\ref{gfc_return}). Given that $\rho_i$ is a stock and $p$ is a flow, I normalize $\rho_i$ by the cost of a retirement annuity in 2020. Thus, $\beta_2$ can be interpreted as the percentage effect that a percentage change of the self-funded pension generated by the GFC has on on the pension benefit.

Table \ref{first_stage_FC_shock} shows the estimation of the first-stage (equation  (\ref{first_stage_eq})). The point estimate is 0.54 and robust to controls, the definition of  the sample and GFC duration. The reason for this coefficient to be smaller than 1 is the existence of the government subsidy. The first stage is strong, with a t-statistic of the instrument above 100. 

I then use the instrumented future pension payment ($\widehat{p_{i}}$) in a second stage, where the dependent variable is taxable earnings for the months after the GFC and before retirement age:
\begin{equation} \label{second_stage_eq}
    log(y_{it})=\alpha_t+X_{it}\beta_1+\gamma \widehat{log(p_{ia})}+u_{i}
\end{equation}
where $\alpha_t$ is a time fixed effect, and $X_{i}$ are controls consisting of age, gender, and pre-financial crisis pension savings. Table \ref{pen_shock} shows the estimation of the second-stage (equation (\ref{second_stage_eq})). 

I find an elasticity of taxable earnings to future pension payment of 0.11. The estimation is precise and robust to the definitions of the sample and the duration of the GFC. 

One concern about my identification strategy is that workers are allowed to switch among funds. Three arguments against this. Figure  \ref{fund_allocation}. Second, less than 6\% of work switched across funds between 2008-2014. This can be seen in figure  sample in the appendix.  I show that the fraction of workers that voluntarily switched their funds is virtually zero. Third, my instrument first stage is strong. 

I provide further evidence on the causal interpretation of my estimation by leveraging on the panel nature of my data. I do an event study analysis by generating placebo treatments before the GFC: 

\begin{equation} \label{dynamics_pen_sav_shock}
    log(y_{it})=\alpha_t+X_{it}\phi+ \sum_{k=B-K}^{B} \beta_t D_{iat}^k+\sum_{k=E+1}^{E+K} \beta_t D_{iat}^k+\epsilon_{it}
\end{equation} 
where $D_{i}^k$ is 0 except for the month $k$ when it takes the value of $\rho_i$. Estimates are normalized with respect to $\beta_{k=B}$. Figure \ref{dynamicsresponse} shows the coefficient estimates ($ \beta_t$) for the long and the short definition of the GFC. The heterogeneity on returns has no effect on taxable earnings before the GFC, and has a consistent effect after. This provides additional support for the causal interpretation of the estimates.

\subsection{Marginal propensity to consume from retirement income}

\label{mpc_estimation}

In Section \ref{frame_work}, I showed that the marginal propensity to consume from pension benefits is a necessary statistic to estimate the welfare effects of reforms to the pension system.

To identify the response of retirement consumption to pension benefits, I use the shock to pension savings, and thus to pension benefits, that the GFC generated. This is the same identification strategy than the previous section, but with a first-stage in levels rather than in logs and with retirement non-durable consumption as dependent variable in the second stage. 

In the second stage I use non-durable consumption in retirement as dependent variable. This variable comes from the EPS survey. Therefore, the sample is composed of those retired workers surveyed at least one time in the rounds of  2009, 2012, 2015, and 2020, and who turn retirement age between 2008-2020 (wide sample in the previous section). As robustness check, I also estimate the regression with the short sample. The first-stage is given by the following equation:

\begin{equation} \label{first_stage_eq_MPC}
    p_{i}=X_{i}\beta_1+\beta_2 \rho_{i}+\epsilon_{i}
\end{equation}

and second one is given by the following equation:
\begin{equation} \label{second_stage_eq_mpc}
    c_{it}=\alpha_t +X_{it}\phi+ \gamma \widehat{p}_i +\epsilon_{i}
\end{equation} 
where $c_{it}$ is non-durable consumption at retirement by individual $i$ in year $t$; $\alpha_t$ is year fixed effect; $X_{it}$ are controls that include age, gender, household composition, and pre-GFC pension savings; and $\widehat{p}_i$ is given by the first-stage from equation \ref{first_stage_eq_MPC}. Again, $\rho_{i}$ is normalized by the cost of an annuity, so $\beta_2$ can be interpreted as the effect that \$1 of pension savings earned during by the GFC has on future pension benefit.

The first stage is reported in table \ref{first_stage_FC_shock}. The instrument is strong with a t-statistic of the instrument of above 100. The point estimate $\beta_2$ is smaller than 1, i.e. \$1 extra of pension savings generate less than \$1 extra of benefits, because the subsidy.

Table \ref{eld_disp_income_pension} shows the results of the 2SLS estimation. Estimates are robust to the specification, with an MPC from pension benefit of 0.78. Overall, this finding supports the fact that income transfers at retirement increase retirement consumption, thus, active workers do not unsave to consume the future pension benefits.

\subsection{Taxable earnings elasticity to the benefits-earnings link}

\label{benefit_contribution_estimation}

In this Section I use the introduction of the pension subsidy, that reduced the relationship between earnings and benefits, to estimate the elasticity of pre-retirement earnings to the benefits-earnings link.

In July 2008, the government introduced a subsidy to complement the self-funded pension financed through the mandatory defied contribution plan. As described in Section \ref{context_data}, this subsidy is means-tested for the 60\% poorer of the population, and its design is defined by two parameters. The first one is the basic amount that someone with \$0 self-funded pension receives, which is called PBS (Basic Pension). The second parameter is the defined by the PMAS (largest pension with subsidy). By design, the PBS phases out with self-funded pension until it becomes \$0 at the PMAS, after which it stays in \$0.  Figure \ref{chilean_design} shows the subsisdy design. 

The subsidy introduction had two effects on future recipients. First, it increased their future pension. Second, it reduced the benefits-earnings link. For every \$1 that a future recipient saved in his personal account, his future pension increased in only \$0.66 because he lost future subsidy. For a worker to be a future recipient, two conditions has to be met: she belongs to the 60\% poorer of the population, and her personal savings at retirement finance an annuity (self-funded pension) below the PMAS. 

I use the subsidy design and the timing of it introduction to identify the effect that the benefits-earnings link has on pre-retirement taxable earnings. To do so, I use two characteristics. First, the savings accumulated until the subsidy introduction are not affected by the design of the subsidy. This savings are exogenous on the particular definition of the PMAS by the reform that introduced the subsidy. Second, the subsidy introduction changed the implicit tax rate on additional pension savings on an heterogeneous way depending on pre-subsidy savings. 

Figure \ref{het_change_implicittax} shows the change of the implicit tax for workers with different pre-subsidy savings, assuming that everyone meets the means test, i.e., belongs to the 60\% poorer. We can see that there are three groups. First, workers with pre-subsidy pension savings low enough that no matter how much they earn after the subsidy introduction they will end up being receivers at retirement. For this group, the subsidy introduction increased the average tax on their additional pension savings in 33\%. I call the upper limit of pre-subsidy savings for this group as $\underline{a}$. Second, workers that have pre-subsidy savings already above PMAS. No matter what they do, they will not be future recipients because they have already save too much. For this group, the subsidy did not change the average tax on additional pension savings. I call the lower limit of pre-subsidy savings for this group as $\overline{a}$. Finally, we have the intermediate group (between $\underline{a}$ and $\overline{a}$) that is composed by workers that if they generate enough earnings after the subsidy introduction, they can be no recipients. For this group, the subsidy introduction changed the average tax in an amount between 33\% and 0\%, depending on their pre-subsidy savings.  

I use this heterogeneous effect of the subisdy introduction on addtional  pension savings average tax across to identify the effect of the benefits-earnings link on taxable earnings. I first defined for every worker the pre-subsidy savings that defined the limits of never recipient ($\overline{a}$) and always recipient ($\underline{a}$). These limits are defined as:
\begin{align*}
    &\overline{a}_{i}= \prod_{t^0}^{T^i} PMAS \cdot P_{i} \cdot (1+r_{it})^{-1} \\
    &\underline{a}_{i}= \prod_{t^0}^{T^i} (PMAS-0.1\cdot z^{max}_t) \cdot P_{i} \cdot (1+r_{it})^{-1}
\end{align*}
where $i$ indicates worker; $t^0$ the date of the subsidy introduction (July 2008); $T^i$ month at which worker $i$ turns retirement age; $P_i$ is the price of an annuity for worker $i$, which depends on the month at which she turns the retirement age, gender and dependents; $r_{it}$ is the return of worker $i$ pension savings investment during month $t$; and $0.1 \cdot z^{max}_t$ is the limit on pension contribution during month $t$. Note that these limits are worker specific because depends on her age at the subsidy introduction, gender, dependent, and the specific return to her investment. Nonetheless, these worker specific characteristics are not in her control.  

Then, I create a dummy variable ($I^a_i$) to take the value of 1 if a worker is always a recipient (i.e $a_i^{pre}<\overline{a}$) and 0 if it is never a receiver ($a_i^{pre}<\overline{a}$), given her pre-subsidy ($a_i^{pre}$). Figure \ref{presubsidy_above_below} shows the average $I^a_i$ for bins of normalized pre-subsidy savings around the limits. The normalization is given by $m_i=(a_i-\underline{a_i})/(\overline{a}_i-\underline{a}_i)$, which takes  in account that the upper and lower bound are worker specific. If $m_i$ is less than one, then she is in the always below region, and if it is above 1 she is in the always above region. We can see that 58\% of those below the lower limit are recipient at retirement, while none of those above the upper limit are. 

Given that a worker is a recipient only if he belongs to the 60\% poorer of the population, I use this variable in a first stage as an instrument for being a receiver after retirement ($r_{i} \in \{0,1\}$):
\begin{equation} \label{first_receiver}
    r_i= \alpha_0 + \alpha_1 I^a_i + \alpha_2 X_{i}+ e_{i}
\end{equation}
where $r_i$ is a dummy that takes value 1 if worker $i$ received subsidy at retirement, $I^a_i$ is a dummy that takes value 1 if workers pre-subsidy savings were below the lower threshold, and $X_{i}$ are worker $i$ specifics controls, like gender and cohort fixed effects.

Table \ref{first_stage_receiver} shows the estimation of this first-stage. Workers below the lower bound are 58\% more likely of receiving the subsidy at retirement than workers above the upper bound. The first stage is strong, with a t-statistic above 90, and robust. This strong first stage can be seen in Figure \ref{presubsidy_above_below}, where we can see that a stable share of 58\% of workers below $\underline{a}_i$ are recipient, while none of those above $\underline{a}_i$ are.

I use the first-stage defined in equation (\ref{first_receiver}) in a second stage following a differences-in-differences design:
\begin{equation}\label{second_receiver}
    log(z_{it})=\alpha_0 + \alpha_t+\alpha_i+ \gamma Post \cdot \widehat{r}_i + \epsilon_{it}
\end{equation}
where $z_{it}$ is taxable earnings of worker $i$ in year $t$, where years are defined starting from July; $\alpha_t$ are time and worker fixed effects; $Post$ is a dummy that takes the value 1 for periods after the subsidy introduction (July 2008); and $\widehat{r}_{it}$ is the instrumented probability of being a receiver of the subsidy.

Table \ref{second_receiver_table} shows the estimation of the second-stage using a Two-Stage Least Squares. I find that recipients reduced taxable earnings in 1.2\%. This estimate is precise; and robust to controls and the bandwidth around the limits (lower and upper) used to the estimation.  

Crucially, my identification strategy relies in the assumption that taxable earnings of workers above and below the limits would have evolved similarly in the absence of the subsidy introduction. I do placebo analysis that support this assumption. First, I test the time difference generated by the timing of the subsidy introduction. I run an event study following the following equation:

\begin{equation} \label{dynamics_subsidy_eq}
    log(z_{it})=\alpha_t+\alpha_i+\sum_{k=I-K}^{B} \beta_t D_{iat}^k+\sum_{k=I+1}^{I+K} \beta_t D_{iat}^k+\epsilon_{it}
\end{equation} 
where I is the subsidy introduction date (July 2008), $D_{it}^k$ is zero except in year $k$ when takes the value of $\widehat{r}_{i}$.

Figure \ref{dynamic_subsidy_intro} plots the coefficients of the event study for 5 years before and after the subsidy introduction (i.e., $K=5$).The parameters are relative to that of the year before the subsidy introduction. We can see that the treatment has no effect before the subsidy introduction. This is evidence that treated and control sample did not have different dynamics before the subsidy introduction. 

The second threat to my identification strategy comes from the fact that I am assign treatment based on earnings before subsidy introduction. It can be that workers with lower earnings before the subsidy introduction would evolve differently, for example, because of particular life-cycle dynamics. I provide evidence against that this drives my results by running a placebo on the definition of the thresholds. Instead of using the PMAS defined by the reform, I use two alternative placebo PMAS; one 20\% larger ($PMAS_{L}^{P}=1.2 PMAS$) and one 20\% smaller ($PMAS_{L}^{P}=0.8 PMAS$). With these two placebo PMAS y build the placebo limits $\underline{a}, \overline{b}$ for pre-subsidy savings. In Figure \ref{presubsidy_above_below_placebo} we can see that the placebo limits do not generate a first stage. I show this formally in Table \ref{first_stage_receiver_placebo}, column (1) for the smaller placebo PMAS  and column (2) for the larger one, where the coefficient of the instrument is a precise 0. 

Given the lack of first stage, in order to run the second stage I build an instrumented recipient variable by using the first stage  the estimation of the real first stage (with the real PMAS). With these, those below the limits are, artificially, assign to have 60\% larger probability of being recipient at retirement. I use this instrumented placebo recipient variable in the second stage of the equation (\ref{first_receiver}) for both placebo limits. Table \ref{first_stage_receiver_placebo} shows this second stage, column (1) for the smaller placebo PMAS and column (2) for the larger one, where we can see that there is no effect on taxable earnings. That is, there is no difference in the dynamics of taxable earnings between high and low earners in the absence of the subsidy's treatment. This is true for sample earners below and above the limits defined by the subsidy design. 

So far, I provided evidence that the subsidy introduction causally affect pre-retirement taxable earnings. However, this effect is driven by the increase in the future benefit and by the reduction in the benefits-earnings link. I am interested in latter, thus I disentangle them. To do so, I first use the pre-subsidy savings to instrument the subsidy amount at retirement for recipients:
\begin{equation} \label{first_subisdy}
    s_i= \alpha_0 + \alpha_1 a_i^{pre} + \alpha_2 X_{i}+ \epsilon_{it}
\end{equation} 
where $s_i$ is the share of total pension that the subsidy represents. Table $YY$ shows the estimates of this first stage. Overall, the share of the pension that the subsidy represents is small because I am considering only workers that are close to the lower limit. 

In then used the instrumented subsidy amount in the second-stage as control in the following way:
\begin{equation} \label{second_withsubisdy}
    log(z_{it})=\alpha_0 + \alpha_t+\alpha_i+ \gamma_0 Post \cdot \widehat{r}_i + \gamma_1 Post \cdot \widehat{r}_i \cdot \widehat{s_i}+ \epsilon_{it}
\end{equation}

Table $XX$ shows the estimates of this second stage. We can see .... . As robustness, I use taxable earnings elasticity to future pension benefit, that I estimate in section \ref{pension_benefit_estimation}, jointly with the first stage to disentangle the effects. I find similar results.

\subsection{Taxable earnings elasticity to payroll tax}
\label{payroll_tax_estimation}

In this subsection, I estimate the elasticity of pre-retirement payroll earnings to payroll taxes using a change in the managment fee that one of the pension fund administrator did in 2014. 

As mention in the Context, since 2008, there is an auction for new workers every 2 years. The Pension Fund Administrator (PFA) that offers the lowest fee is forced to reduce its managment fee to the bidded one and in exchange receives every new worker until the next auction.  In 2014, Planvital PFA proposed a reduction of 1.89pp in its fee and won that year's auction for new workers. This fee reduction generated an exogenous increase in the monthly net-of-tax-income of 2.44pp to the 384,778 workers already affiliated with Planvital, while the rest of the workers affiliated with other PFAs kept their net-of-tax rate at the same level. Figure \ref{fees_ts_timeframe} shows the time series of administration fees of the different pension fund administrators for the period 2008-2019. Fees during this period were stable, with the only significant change in fee that of Planvital in 2014. 

One crucial aspect of the design is that workers choose their PFAs when they sign their first contract and then they can switch among them. This generates two concerns. First, workers differ across PFAs. This concern is alleviated by the fact that the two variables  shown when choosing PFAs are management fees and historical returns. These two variables were similar across pension funds for the period 1981-2008. Figure \ref{historical_fees_return_ts} shows the time series of management fees and 5-year moving average returns since system's inception for the different pension funds administrators. We see that historically, Planvital's fees and returns has been similar to the other administrators. For example, Planvital had a management fee below the system's average in 62\% of the months before 2014, and a 5-year average return above the average of the system in 67\% of the time. Therefore, PFA has not been consistently different from the others, which is reflected in the characteristics of its workers compared to the others PFAs. Table $XX$ shows the summary statistics of workers affiliated to Planvital and to other administrators for the period 2006-2014, and there were no large differences with respect to demographics and earnings. 

The second concern is that workers could have responded to the fee change by switching PFAs after the fee change. In the data I do not find evidence of this happening. The main reason for the introduction of new workers' auction in 2008 was the stickiness of workers to their PFA. \cite{switching_illanes} estimates that the switching cost between administrators must have been very large to explain the lack of switching between PFAs. 

I use the fee variation at the administrator level in 07/2014 in a difference-in-difference design. The treated group is the workers affiliated with Planvital before 2008, while the control group is those affiliated with other PFA. The main specification is the following: 

\begin{equation} \label{dif_dif}
y_{iat}=\alpha_t+\gamma_a+X_{i}\phi+\beta D_{iat}+\epsilon_{it}
\end{equation}
where $t$ is month; $i$ is worker; and $a$ is his Pension Fund Administrator. $\alpha_t$ and $\gamma_a$ are time and pension fund administrator fixed effects, respectively. $D_{iat}$ is the change in net-of-tax income generated by the fee change in month $t$. $X_{i}$ are controls at the worker level. I also allow for PFA-specific time trends. In the estimation, the independent variable $y_{ita}$ is the log of payroll earnings, so $\beta$ is the elasticity of payroll earnings to the net-of-tax rate.  

To study the dynamics of the response by doing and event study analysis: 

\begin{equation} \label{dif_dif_dynamic}
    y_{iat}=\alpha_t+\gamma_a+X_{i}\phi+ \sum_{k=-K+e}^{e-1} \beta_t D_{iat}^k+\sum_{k=e+1}^{e+K} \beta_t D_{iat}^k+\epsilon_{it}
\end{equation}
where $D_{iat}^k$ is equal 0 except for the month $k$ when it takes the value of the change in net-of-tax generated by the fee reduction for Planvital's workers. $e$ is the month when Planvital changed the fee and is omitted in the specification, so the estimates are normalized respect to $\beta_{t=e}$.    

The time frame is restricted to the 48 months before and after the fee change. This period is symmetric to the timing of the treatment and the analyzed fee change is the only one in the period. I restrict the sample to workers who were affiliated to the pension system before 2008. There are two reasons for this sample restriction. First, in 2008 Planvital made a fee increase making it the most expensive PFA in the system. Second, workers that receive their first paycheck after 2010 were assigned to PFAs based on the auctions. This sample restriction removes 22\% of the observations from the original sample.

Table \ref{ETI_table} shows estimates from equation (\ref{dif_dif}). The point estimate is $0.38$. The estimation is robust to different sample definition and using as control pre-treatment earnings and extensive margin participation. Figure \ref{ETI_dynamic_figure} plots the dynamic response captured by equation (\ref{dif_dif_dynamic}). We can see that there is no treatment effect before the treatment, which provides evidence that the parallel-trend assumption is satisfied, supporting the causal interpretation of the estimation. 

\subsection{Lifetime payroll earnings and consumption}
\label{earnings_consumption_section}
In this Section I analyze the relationship between lifetime payroll earnings and consumption. 

I use the merge between the adminstrative and survey data to build the joint distribution of lifetime payroll earnings ($z_i$), active consumption ($c_{i1}$) and retired consumption ($c_{i2}$) for workers that retired between 2006 and 2020. In total, I have 3,246 observations for whom I can observe these three variables.  

I start analyzing the relationship between consumption drop and lifetime earnings. To do so, I do an event analysis of consumption around retirement. For each worker in the data, I denote the three years after the individual retired from working by $t=0$, and index all years relative to that year. I bin 3 years in each bin for a sample power reason. Thus, $t=-2$ indicates the period between 3 and 6 years before retirement, $t=-1$ the 3 years before retirement, $t=0$ the 3 years after retirement, and so on. Result are robust to three definitions of the retirement moment: legal retirement age, year at which started receiving retirement benefits, or self-reported retirement year.\footnote{This robutness across definitions is because the difference of the retirement moment between this three definition is small.} Given the attrition of the survey, the panel is not balance across event times. 

I study the evolution of consumption as function of the event time for grouping workers in two groups $(g)$: those below ($g=0$) and above $(g=1)$ the average of lifetime earnings of their respective cohort and gender. Specifically, denoting by $c_{ist}^g$ the self-reported non-durable consumption individual $i$ of group $g$ in year $s$ and at event time $t$, I run the following regression separately for each group ($g$) of workers:

\begin{equation} \label{consumption_drop_eventstudy}
   log(c_{ist}^g)=\sum_{j \neq -1}\alpha_j^g \cdot \mathds{1}\left[j=t\right]+\sum_{y} \gamma^g_y \cdot  \mathds{1}\left[y=s\right]+ \beta_{HH}^g HH_{ist}+\beta_{dob}^g + v^g_{ist}  
\end{equation}
where I include a full set of event time dummies (first term on the right-hand side), year dummies (second term), and controls by household composition (third term) and cohort fixed effects (fourth term).

Figure \ref{retirement_drop} panel $(a)$ plots the group-specific effect of retirement on consumption across event time ($\alpha_t^g$). As defined above, these are outcomes at event time $t$ relative to the three years before retirement. The figure includes 95 percent confidence intervals bands around the event coefficient. We see that consumption drops more at retirement for workers with earnings below the average of lifetime earnings. To further highlight the consumption drop at retirement, I control for linear pre-trends before retirement. That is, I estimate a linear trend separately for workers above and below average lifetime earnings using only pre-retirement data, and then run the main event specifiation residualizing the outcome variable with the estimated pre-trend. The results are in Figure \ref{retirement_drop} panel $(b)$. 

Figure \ref{retirement_drop} panel $(c)$ plots the effect of retirement on self-reported income using the same estimation, but replacing consumption ($c_{its}^g$) by self-reported income ($y_{its}^g$). We see that income drops more for worker with low earnings, that is, the pension system replace less pre-retirement income to low-earners than to high-earners. Figure \ref{retirement_drop} panel $(d)$ plots the consumption drop at retirement but controlling for the income drop. After controlling for this drop, the consumption drop of low-earners gets reduced, while that of high-earners stays mostly unaffected. The difference in the consumption drop between the two groups is statistically insignificant after we control for income. This result is in-line with low-earners being less prepared to face the income drop at retirement and having to adjust consumption. 

The negative relationship between consumption drop at retirement and lifetime earnings is present for the whole distribution of earnings. Figure \ref{cons_drop_earnings} shows the average consumption drop for 52 equal-density bins of lifetime earnings in the scatter and plots the quadratic regression relationship between consumption drop at retirement and lifetime earnings in the line. The consumption drop is measured as the difference of log consumption residualized on year and cohort fixed effects, and household composition. We see that the consumption drop relationship with lifetime earnings is strong and negative. Therefore, to do retirement income redistribution from high to low earners generates a positive social value by improving social insurance against retirement income drop.

Given the relevance of the relationship between lifetime earnings and consumption drop at retirement, I use the survey to explore the mechanism behind this relationship. The analysis is in Appendix \ref{appendix_c}.  I find that the most likely explanation is the poor-preparedness for retirement of low-earners. Specifically, workers below the median of lifetime earnings receive less pension benefits relative to their pre-retirement income because a larger fraction of their income is informal and therefore is not covered by the mandatory Defined Contribution plan, and the progressive pension subsidy does not totally compensate this differential. For workers below the median, the pension benefit replacement rate of the average 10 years before retirement is around 35\%, while this number jumps over 60\% for those above the median of lifetime earnings. Additionally, workers with lifetime earnings below the median have less access to saving instruments, they do not use tax-advantaged pension savings, are more likely to retired because of unexpected shocks like health and employment and their jobs provide less protection against retirement. The relationship between lifetime payroll earnings and different socioeconomic variables related to retirement preparedness are in Figures \ref{first_ret_preparedness}-\ref{last_ret_preparedness}. Overall, the lifetime payroll earnings are a good tagging for retirement preparedness, making them a good tool to do inter-worker  redistribution. 

The second statistics that drives the value of inter-worker redistribution is the covariance between active-life consumption and lifetime earnings. Figure \ref{act_cons_earnings} shows that the relationship between these two variables is positive and strong. Thus, to do income redistribution using lifetime earnings does redistribution across workers with different productivity, generating a positive social value. 

\subsection{Calibration of preference parameters}
\label{calibration_estimaion}

There are three parameters that I do not estimate directly in the data: risk aversion ($\gamma$), retirement dependent preferences ($\theta$), and hyperbolic discounting ($\beta$). 

I calibrate the values for the RRA using the literature. \cite{landais2021value} estimate an relative risk aversion above 4 for Swedish workers. I calibrate $\gamma$ equal 4.  

Preferences over consumption may vary with retirement. At retiement, a smaller consumption spending may generate the same marginal utility of consumption. That is, some consumption drop at retirement is justified by worker preferences. The literature has found that reduction in food spending and transportation, among others, can explain part of the drop (\cite{aguiar2005consumption}). \cite{battistin2009retirement}  uses quasi-experimental variation on retirement dates to identify the consumption drop at retirement not driven by income and liquidity changes. They find that a 9.8\% of consumption drop can be accounted by drops in active-life related expenses. Using this value, jointly with the relative risk aversion parameter, I calibrate $\theta$ to be 0.62.

I calibrate $\beta$ using the meta-analysis of \cite{hyperbolic_meta}. This study finds that the average estimation of the hyperbolic discounting for monetary rewards find in the literature is 0.82 $[0.74, 0.90]$.  

I do comparative statics of the results with respect to this three parameters that I calibrate using the literature. 

Finally, I calibrate the return rate of the pension fund ($R$) using the average real return of the Fund D during the period 2010-2020, which was 1.048. By assumption, the discount factor ($\delta$) is the inverse of $R$.

\section{Results}
\label{results}

In this section I present the main result of the paper, splitting them in two. Section \ref{social_gains_result} shows the social gains of marginal reforms to the current design of the Chilean pension system. I decompose the social gains in its different components and do robustness analysis. Section \ref{optimal_design} shows the optimal design of the pension system and compare it to other countries.

\subsection{The social gains of marginal reforms}
\label{social_gains_result}

\noindent \textbf{Current design of the Chilean pension system-.} The current design of the Chilean pension system is similar, but different than the framework of Section \ref{frame_work}. Therefore, I modify equations (\ref{foc_progressivity}) and (\ref{foc_contribution_rate}) to adjust them to the current design of the Chilean pension system. The main difference is that there is no linear tax on the self-funded pension, but instead the government subsidy is progressive with respect to the self-funded pension, creating an implicit tax on recipients. This has two consequences. First, a share of the pension contribution is not received in retirement, because the pension subsidy is reduced with the additional pension saving. Second, an increase in pension contribution rate ($\kappa$) has an effect on the fiscal budget through the pension subsidy. To make the reform budget-balanced, I assume that this positive or negative effect on the fiscal budget is given back to workers through a lump-sum transfer to those who receive the subsidy. 

Let $I_i^s$ be an indicator that worker $i$ is receiving the pension subsidy and $\phi'$ be the implicit tax generated by the current pension subsidy ($\approx 33\%$). Then, the welfare effect of marginal reform to benefits' progressivity ($\phi$) in the current design of the Chilean system are given by:
\begin{align}
    \frac{dW}{d\phi}=-&R\kappa\left(\mu Cov\left[d_i,z_i\right]+Cov\left[U_{c1}(c_{i1}),z_i\right]\right)  +(\tau +\phi' I^s_i \kappa)R \overline{z}[-\varepsilon_{m\overline{z}}+\varepsilon_{b\overline{z}}] \int_i (\mu d_i+U_{c1}(c_{i1})) di \notag  \\ &+   \mu(1-\beta)\kappa \int_i(I^s_i(1-\phi')+(1-I_s))(\mu d_i+U_{c1}(c_{i1}))[-\varepsilon_{m z_i}+\varepsilon_{b z_i}] z_i di \label{foc_progressivity_chile}
\end{align}
And the social gains of marginal reform to pension contribution rate in the current design of the Chilean system are given by:
\begin{align}
    \frac{dW}{d\kappa}= &  \mu Cov\left[d_i,zi \right]+  \mu \overline{d}\overline{z}-R\kappa\left( \mu Cov\left[d_i,z_i\right]+Cov\left[U_{c1}(c_{i1}),z_i\right]\right) \notag \\
     & +(\tau +\phi' I^s_i \kappa)R \overline{z} \left[\frac{1-\kappa}{1-\tau}\varepsilon_{(1-\tau)\overline{z}}-\frac{1-\phi}{\phi}\varepsilon_{m\overline{z}}+\varepsilon_{b\overline{z}}\right] \int_i (\mu d_i+U_{c1}(c_{i1})) di \notag  \\
    &+\mu(1-\beta) \kappa  \int_i ((1-\phi')I^s_i+(1-I^s_i))\left[\frac{1-\kappa}{1-\tau}\varepsilon_{(1-\tau)z_i}-\frac{1-\phi}{\phi}\varepsilon_{mz_i}+\varepsilon_{bz_i}\right]( d_i+U_{c1}(c_{i1}))  z_i  di \label{foc_contribution_rate_chile}
\end{align}

\noindent \textbf{Social gains of marginal reforms-.} Table \ref{model_parameters} shows the value and source of the parameters used in the estimation of equations (\ref{foc_progressivity_chile}) and (\ref{foc_contribution_rate_chile}). 

I find that there are significant social gains of increasing benefit progressivity and pension contribution. For exposition, I measure the money-metric value of reforms' social gains as a share of the mechanical transfer. This measure tells how many dollars of social gain are generated by each dollar that the reform would transfers across workers and time without considering behavioral responses (mechanical transfer). 

This measure is similar to the Marginal Value of Public Funds (\cite{hendren2020unified}) but it is more intuitive to the present analysis. By construction, the reforms analyzed here are budget balanced and therefore the MVPF does not apply directly. An alternative to my measure is to split the reform into two: the revenue collecting reform (increase in tax) and the revenue spending one (increase in transfer). Then build the MVPF for each of them and compare those values. My measure is more direct, and a positive value of it means that the MVPF from revenue spending policy is lager than the MVPF from the revenue collection. An additional advantage of my measure is that it is easy to approximate the social gains of large reforms because we know the mechanical transfer done by reforms. 

My results suggest that both, raising the pension contribution rate ($\kappa$) and benefit progressivity ($\phi$), are socially desirable. For \$1 mechanically transfered by an increase in benefits progressivity or by an increase in the pension contribution rate, the money-metric social gains are \$0.45  and \$0.12, respectively. Given the size of the pension system, these gains are economically relevant. For example, an increase in 1\% in benefits progressivity would generate social gains of 0.38\% of GDP, while an increase in 1\% in the pension contribution rate would generate social gains of 0.14\% of GDP.

Figures \ref{decomp_progressivity} and \ref{decomp_progressivity} show the social gains of reform decomposing the welfare effect into its different components: intra-worker and inter-worker redistribution, fiscal externality, and bias correction. The socially desirability of these reforms is driven, mainly, by the inter-worker value of redistribution. Transferring retirement income to worker with low lifetime earnings generates large social gains, which overcomes the cost of behavioral distortion generated by the reforms. 

Nonetheless, the cost of reforms' behavioral distortion is economically significant, increasing the cost of doing redistribution across time and workers through the pension system. Specifically, for \$1 mechanically transferred by an increase in benefits progressivity, the behavioral distortion cost is \$0.32, i.e., almost one-third of the mechanical transfer. Here, both the fiscal externality and the bias correction term are important, the latter accounting for 38\% of the behavioral distortion cost. Thus, to not consider the present-focused bias of workers will overestimate in 26\% the social gains of a marginal increase in benefits progressivity.

The behavioral distortion generated by an increase in the pension contribution rate is smaller than that of the benefits progressivity, still its economic impact is significant. For every \$1 that is mechanically transferred by an increase in the pension contribution rate, the welfare cost of the behavioral distortion is \$0.11. To not consider the workers' present focused bias would overestimate the social gains in 25\%. \\

\noindent \textbf{Robustness-.} I analyze the sensitivity of my results to:  relative risk aversion ($\gamma$), state-dependent preferences ($\theta$), and present-focus bias ($\beta$). These are the three parameters that I do not estimate directly. 

Figure \ref{comp_stat_rra} showd the comparative statics of the social gains of each reform with respect to relative risk aversion ($\gamma$). The social gains of reforms are increasing with relative risk aversion (RRA) as the income redistribution becomes more socially valuable. Both reforms are socially desirable for RRA greater than 1. 

Figure \ref{comp_stat_consdrop} shows the comparative statics of the social gains of each reform with respect to retirement consumption preference ($\theta$). I do the comparative statistics with respect to the rational consumption drop at retirement implied by each $\theta$, a more intuitive measure. For example, the benchmark $\theta=0.62$  implies a rational consumption drop at retirement of 9.8\%. I find that the social gain of reforming benefit progressivity is unaffected by $\theta$ because this reform transfer retirement income across workers and not across time. In the other hand, the social gains of an increase in the pension contribution rate is strictly decreasing in $\theta$, because as $\theta$ decreases, consumption is less valuable in retirement. Overall, there are social gains of increasing the pension contribution rate for any $\theta$ such that the rational consumption drop at retirement is smaller than 22.5\%. 

Finally, in Figure \ref{comp_stat_bias}, I show the comparative statics pf social gains for changes in the level of present-focused bias ($\beta$). As workers become more biased (smaller $\beta$), the social gains of reforms are reduced. However, for almost every value of bias, both reforms are socially desirable.\\

\noindent \textbf{Heterogeneous life-expectancy extension-.} I can also extend the model to capture heterogenous life expectancy. In Appendix \ref{appendix_d_LE}, I estimate the relationship between life expectancy and lifetime payroll earnings. I find that low-earners have shorter life expectancies in retirement. In average, workers below median lifetime earnings live 18\% less after their retirement age than do workers above (result in line with \cite{cristia2007empirical_lifeexpectancy}).

Introducing this heterogeneity in my analysis reinforces the results. I find that under heterogeneous life expectancy, the money-metric social gains of increasing benefit progressivity goes from \$0.45 to \$0.58 per \$1 mechanically transferred. The reason for this large increase in the social gains of reforming progressivity is that this reform focuses benefits on those with lower income, who also have a shorter retirement time. Therefore, heterogeneous life expectancy makes it fiscally cheaper to raise progressivity. Similarly, the social gains of increasing pension contribution increase from \$0.12 to \$0.15 by incorporating heterogenous life expectancy. 

\subsection{Optimal design}

\label{optimal_design}

In this Section, I estimate the optimal value of the two parameters of the pension system: the pension contribution rate ($\kappa$) and benefits' progressivity ($\phi$). I then compare this optimal design with that of other systems around the world.

The optimal system is given by parameters ($\kappa^*, \phi^*$) such that there are no social gains from reforms. I first show numerically that the determinant of the social welfare's Hessian is positive ($D(H W(\kappa,\phi))>0$) and that the social welfare function is strictly concave on $\kappa$ and $\phi$ for any $0 \leq \kappa,\phi \leq 1$. I then use (\ref{foc_progressivity}) and (\ref{foc_contribution_rate}) to build the first order conditions. The solution is unique. 

The optimal pension contribution rate is 16.8\% and the optimal tax on pension contribution is 68\%. Confidence area around the point estimates employ two different methods. First, using each parameter confidence interval I estimate the area of points that satisfy the first order conditions. Second, I use in the fact that my model is not computationally burdensome and I do a \emph{Pairs Bootstrap} (\cite{freedman1984bootstrapping}) for the whole data after which I estimate the optimal parameters. Both methods give similar confidence areas. Figure \ref{world_progr_contribution} shows the confidence area for the Bootstrap method. 

I find that the optimal system significantly increases the pension contribution rate and benefits progressivity. To put the reforms in perspective, I compare the optimal design with the current design of other countries' pension system. To do so, I build a systems' \emph{progressivity} measure as the  the amount of resources spent by on workers below the median benefit. For example, in New Zealand, the public pension is a lump-sum transfer, therefore the measure of progressivity is 0.5 because everyone (above and below the median) receives the same benefit. Given that in almost every system the relationship of lifetime earnings and benefits is positive, so median benefit is close to median lifetime earnings. Figure \ref{world_progr_contribution} shows this \emph{progressivity} measure in the x-axis and the pension contribution rate in the y-axis for 23 countries.

We see that the pension contribution rate is increased from levels similar to the USA to that of Germany and Netherlands. Similarly, with the optimal tax on pension contribution, the ratio between pension funds going to workers below and above the median goes from 29\%-- similar to that of USA and France-- to 42\%-- similar to that of Japan and Canada--.

\section{Conclusion}

There are two questions at the core of reforms to the public pension system design: Should we increase the pension system size by raising workers' pension contribution when active? How should savings be divided among workers when they retire? 

To answer these two questions, I provide a general framework that incorporates the trade-off between incentive provision and redistribution that drives the welfare effect of these reforms. I then estimate the trade-off for Chile. On the incentive distortion side, I causally estimate the effect that a change in the pension contribution rate and benefit progressivity has on taxable earnings. On the redistribution side, I measure the equity gains using a panel survey matched with administrative data. I find that the trade-off resolves in favor of increasing both, i.e., it would be socially desirable to increase pension contributions from earnings, increasing savings for retirement, and to raise benefits progressivity, focusing pension benefits on low-earners.     

With population aging, the financial sustainability of public pension systems has been stressed, forcing governments to revisit the design of their public pension systems. To reduce pension system deficits, governments can increase pension contributions or reduce benefits. Even though my analysis focuses on budget-balanced reforms, the results of this paper shed light on this discussion. It shows that for Chile, instead of reducing benefits, it is better to increase pension contributions. Additionally, if a reduction of benefits must be made, it is better to do it in a progressive way, reducing benefits more for high-income workers. 

There are many remaining aspects relevant to the optimal design of the pension system. A crucial one pertain to the reasons for the inadequate preparation of workers for retirement. What role do individual preferences (present-focused), information, and workers' ability to save play? The relative importance of each is a crucial ingredient for future pension design reforms. 

\newpage

\bibliography{pensions}

\nocite{*}

\FloatBarrier

\newpage

\section*{Figures}
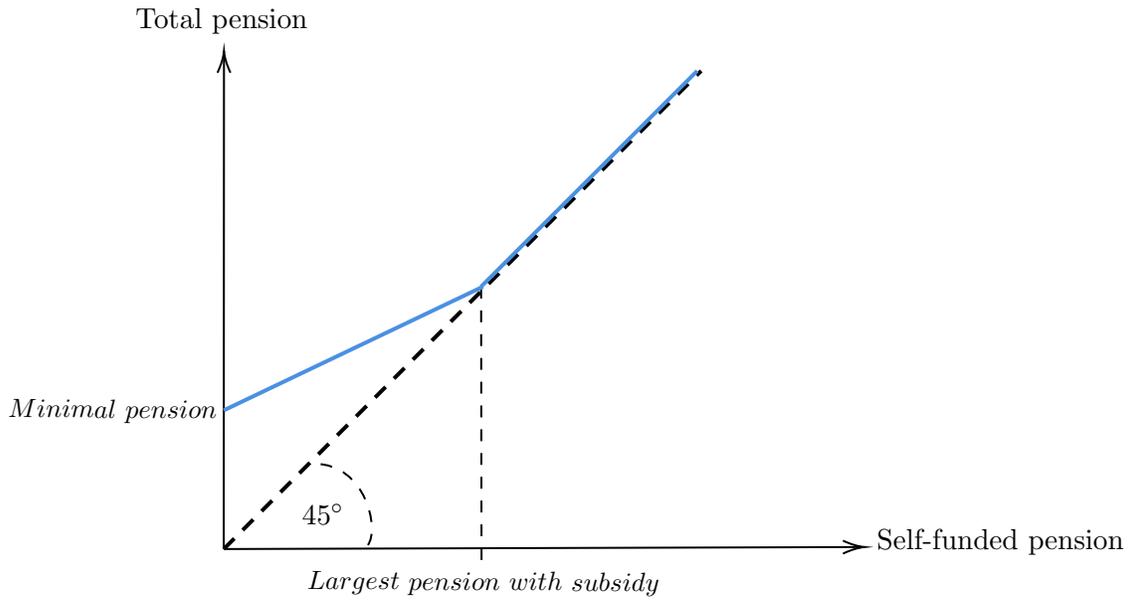
\begin{figure}[h]
    \centering
\tikzset{every picture/.style={line width=0.75pt}} 

\begin{tikzpicture}[x=0.75pt,y=0.75pt,yscale=-1,xscale=1]

\draw    (154,414.67) -- (475.22,413.67) ;
\draw [shift={(477.22,413.67)}, rotate = 179.82] [color={rgb, 255:red, 0; green, 0; blue, 0 }  ][line width=0.75]    (10.93,-3.29) .. controls (6.95,-1.4) and (3.31,-0.3) .. (0,0) .. controls (3.31,0.3) and (6.95,1.4) .. (10.93,3.29)   ;
\draw    (154,414.67) -- (154.11,165.33) ;
\draw [shift={(154.11,163.33)}, rotate = 90.03] [color={rgb, 255:red, 0; green, 0; blue, 0 }  ][line width=0.75]    (10.93,-3.29) .. controls (6.95,-1.4) and (3.31,-0.3) .. (0,0) .. controls (3.31,0.3) and (6.95,1.4) .. (10.93,3.29)   ;
\draw  [dash pattern={on 4.5pt off 4.5pt}]  (283.78,282.22) -- (284,421) ;
\draw [color={rgb, 255:red, 6; green, 0; blue, 1 }  ,draw opacity=1 ][line width=1.5]  [dash pattern={on 5.63pt off 4.5pt}]  (154,414.67) -- (235.24,333.24) -- (394.78,173.33) ;
\draw [color={rgb, 255:red, 74; green, 144; blue, 226 }  ,draw opacity=1 ][line width=1.5]    (154,344.67) -- (284.78,282.22) ;
\draw [color={rgb, 255:red, 74; green, 144; blue, 226 }  ,draw opacity=1 ][line width=1.5]    (283.78,282.22) -- (392.78,173.33) ;
\draw  [dash pattern={on 4.5pt off 4.5pt}]  (226.67,412.67) .. controls (232.67,404.67) and (225.67,371.67) .. (199.67,371.67) ;

\draw (204,397) node  [color={rgb, 255:red, 0; green, 0; blue, 0 } ,draw opacity=1 ,rotate=-0.34] {$45^{\circ}  $};
\draw (98,345) node  [font=\footnotesize,color={rgb, 255:red, 0; green, 0; blue, 0 }  ,opacity=1 ,rotate=-0.34]  {$Minimal\ pension$};
\draw (285,432) node  [font=\footnotesize,color={rgb, 255:red, 0; green, 0; blue, 0 }  ,opacity=1 ,rotate=-0.34]  {$Largest\ pension\ with\ subsidy$};
\draw (482,403) node [anchor=north west][inner sep=0.75pt]   [align=left] {Self-funded pension};
\draw (108,140) node [anchor=north west][inner sep=0.75pt]   [align=left] {Total pension};

\end{tikzpicture}

\caption{Chilean pension system design}
    \label{chilean_design}
\end{figure}

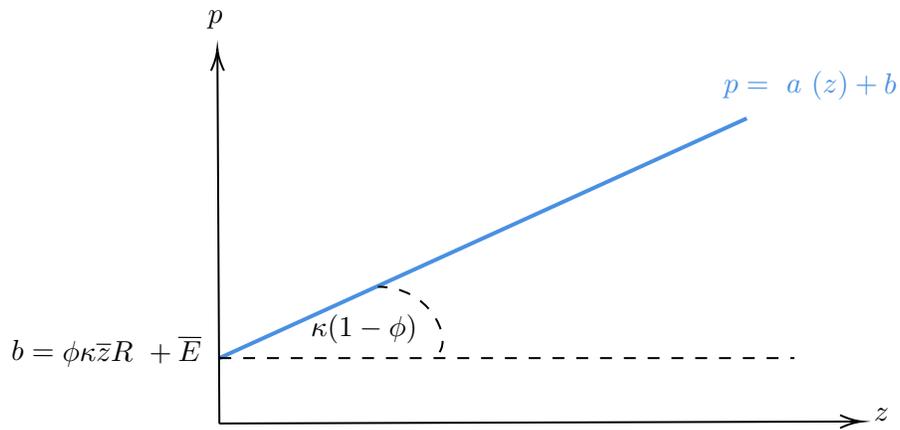
\begin{figure}[h!]
\tikzset{every picture/.style={line width=0.75pt}} 
\centering
\begin{tikzpicture}[x=0.75pt,y=0.75pt,yscale=-1,xscale=1]
\draw [color={rgb, 255:red, 74; green, 144; blue, 226 }  ,draw opacity=1 ][line width=1.5]    (232,359.67) -- (498,238.67) ;
\draw    (232,392.67) -- (554,391.67) ;
\draw [shift={(556,391.67)}, rotate = 179.82] [color={rgb, 255:red, 0; green, 0; blue, 0 }  ][line width=0.75]    (10.93,-3.29) .. controls (6.95,-1.4) and (3.31,-0.3) .. (0,0) .. controls (3.31,0.3) and (6.95,1.4) .. (10.93,3.29)   ;
\draw    (232,392.67) -- (231.01,205.67) ;
\draw [shift={(231,203.67)}, rotate = 89.7] [color={rgb, 255:red, 0; green, 0; blue, 0 }  ][line width=0.75]    (10.93,-3.29) .. controls (6.95,-1.4) and (3.31,-0.3) .. (0,0) .. controls (3.31,0.3) and (6.95,1.4) .. (10.93,3.29)   ;
\draw  [dash pattern={on 4.5pt off 4.5pt}]  (232,359.67) -- (522,359.67) ;
\draw  [dash pattern={on 4.5pt off 4.5pt}]  (343,356.67) .. controls (349,348.67) and (338,323.67) .. (312,323.67) ;

\draw (530,222) node  [color={rgb, 255:red, 74; green, 144; blue, 226 }  ,opacity=1 ,rotate=-0.34]  {$p=\ a\ ( z) +b$};
\draw (566,388) node  [color={rgb, 255:red, 0; green, 0; blue, 0 }  ,opacity=1 ,rotate=-0.34]  {$z$};
\draw (175,356) node  [color={rgb, 255:red, 0; green, 0; blue, 0 }  ,opacity=1 ,rotate=-0.34]  {$b=\phi \kappa \overline{z} R\ +\overline{E}$};
\draw (230,189) node  [color={rgb, 255:red, 0; green, 0; blue, 0 }  ,opacity=1 ,rotate=-0.34]  {$p$};
\draw (305,345) node  [color={rgb, 255:red, 0; green, 0; blue, 0 }  ,opacity=1 ,rotate=-0.34]  {$\kappa ( 1-\phi )$};
\end{tikzpicture}
\caption{Pension system}
\label{graph_pension_system}

\end{figure}

\begin{figure}[h!] 
  \tikzset{every picture/.style={line width=0.75pt}} 
\centering
\scalebox{0.78}{
\begin{tikzpicture}[x=0.75pt,y=0.75pt,yscale=-1,xscale=1]

\draw [color={rgb, 255:red, 74; green, 144; blue, 226 }  ,draw opacity=1 ][line width=1.5]    (41,352.67) -- (292.33,237) ;
\draw    (41,385.67) -- (355.01,386.36) ;
\draw [shift={(355.01,386.36)}, rotate = 180] [color={rgb, 255:red, 0; green, 0; blue, 0 }  ][line width=0.75]    (10.93,-3.29) .. controls (6.95,-1.4) and (3.31,-0.3) .. (0,0) .. controls (3.31,0.3) and (6.95,1.4) .. (10.93,3.29)   ;
\draw    (41,385.67) -- (40.01,198.67) ;
\draw [shift={(40,196.67)}, rotate = 89.7] [color={rgb, 255:red, 0; green, 0; blue, 0 }  ][line width=0.75]    (10.93,-3.29) .. controls (6.95,-1.4) and (3.31,-0.3) .. (0,0) .. controls (3.31,0.3) and (6.95,1.4) .. (10.93,3.29)   ;
\draw [color={rgb, 255:red, 208; green, 2; blue, 27 }  ,draw opacity=1 ][line width=1.5]    (41.33,328) -- (250.67,193) ;
\draw [color={rgb, 255:red, 208; green, 2; blue, 27 }  ,draw opacity=1 ] [dash pattern={on 4.5pt off 4.5pt}]  (234.33,262) .. controls (243.15,240.44) and (230.84,228.48) .. (204.94,226.13) ;
\draw [shift={(203.33,226)}, rotate = 4.24] [color={rgb, 255:red, 208; green, 2; blue, 27 }  ,draw opacity=1 ][line width=0.75]    (10.93,-3.29) .. controls (6.95,-1.4) and (3.31,-0.3) .. (0,0) .. controls (3.31,0.3) and (6.95,1.4) .. (10.93,3.29)   ;
\draw [color={rgb, 255:red, 74; green, 144; blue, 226 }  ,draw opacity=1 ][line width=1.5]    (400,352.67) -- (680,218) ;
\draw    (400,385.67) -- (722,384.67) ;
\draw [shift={(724,384.67)}, rotate = 179.82] [color={rgb, 255:red, 0; green, 0; blue, 0 }  ][line width=0.75]    (10.93,-3.29) .. controls (6.95,-1.4) and (3.31,-0.3) .. (0,0) .. controls (3.31,0.3) and (6.95,1.4) .. (10.93,3.29)   ;
\draw    (400,385.67) -- (399.01,198.67) ;
\draw [shift={(399,196.67)}, rotate = 89.7] [color={rgb, 255:red, 0; green, 0; blue, 0 }  ][line width=0.75]    (10.93,-3.29) .. controls (6.95,-1.4) and (3.31,-0.3) .. (0,0) .. controls (3.31,0.3) and (6.95,1.4) .. (10.93,3.29)   ;
\draw [color={rgb, 255:red, 208; green, 2; blue, 27 }  ,draw opacity=1 ][line width=1.5]    (398.5,306.67) -- (689.5,264) ;
\draw [color={rgb, 255:red, 208; green, 2; blue, 27 }  ,draw opacity=1 ] [dash pattern={on 4.5pt off 4.5pt}]  (666,226) .. controls (679.72,225.02) and (689.6,232.68) .. (679.63,261.23) ;
\draw [shift={(679,263)}, rotate = 290.14] [color={rgb, 255:red, 208; green, 2; blue, 27 }  ,draw opacity=1 ][line width=0.75]    (10.93,-3.29) .. controls (6.95,-1.4) and (3.31,-0.3) .. (0,0) .. controls (3.31,0.3) and (6.95,1.4) .. (10.93,3.29)   ;
\draw [color={rgb, 255:red, 208; green, 2; blue, 27 }  ,draw opacity=1 ] [dash pattern={on 4.5pt off 4.5pt}]  (432,335) .. controls (416.96,319.96) and (421.36,313.76) .. (427.76,306.42) ;
\draw [shift={(429,305)}, rotate = 131.19] [color={rgb, 255:red, 208; green, 2; blue, 27 }  ,draw opacity=1 ][line width=0.75]    (10.93,-3.29) .. controls (6.95,-1.4) and (3.31,-0.3) .. (0,0) .. controls (3.31,0.3) and (6.95,1.4) .. (10.93,3.29)   ;
\draw  [dash pattern={on 0.84pt off 2.51pt}]  (538,287.33) -- (538,386) ;

\draw (363,384) node  [color={rgb, 255:red, 0; green, 0; blue, 0 }  ,opacity=1 ,rotate=-0.34]  {$z$};
\draw (27,350) node  [color={rgb, 255:red, 0; green, 0; blue, 0 }  ,opacity=1 ,rotate=-0.34]  {$b$};
\draw (40,184) node  [color={rgb, 255:red, 0; green, 0; blue, 0 }  ,opacity=1 ,rotate=-0.34]  {$p$};
\draw (205,249) node  [color={rgb, 255:red, 208; green, 2; blue, 27 }  ,opacity=1 ,rotate=-0.34]  {$\Delta ^{+} \kappa $};
\draw (733,382) node  [color={rgb, 255:red, 0; green, 0; blue, 0 }  ,opacity=1 ,rotate=-0.34]  {$z$};
\draw (385,351) node  [color={rgb, 255:red, 0; green, 0; blue, 0 }  ,opacity=1 ,rotate=-0.34]  {$b$};
\draw (399,184) node  [color={rgb, 255:red, 0; green, 0; blue, 0 }  ,opacity=1 ,rotate=-0.34]  {$p$};
\draw (444,314) node  [color={rgb, 255:red, 208; green, 2; blue, 27 }  ,opacity=1 ,rotate=-0.34]  {$\Delta ^{+} \phi $};
\draw (654,252) node  [color={rgb, 255:red, 208; green, 2; blue, 27 }  ,opacity=1 ,rotate=-0.34]  {$\Delta ^{+} \phi $};
\draw (537,395) node  [color={rgb, 255:red, 0; green, 0; blue, 0 }  ,opacity=1 ,rotate=-0.34]  {$\overline{z}$};
\draw (78,412) node [anchor=north west][inner sep=0.75pt]   [align=left] {(a) Increase in contribution rate};
\draw (438.33,412) node [anchor=north west][inner sep=0.75pt]   [align=left] {(b) Increase in benefits' progressivity};
\end{tikzpicture}}
\caption{Reforms to pension system parameters}
\label{graph_reform}
\end{figure}
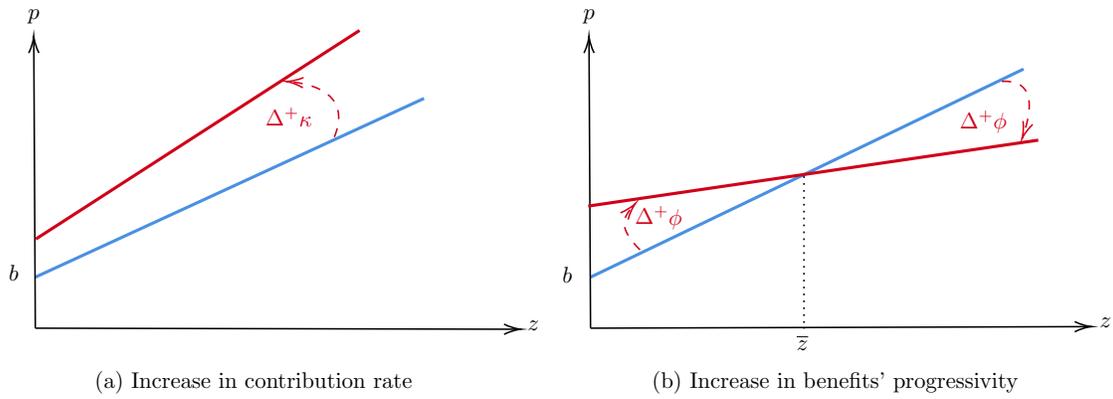

\begin{figure}
    \centering
    \includegraphics[scale=0.6]{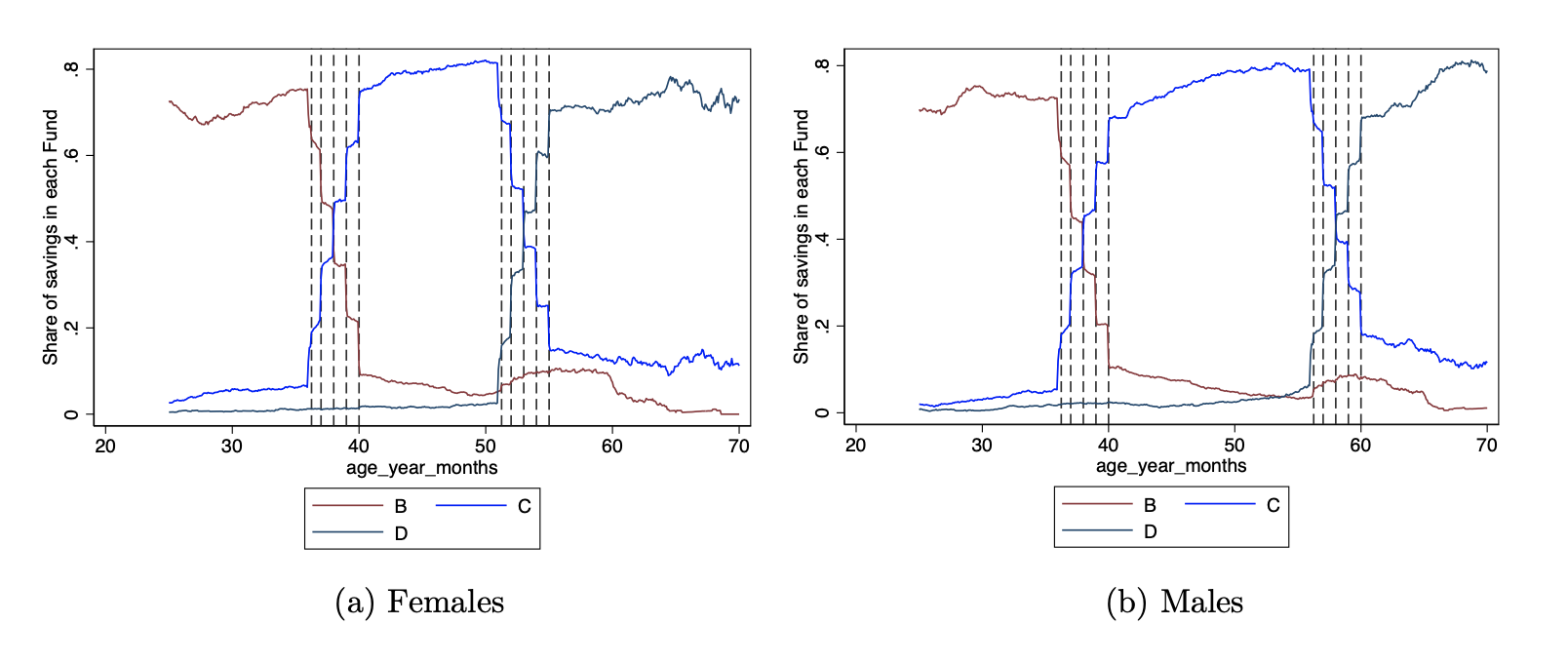}
    \caption{Fund allocation by age (months) in period 2008-2014}
    \label{fund_allocation}
    \begin{minipage}{\textwidth}
        {\footnotesize \textit{Notes:} This figure shows the investment share of workers' pension savings in funds B,C and D any given age (defined in months) for the period 2008-2014 . \par}
    \end{minipage}
\end{figure}

\begin{figure}[ht] 
  \begin{subfigure}[b]{0.5\linewidth}
    \centering
    \includegraphics[width=1\linewidth]{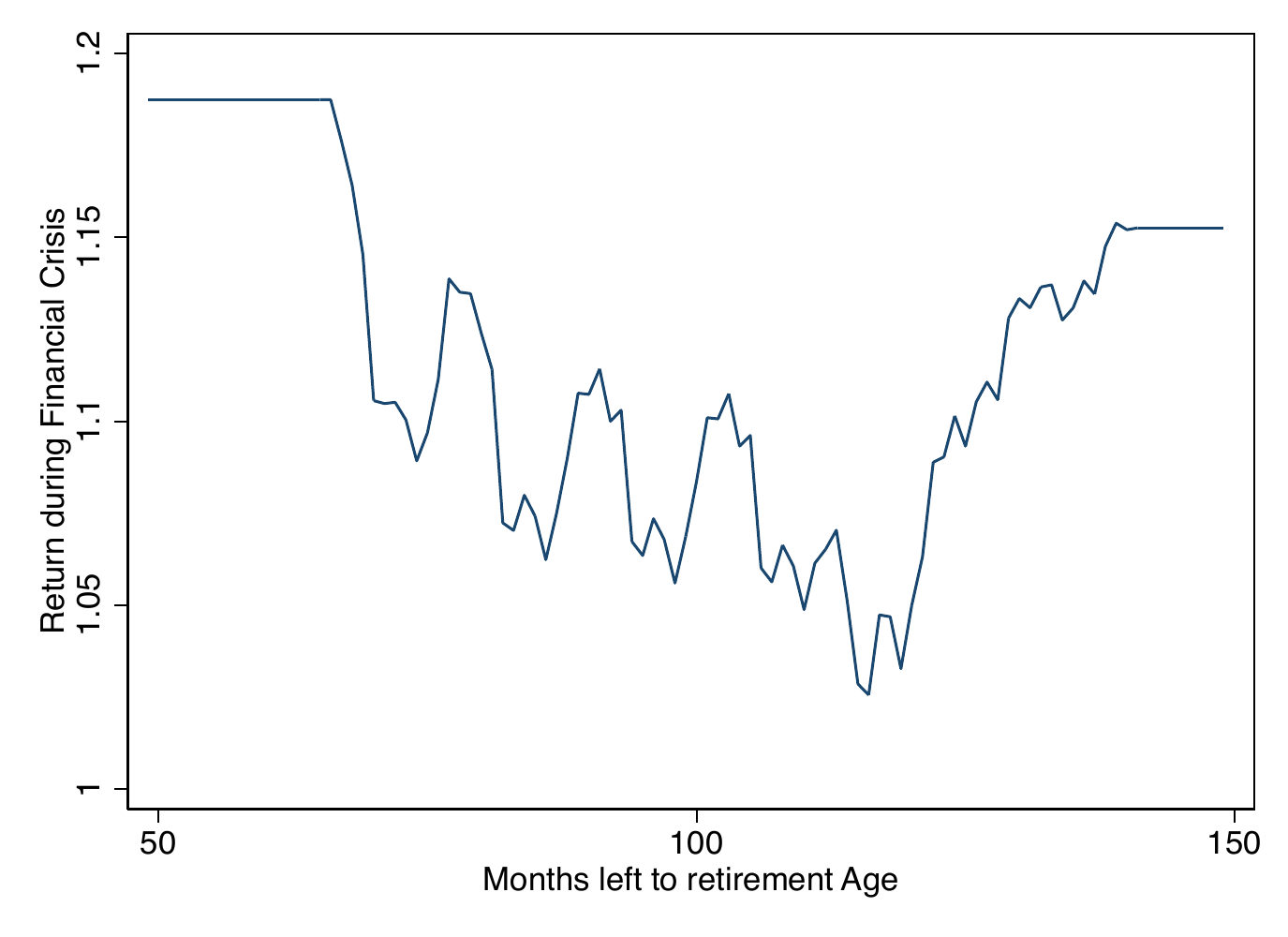} 
    \caption{Short crisis definition} 
    \vspace{4ex}
  \end{subfigure}
  \begin{subfigure}[b]{0.5\linewidth}
    \centering
    \includegraphics[width=1\linewidth]{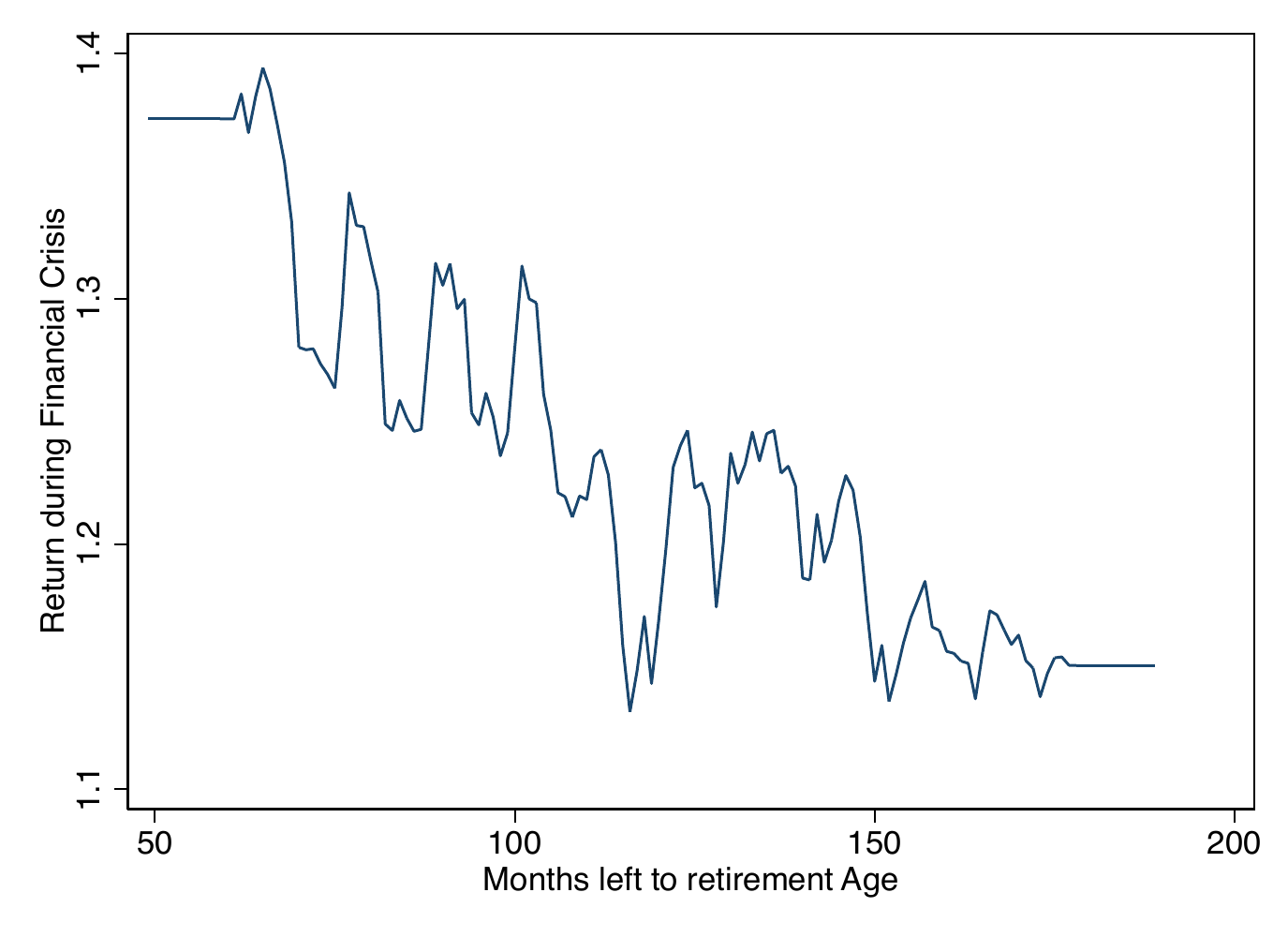} 
    \caption{Long crisis definition} 
    \vspace{4ex}
  \end{subfigure} 
  \caption{Return of Pre-crisis savings during Global financial crisis}
      \label{pen_sav_shock}
          \begin{minipage}{\textwidth} 
    {\footnotesize \textit{Notes:} This figure shows $\rho_i/S_{i}^{pre}$, as defined in equation (\ref{gfc_return}) by cohorts defined as the number of months left to turn the retirement age at the beginning of the Great Financial Crisis. \par}
    \end{minipage}
\end{figure}

\begin{figure}
    \centering
    \includegraphics[scale=0.7]{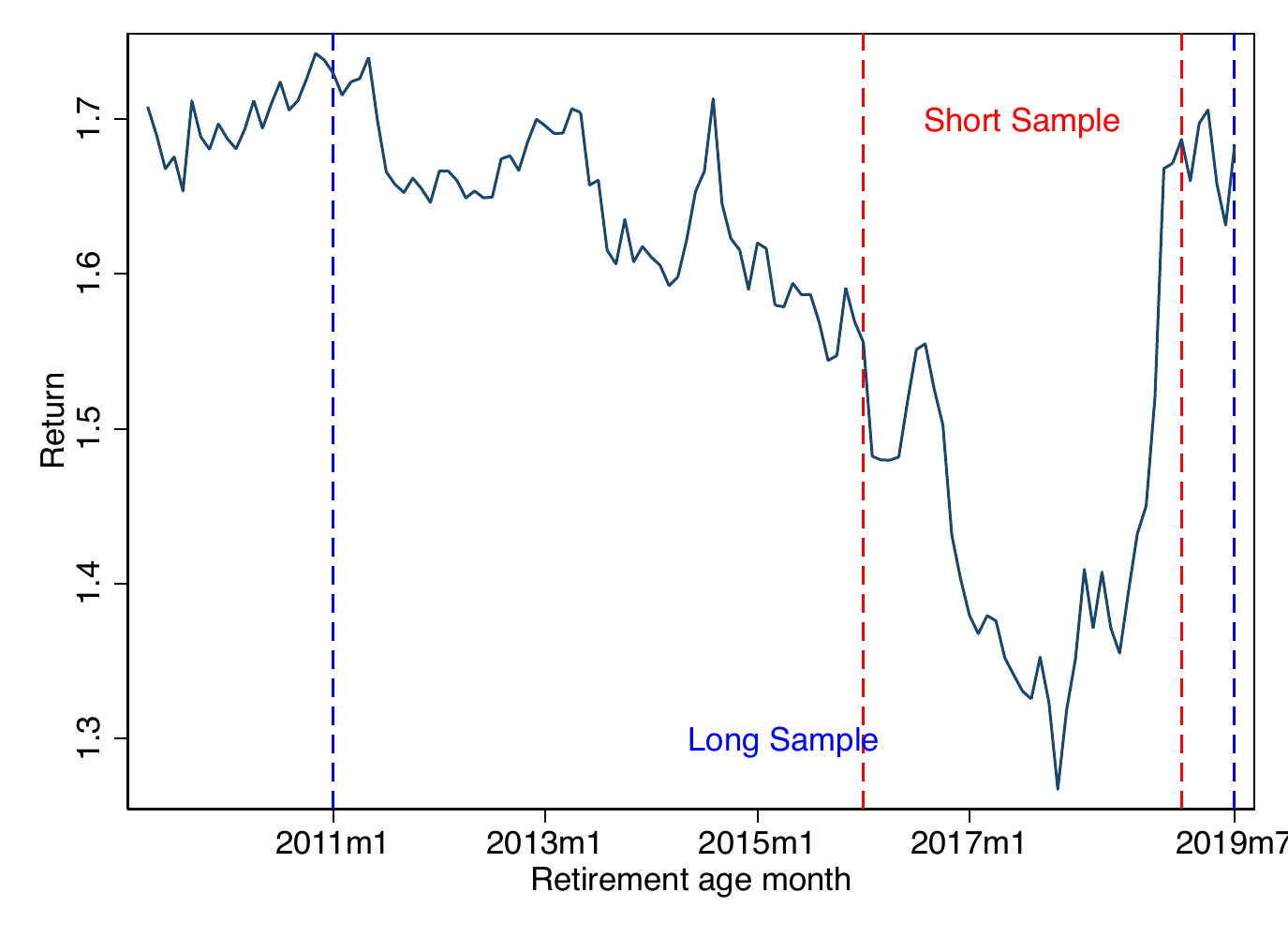}
    \caption{Return of pension savings in the 10 year before retirement}
    \label{last10y}
          \begin{minipage}{\textwidth} 
    {\footnotesize \textit{Notes:} This figure shows the average  return during workers’ last 10 years by cohort for those turning the retirement age between 2009 and 2020  \par}
    \end{minipage}
\end{figure}

\begin{figure}[h!] 
  \begin{subfigure}[b]{0.5\linewidth}
    \centering
    \includegraphics[scale=0.5]{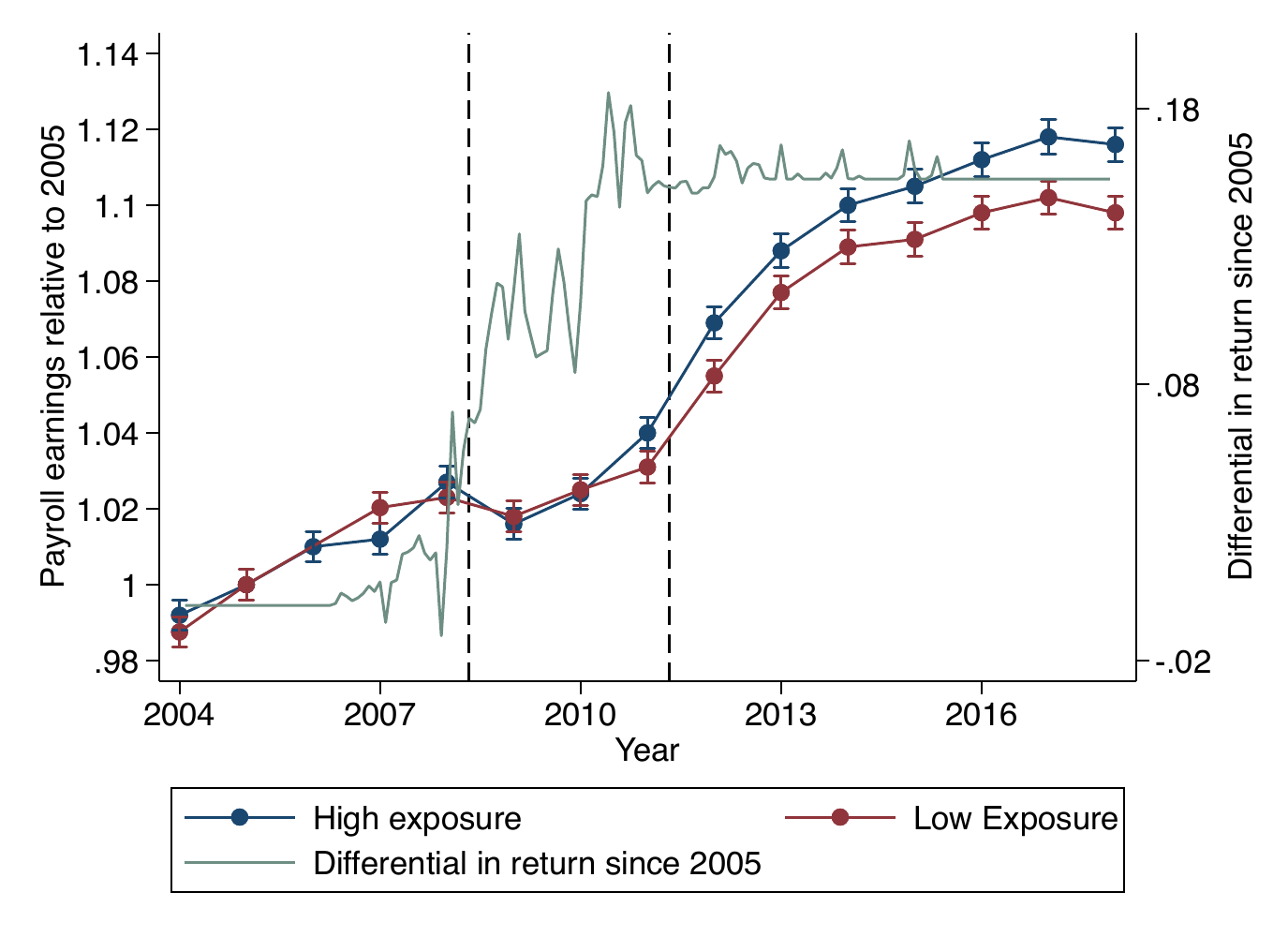}
    \caption{Treated cohorts}
    \label{fig7:a} 
    \vspace{4ex}
  \end{subfigure}
  \begin{subfigure}[b]{0.5\linewidth}
    \centering
    \includegraphics[scale=0.5]{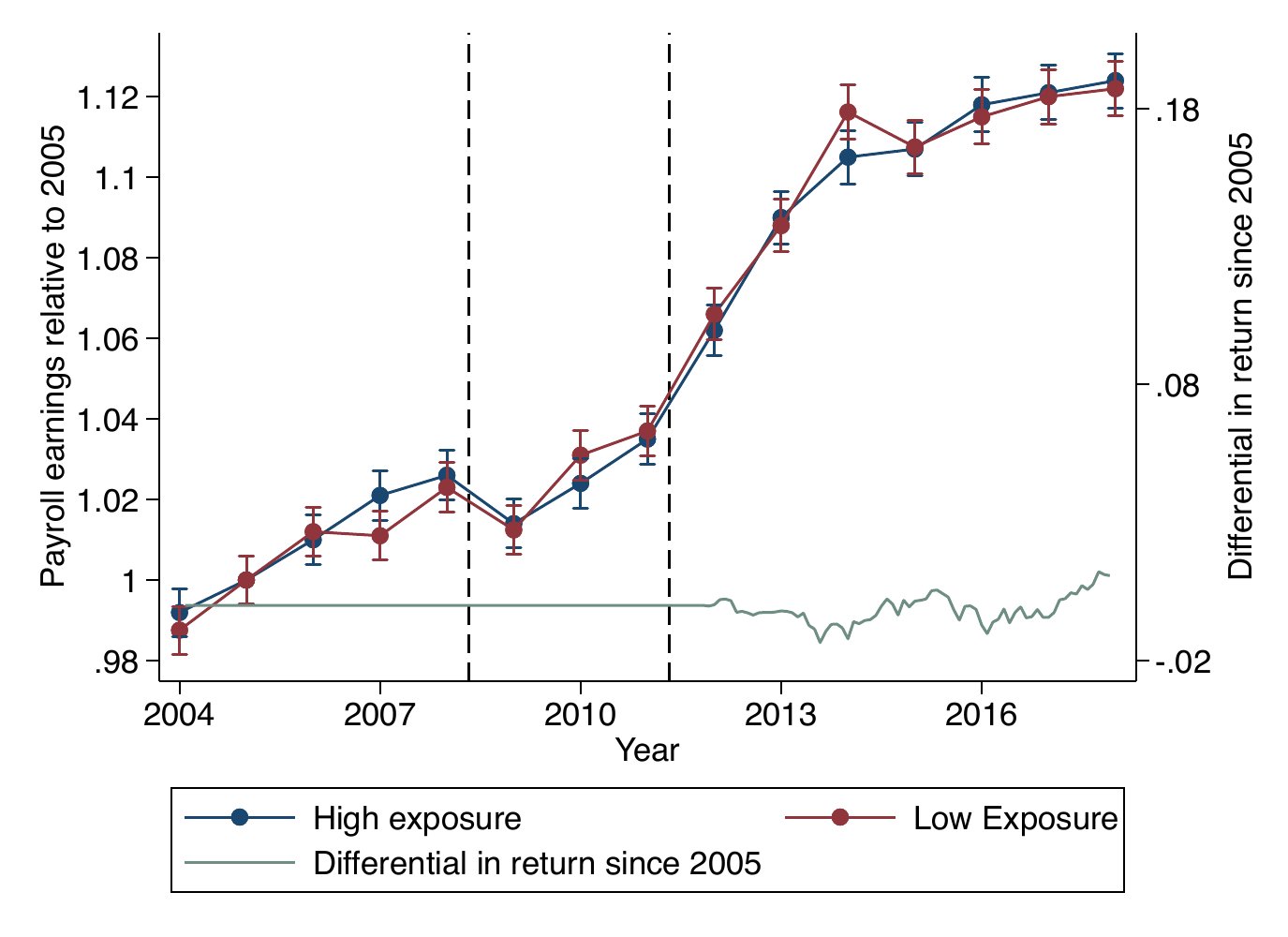}
    \caption{Placebo cohorts}
    \label{fig7:b} 
    \vspace{4ex}
  \end{subfigure} 
  \caption{Return during GFC and taxable earnings}
  \label{gfc_earnings_reducedform} 
      \label{last10y}
          \begin{minipage}{\textwidth} 
    {\footnotesize \textit{Notes:} The left axis of this figure shows the outcome, i.e., the taxable earnings relative to 2005, after controlling by age fixed effect and gender, separately for cohorts with High and Low exposure to the GFC. The right axis of this figure shows the treatment, i.e., the average difference in returns since 2005 between cohorts with High and Low exposure. Panel (a) uses the treated sample: workers that turned the retirement age between 2016-2019 and had the largest heterogeneity on returns during the GFC. For this sample, the High exposure is defined as those cohorts that obtained a return during the GFC below the median of the sample, and the low exposure is the complement. Panel (b) uses the placebo sample: workers that will turn the retirement age between 2021-2024 and had no heterogeneity on returns during the GFC. For this sample, High and Low exposure is defined by the age rank of each cohort inside the sample using the equivalent rank and treatment allocation of the treated sample.   \par}
    \end{minipage}
\end{figure}

\begin{figure}[h!]
    \centering
    \includegraphics[scale=0.4]{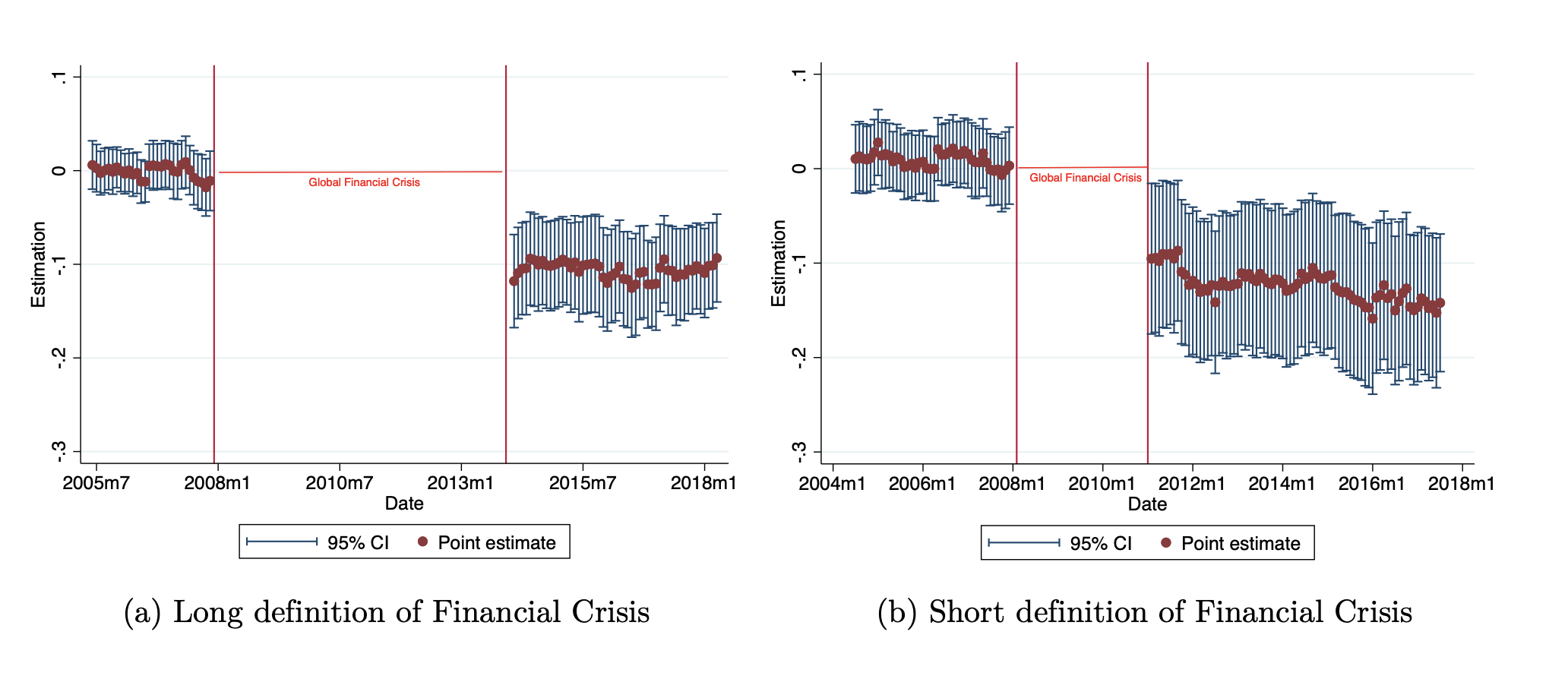}
    \caption{Dynamics of the effect of Pension Payment on Income}
    \label{dynamicsresponse}
    \medskip 
    \begin{minipage}{\textwidth} 
    {\footnotesize \textit{Notes:} Plotted is the estimation of $(\{\beta_{-\tau}\}_{\tau=0}^{\tau=m}, \{\beta_{\tau}\}_{\tau=1}^{\tau=q})$ from equation (\ref{dynamics_pen_sav_shock}). In panel (a) is used the long definition of financial crisis, and on panel (b) the short one. All standard error are clustered at date of birth-gender level. Both estimations use the wide sample.  \par}
    \end{minipage}
\end{figure}

\begin{figure}
    \centering
\includegraphics[scale=0.45]{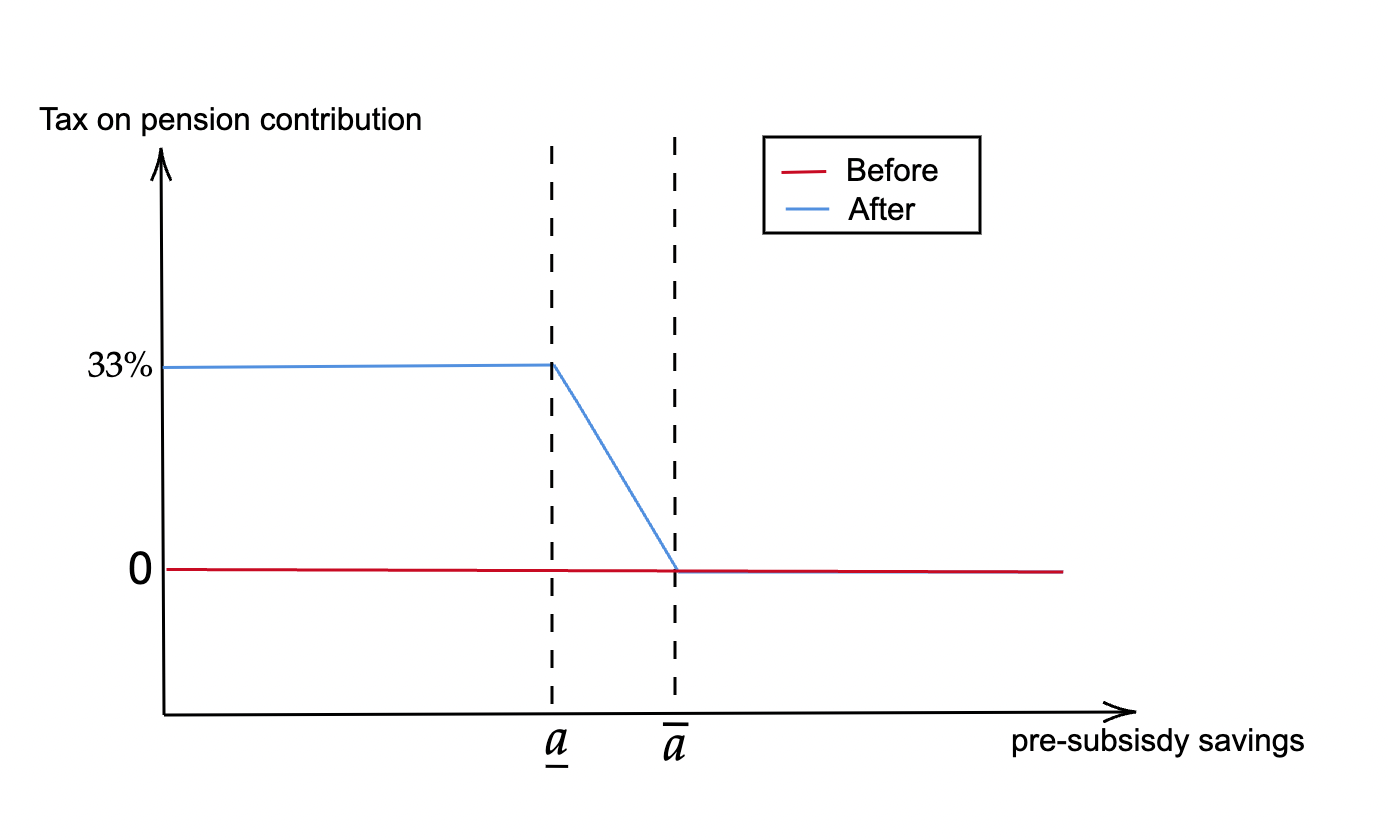}
\caption{Effect subsidy introduction on benefits-earnings link}
    \label{discont_APS_tax}
                \begin{minipage}{\textwidth} 
    {\footnotesize \textit{Notes:} This figure shows the average tax on pension contribution on earnings after the subsidy introduction for different levels of pre-retirement savings if this worker earns the limit of taxable earnings in every month until her retirement age, and she belongs to the 60\% poorer of the population. \par}
    \end{minipage}
\end{figure}

\begin{figure}
    \centering
\includegraphics[scale=0.65]{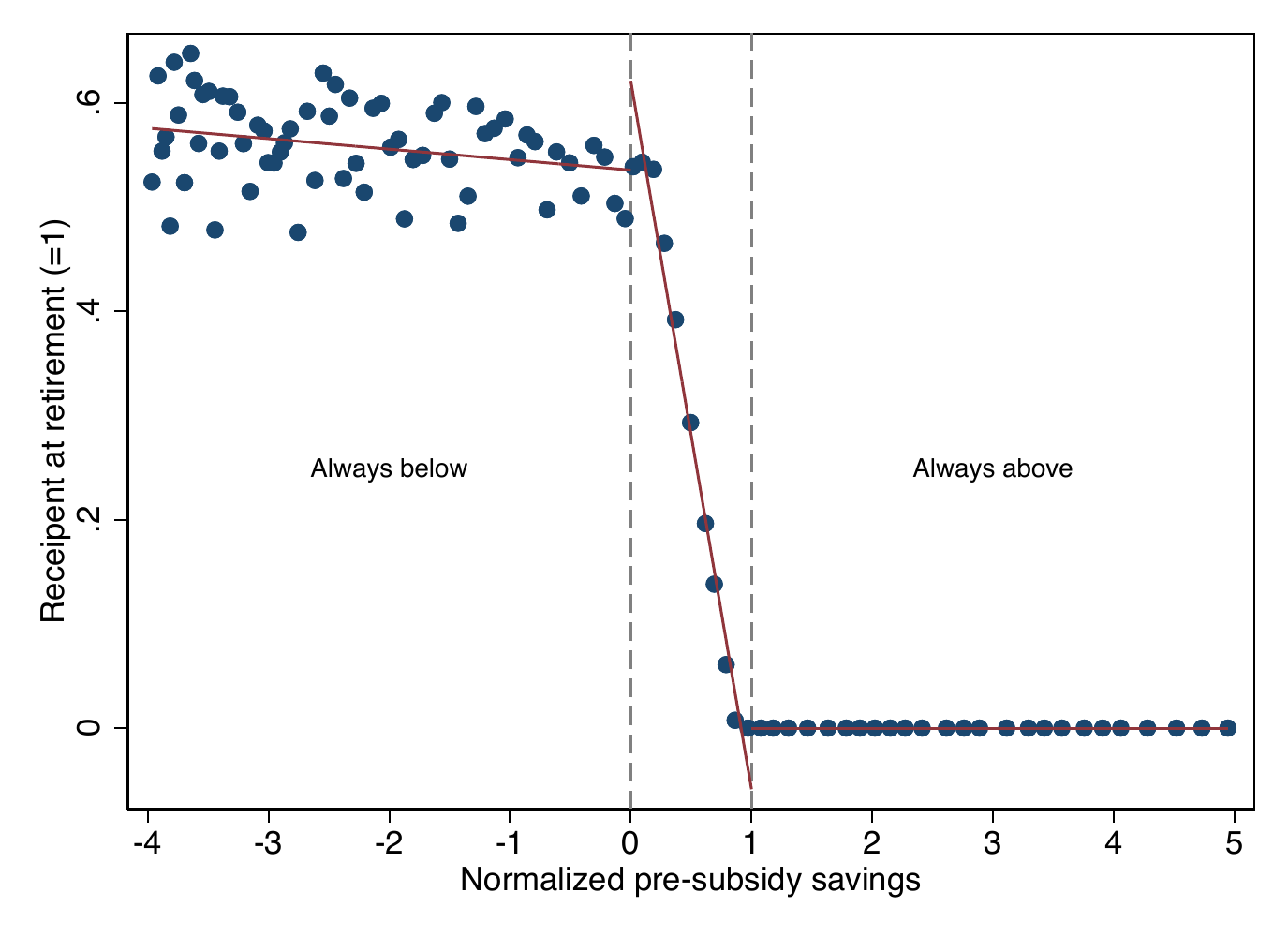}
\caption{Pre-subsidy savings and subsidy reception at retirement} \label{presubsidy_above_below}
            \begin{minipage}{\textwidth} 
    {\footnotesize \textit{Notes:} This figure shows the average recipient status for bins of normalized pre-subsidy earnings. The normalization is $(a_i-\underline{a}_i)/(\overline{a}_i-\underline{a}_i)$, and is done to taken in account that the lower and upper bound are worker specific. If the normalized pre-subsidy savings is below 0, then the worker is in the area of always below, and if it is above 1, then the worker is always above.  \par}
    \end{minipage}
\end{figure}




\begin{figure}
    \centering
    \includegraphics[scale=0.75]{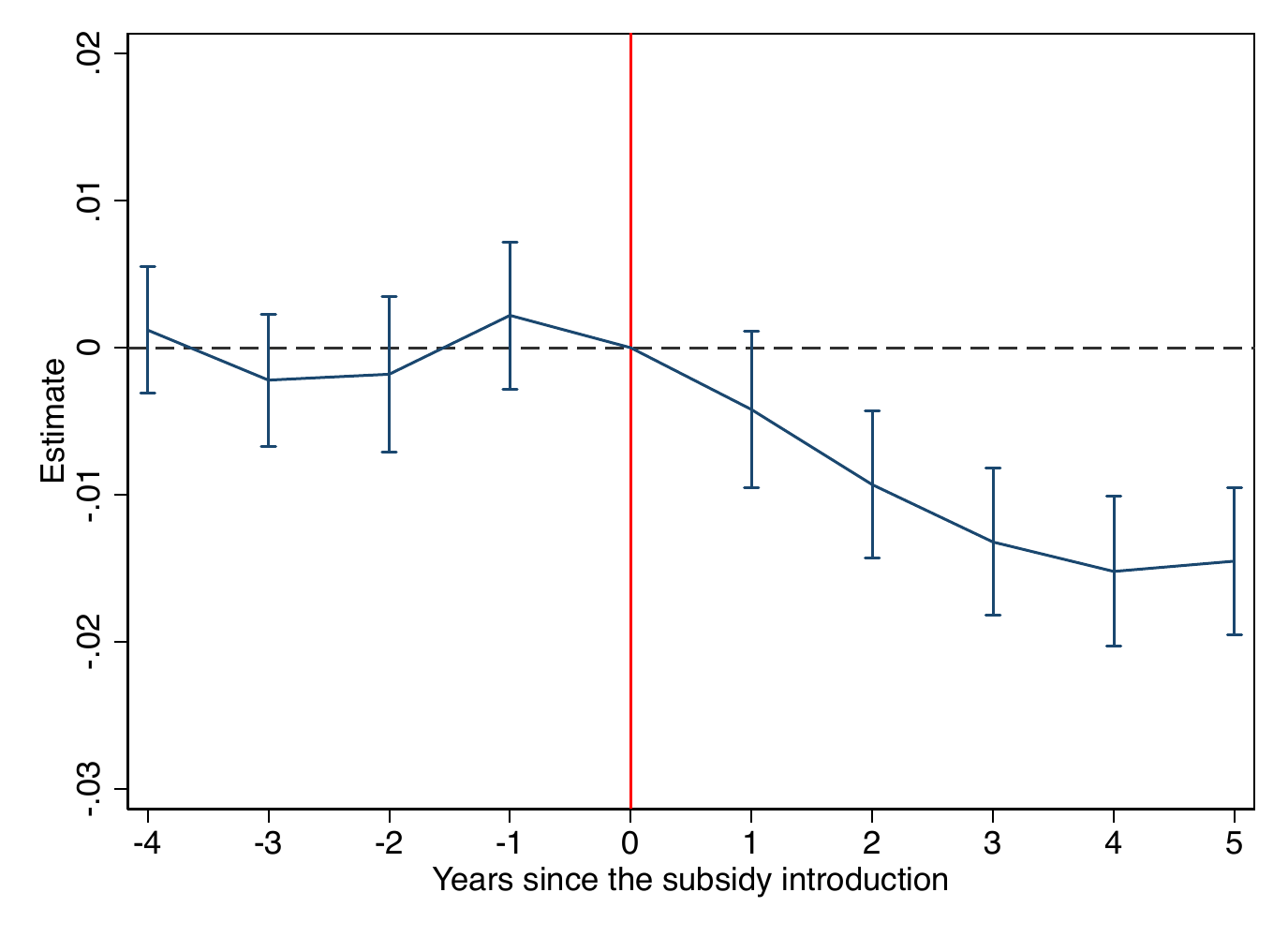}
    \caption{Event study of the effect of subsidy introduction on earnings}
    \label{dynamic_subsidy_intro}
    \begin{minipage}{\textwidth} 
    {\footnotesize \textit{Notes:} Plotted is the estimation of $\{\beta_t\}_{t=-5}^{t=5}$ from equation (\ref{dynamics_subsidy_eq}). The parameters are normalized to the year before the subsidy introduction ($t=0$). The estimation controls with age, year and worker fixed-effects. The sample of workers is given by workers less than 6 years away of retirement at the subsidy introduction, with pre-subsidy savings larger than 0.9 times the lower bound ($0.9 \underline{a}_i$), smaller than 1.1 times the upper bound ($1.1\overline{a}_i$), and are not between bounds. The standard error are cluster at the date of birth and the caps are the 95\% confidence interval for each parameter. \par}
    \end{minipage}
\end{figure}

\begin{figure}
    \centering
\begin{subfigure}[b]{0.5\linewidth}
    \centering
    \includegraphics[scale=0.5]{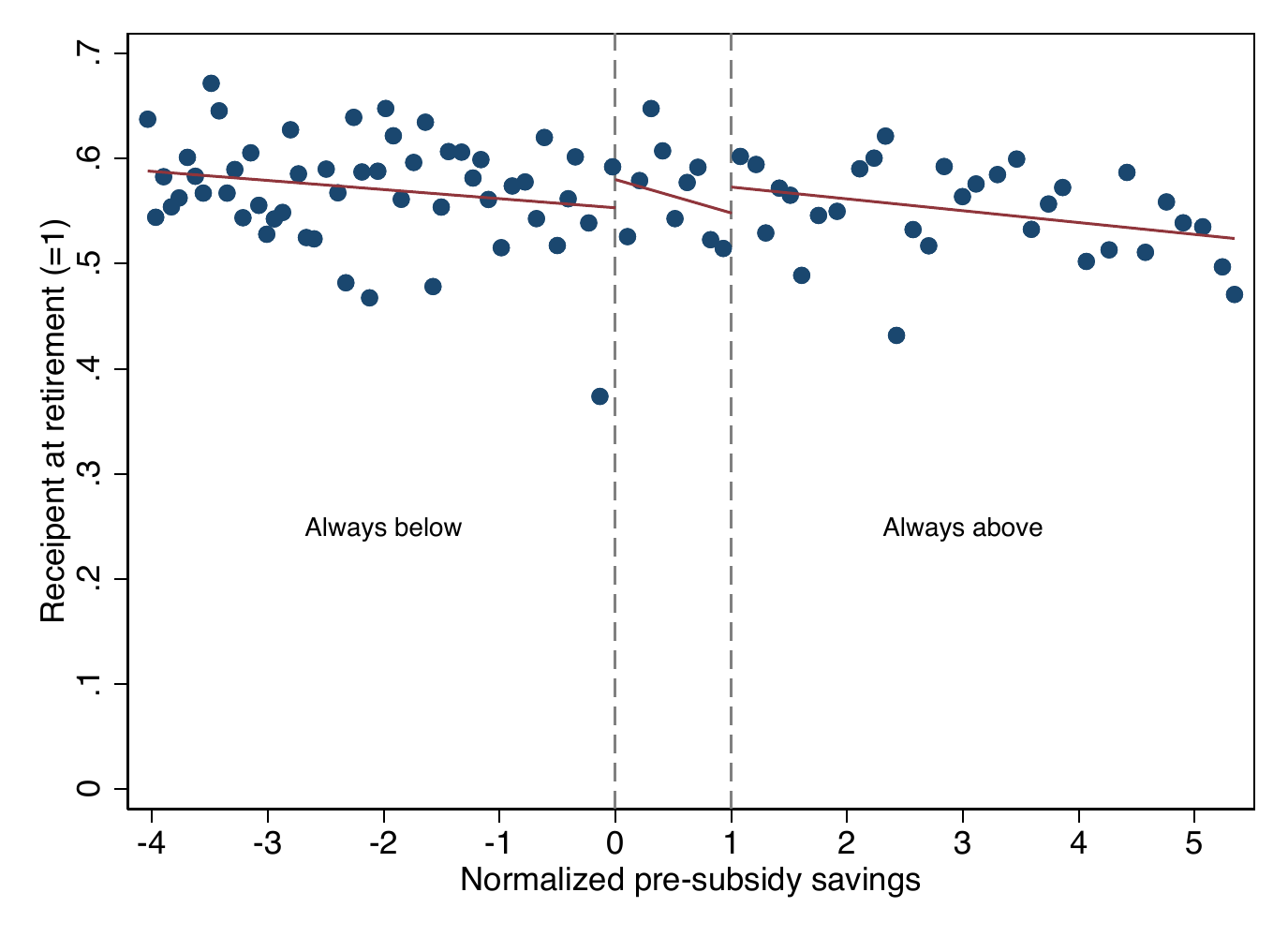}
    \caption{Lower PMAS}
    \label{fig7:a} 
    \vspace{4ex}
  \end{subfigure}
  \begin{subfigure}[b]{0.5\linewidth}
    \centering
    \includegraphics[scale=0.5]{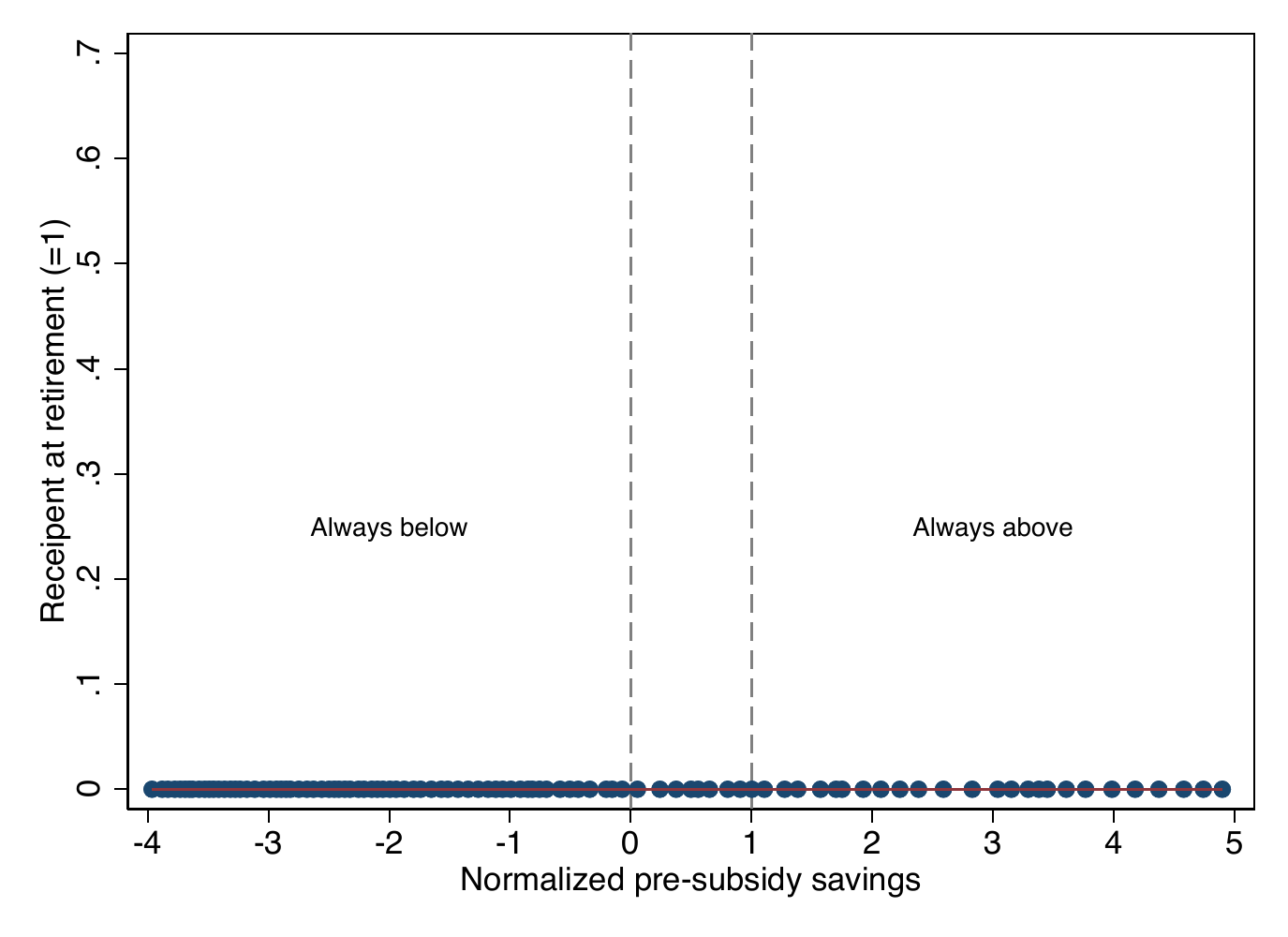}
    \caption{Higher PMAS}
    \label{fig7:b} 
    \vspace{4ex}
  \end{subfigure} 
  \caption{Placebo and recipient status}
    \label{presubsidy_above_below_placebo}
            \begin{minipage}{\textwidth} 
    {\footnotesize \textit{Notes:} This figure shows the average recipient status for bins of normalized pre-subsidy earnings for the placebo definitions of PMAS. The normalization is the same used for Figure \ref{presubsidy_above_below}, i.e., $(a_i-\underline{a}_i)/(\overline{a}_i-\underline{a}_i)$, where the difference is that to build the limits ($\underline{a}_i, \overline{a}_i$) I use placebo PMAS. Panel (a) uses a placebo PMAS that is the 20\% smaller than the real PMAS, and Panel (b) uses a PMAS that is 20\% larger than the real PMAS.  \par}
    \end{minipage}
\end{figure}

\begin{figure}[h!]
    \centering
    \includegraphics[scale=0.7]{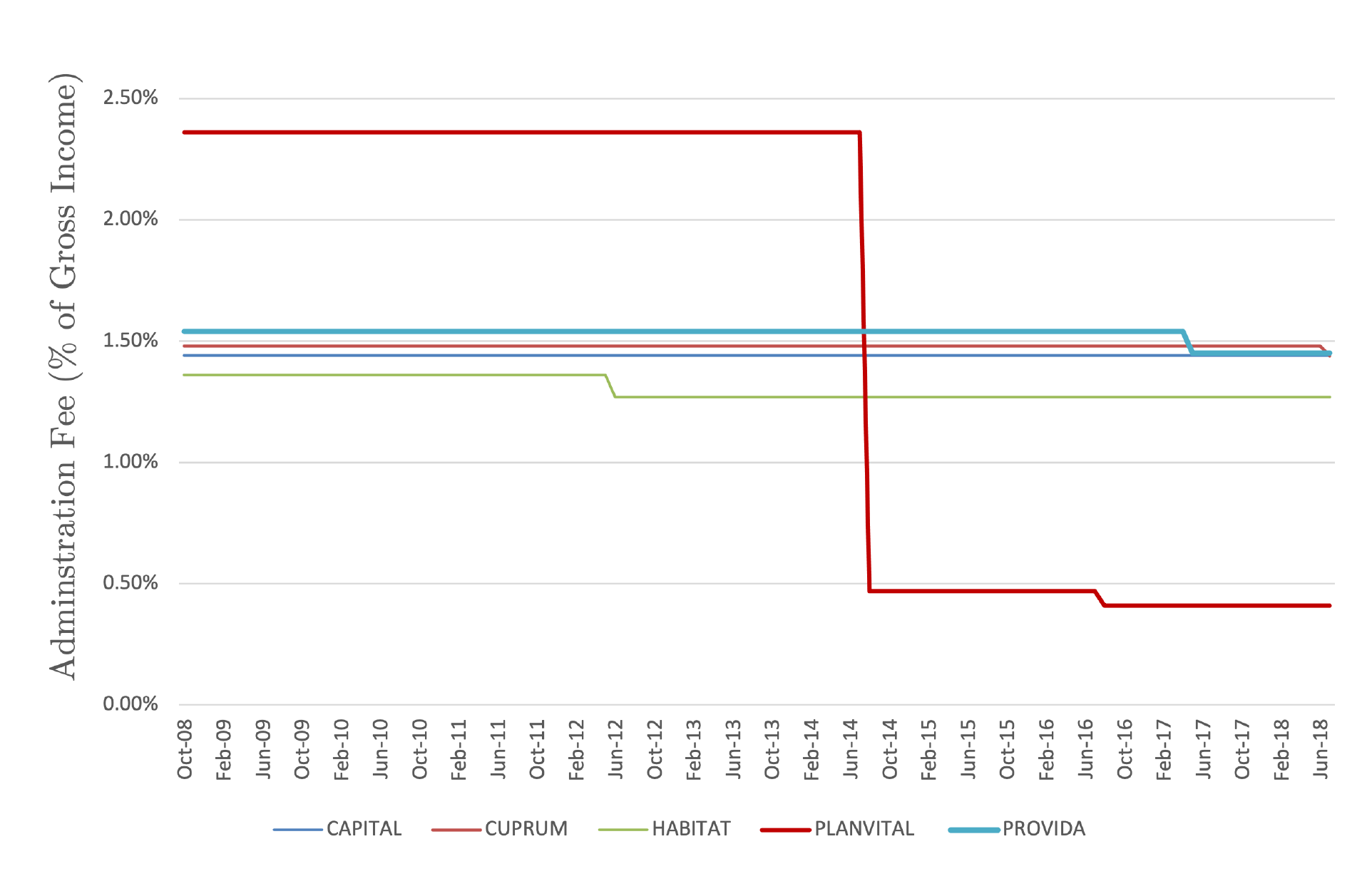}
    \caption{Time Series of Pension Fund Administrator Administration Fees}
    \label{fees_ts_timeframe}
\end{figure}

\begin{figure}[h!]
    \centering
    \includegraphics[scale=0.75]{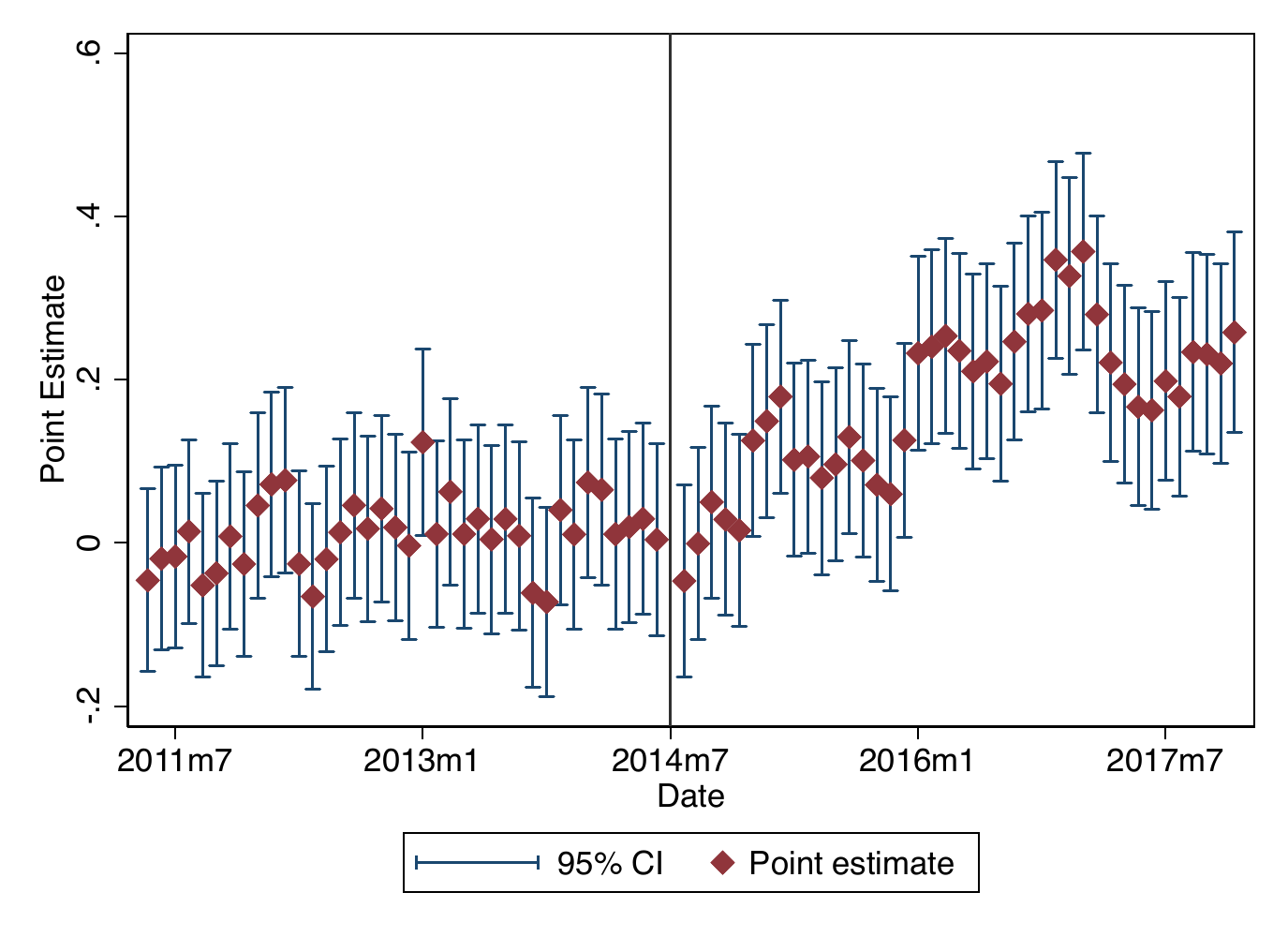}
    \caption{Elasticity of Taxable Income Dynamics}
    \label{ETI_dynamic_figure}
    \medskip 
    \begin{minipage}{\textwidth} 
    {\footnotesize \textit{Notes:} Plotted is the estimation of $(\{\beta_{-\tau}\}_{\tau=0}^{\tau=m}, \{\beta_{\tau}\}_{\tau=1}^{\tau=q})$ from equation (\ref{dif_dif_dynamic}). The estimation has time and PFA fixed-effects. Controls are age; gender; and pre-treatment earnings, extensive margin participation and pension savings. Workers that have ever switched between pension fund are part of the estimation. \par}
    \end{minipage}
\end{figure}

\begin{figure}[h!] 
  \begin{subfigure}[b]{0.5\linewidth}
    \centering
    \includegraphics[scale=0.5]{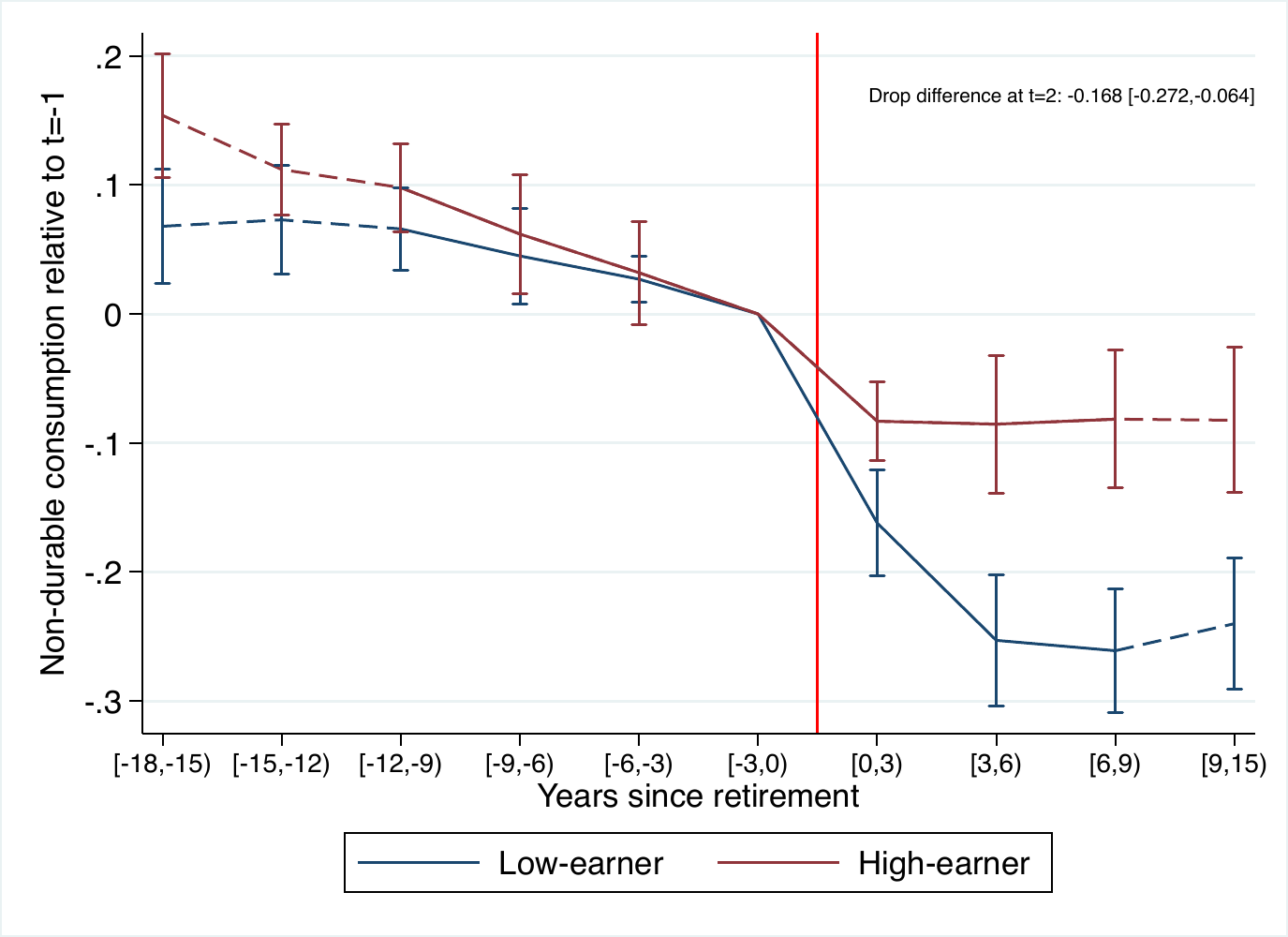}
    \caption{Consumption}
    \label{fig7:a} 
    \vspace{4ex}
  \end{subfigure}
  \begin{subfigure}[b]{0.5\linewidth}
    \centering
    \includegraphics[scale=0.5]{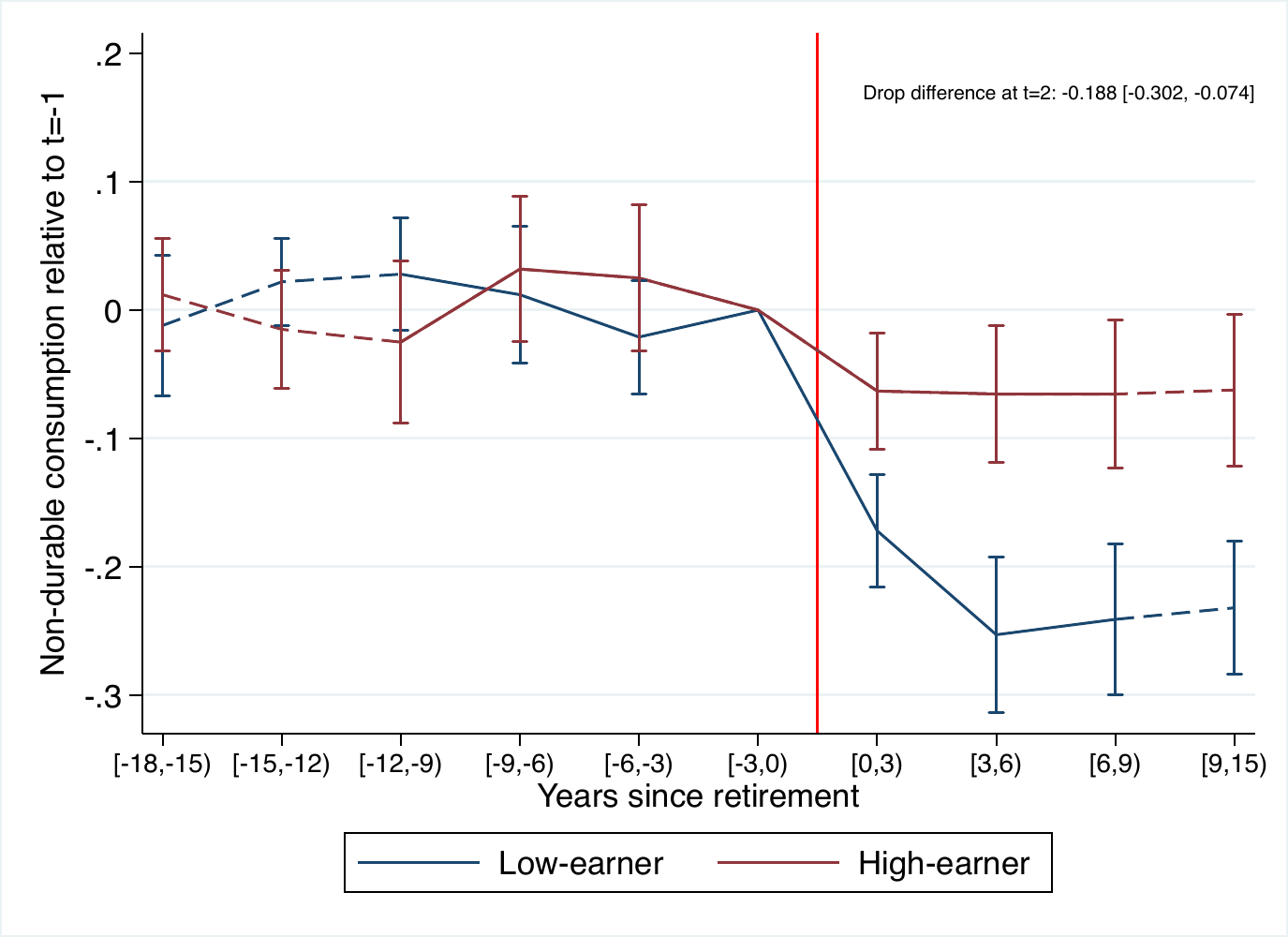}
    \caption{Consumption detrended}
    \label{fig7:b} 
    \vspace{4ex}
  \end{subfigure} 
    \begin{subfigure}[b]{0.5\linewidth}
    \centering
    \includegraphics[scale=0.5]{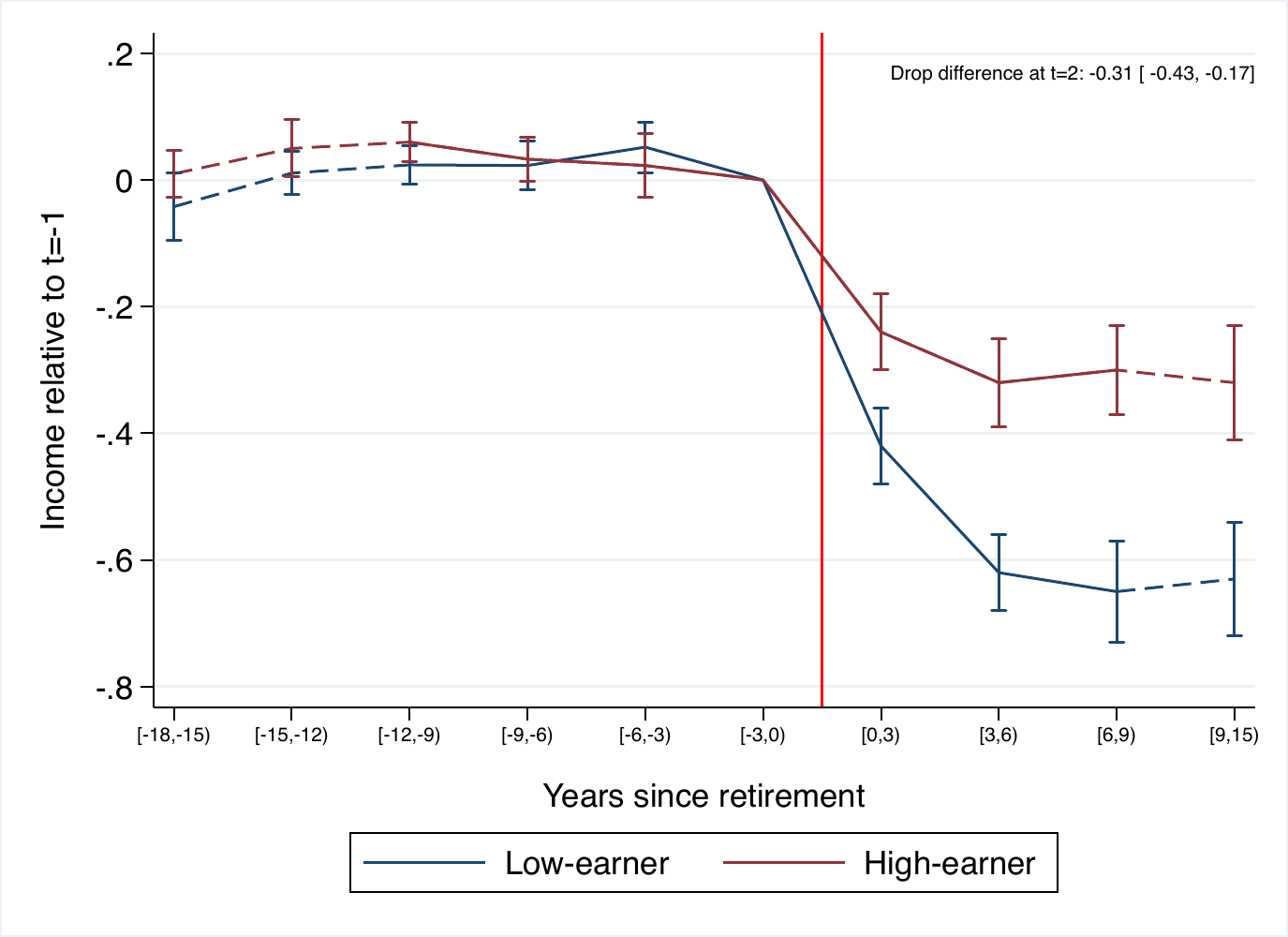}
    \caption{Income}
    \label{fig7:a} 
    \vspace{4ex}
  \end{subfigure}
  \begin{subfigure}[b]{0.5\linewidth}
    \centering
       \includegraphics[scale=0.5]{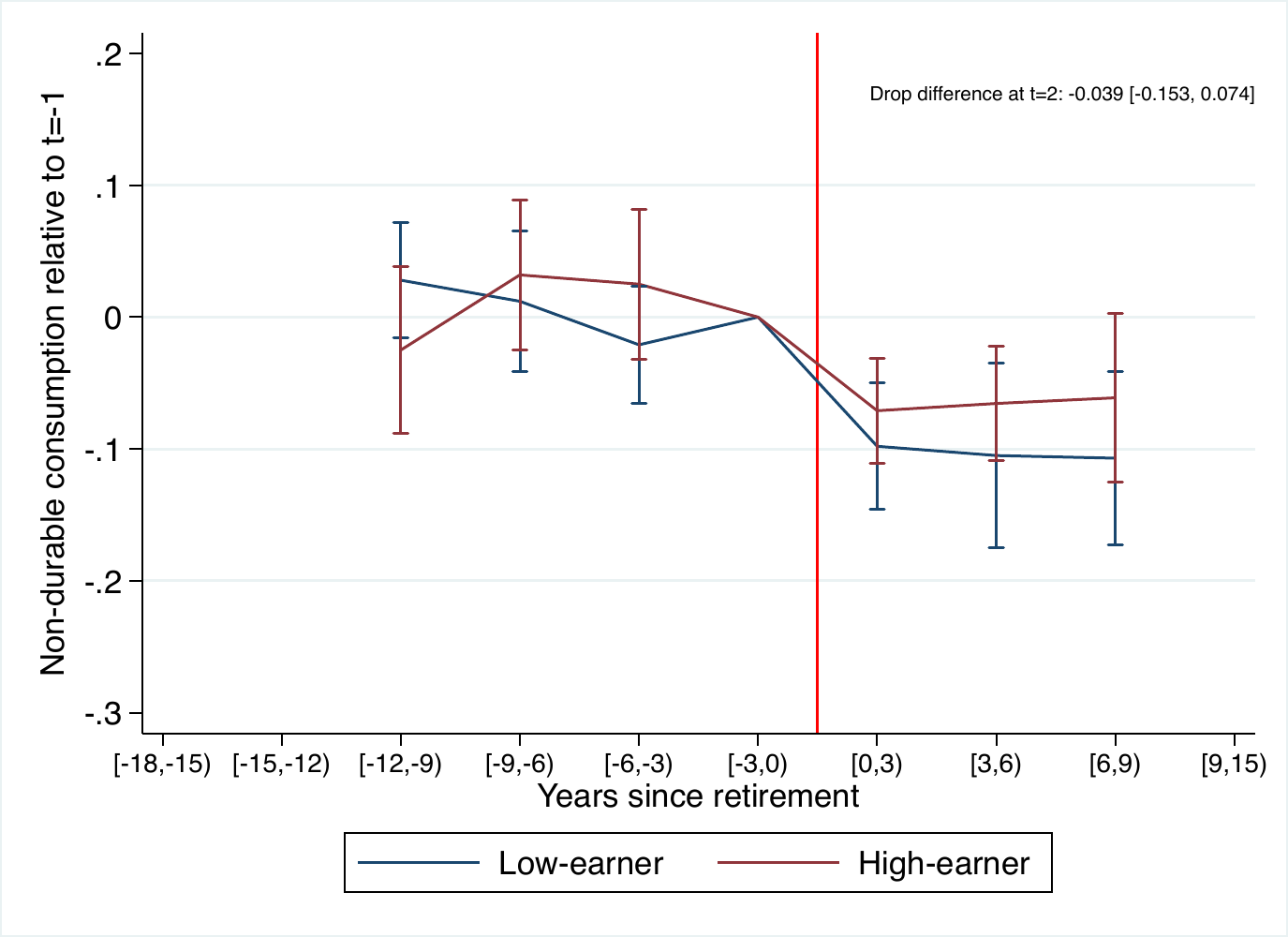}
    \caption{Consumption with income control}
    \label{fig7:b} 
    \vspace{4ex}
  \end{subfigure}  
  \caption{Impacts of retirement}
  \label{retirement_drop} 
      \begin{minipage}{\textwidth} 
    {\footnotesize \textit{Notes:} The figure shows event time coefficients estimated from equation (\ref{consumption_drop_eventstudy}) separately for workers with lifetime earnings above and below the average. In each panel, the outcome variable is controlled by year dummies, household composition, and cohort fixed effects. Panel (a), (b), and (d), the outcome variable is self-reported non-durable consumption, while in Panel (c) is self-reported income. In Panel (b) and (d), I further control for pre-retirement trends.  Each panel reports the difference in the drop measured at event time 2 (between 3-6 years after retirement). These statistics are estimated on an unbalanced sample of workers that retired between 2006–2019. Dashed lines extend the analysis using cross-section. The caps represent the 95 percent confidence intervals based on robust standard errors.
 \par}
    \end{minipage}
\end{figure}



\begin{figure}
    \centering
    \includegraphics[scale=.65]{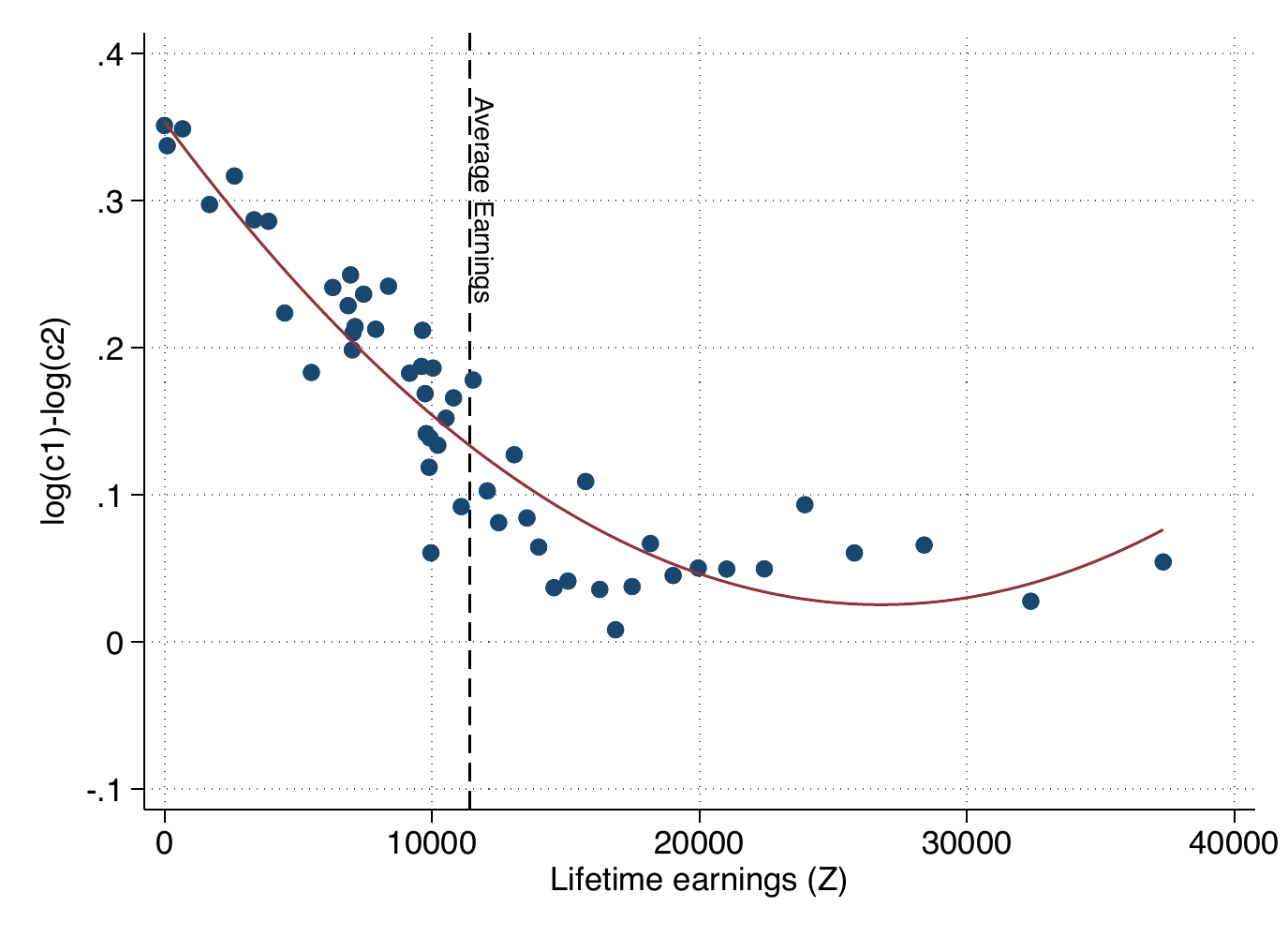}
    \caption{Relationship between lifetime payroll income and consumption drop income at retirement}
    \label{cons_drop_earnings}
          \begin{minipage}{\textwidth} 
    {\footnotesize \textit{Notes:} The scatter of the figure shows the average consumption
    drop for 52 equal-density bins of lifetime earnings and the line shows the quadratic linear relationship between consumption drop at retirement and lifetime earnings. Consumption drop is defined as the difference between the logs of active and retired consumption. There is one observation for every worker, where the active consumption is the availabe observation closest to retirement and the retired consumption is closest to event time 2 (between three and six years after retirement). The controls are one year dummies for the active and one year dummy for the retired consumption, cohort fixed effects, and household composition.
 \par}
    \end{minipage}
\end{figure}

\begin{figure}
    \centering
    \includegraphics[scale=.65]{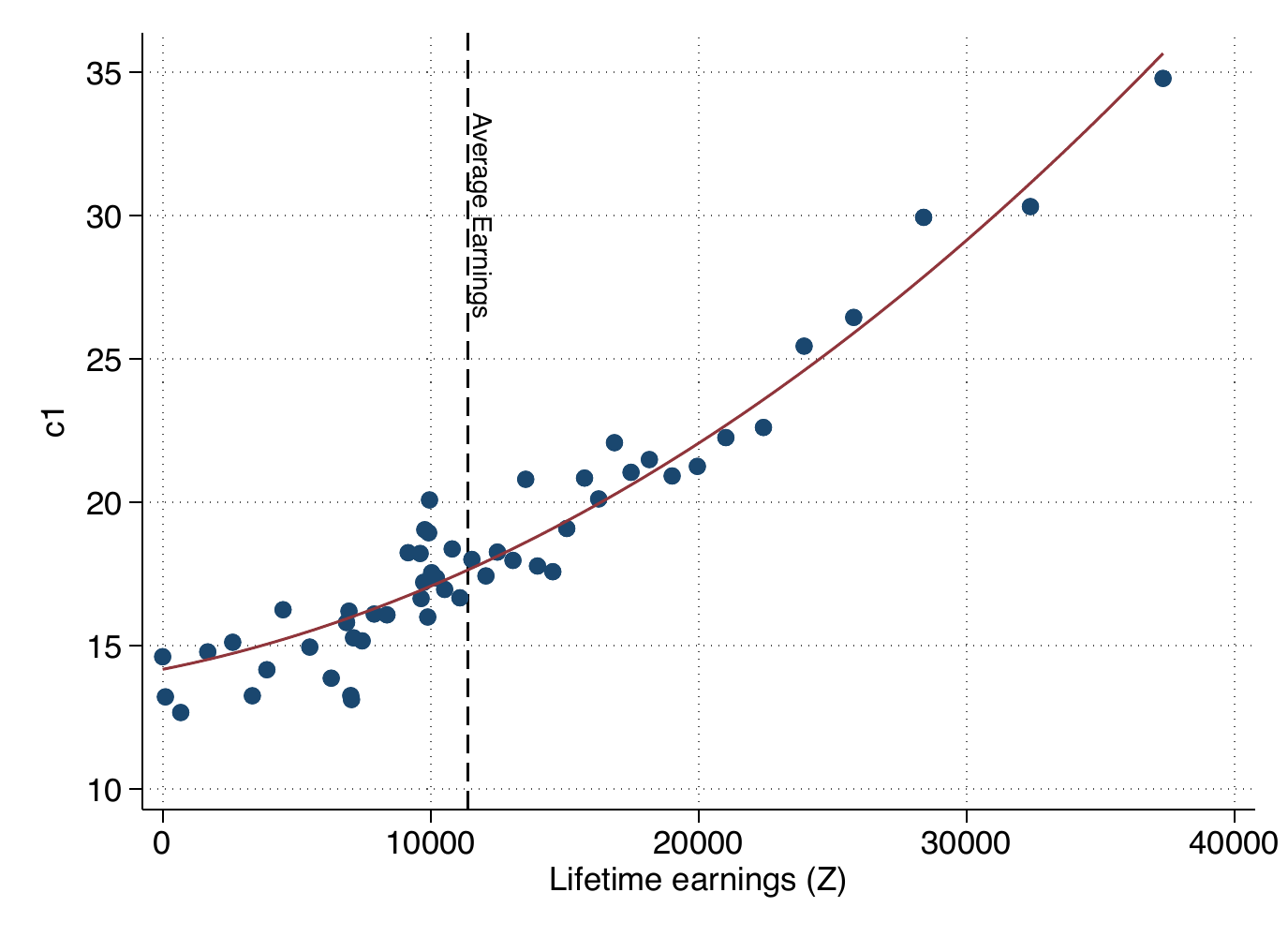}
    \caption{Relationship between lifetime earnings and active consumtpion}
    \label{act_cons_earnings}
    \begin{minipage}{\textwidth} 
        {\footnotesize \textit{Notes:} The scatter of the figure shows the average active consumption for 53 equal-density bins of lifetime earnings and the line shows the quadratic linear relationship between active consumption and lifetime earnings. There is one observation for every worker, where the active consumption is the available observation closest to retirement. The controls are year dummies, age and cohort fixed effects, and household composition.
 \par}
     \end{minipage}
\end{figure}

\begin{figure}
    \centering
\includegraphics[scale=0.3]{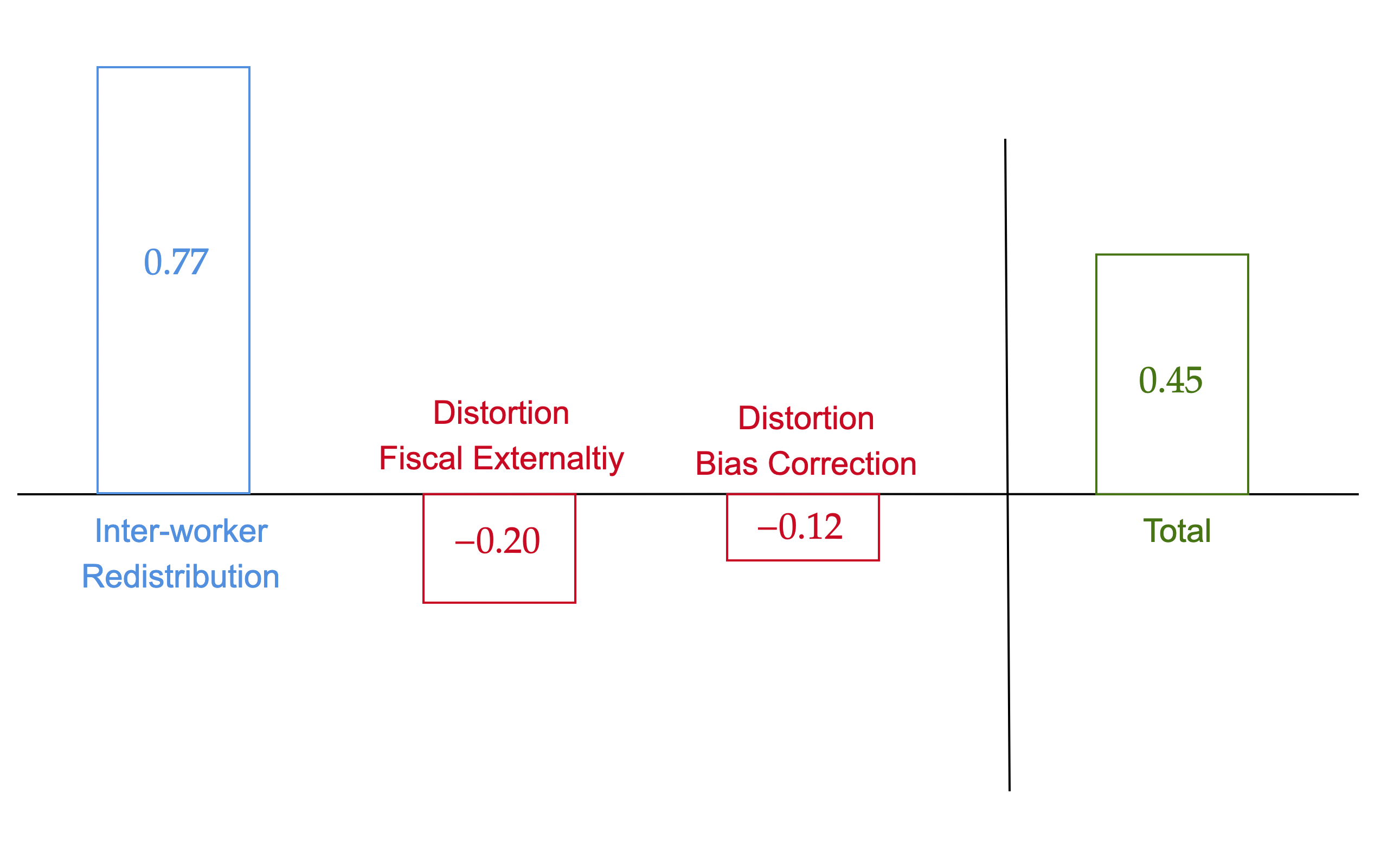}
    \caption{Decomposition of social gains of increasing benefits' progressivity }
    \label{decomp_progressivity}
\end{figure}

\begin{figure}
    \centering
\includegraphics[scale=0.65]{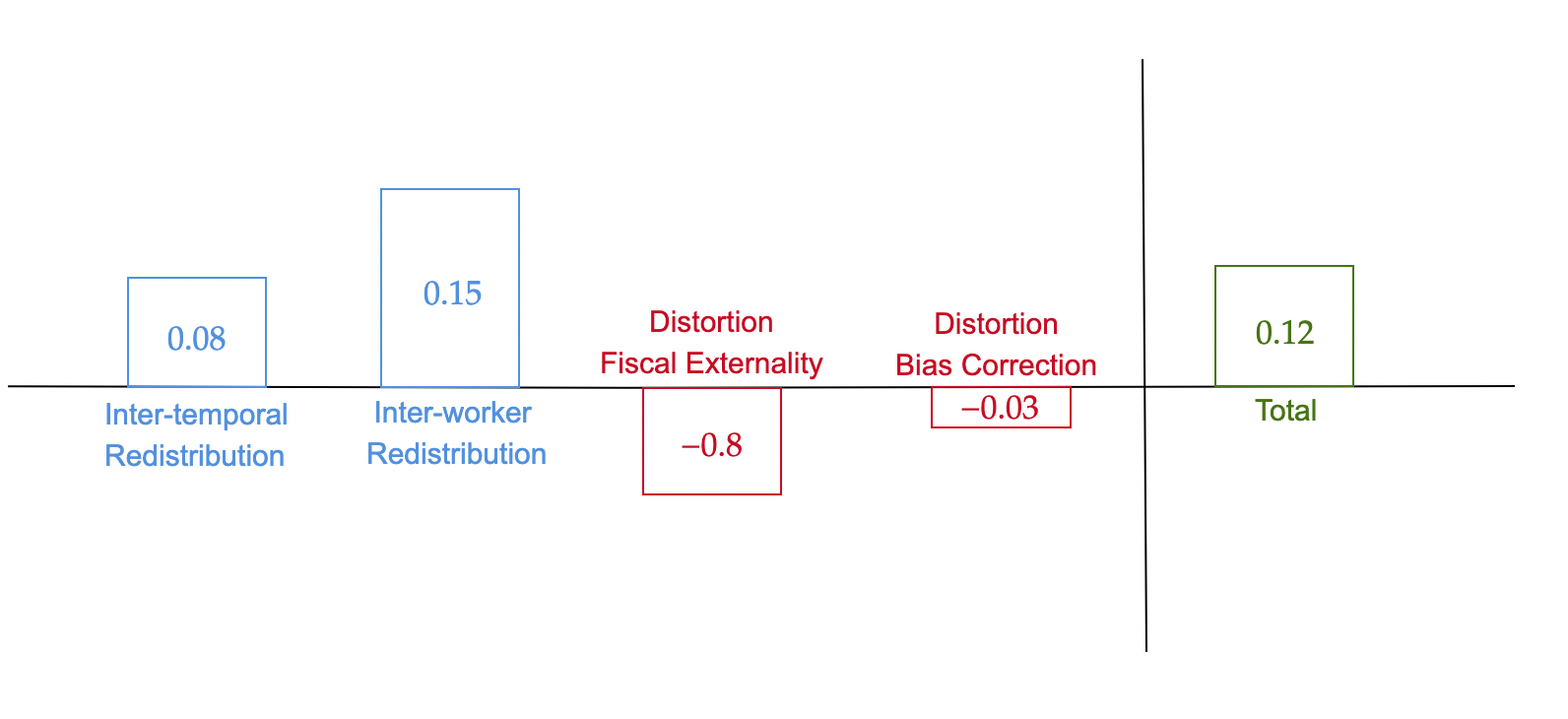}
    \caption{Decomposition of social gains of increasing pension contribution rate}
    \label{decomp_contribution}
\end{figure}

\begin{figure}
    \centering
    \includegraphics[scale=.5]{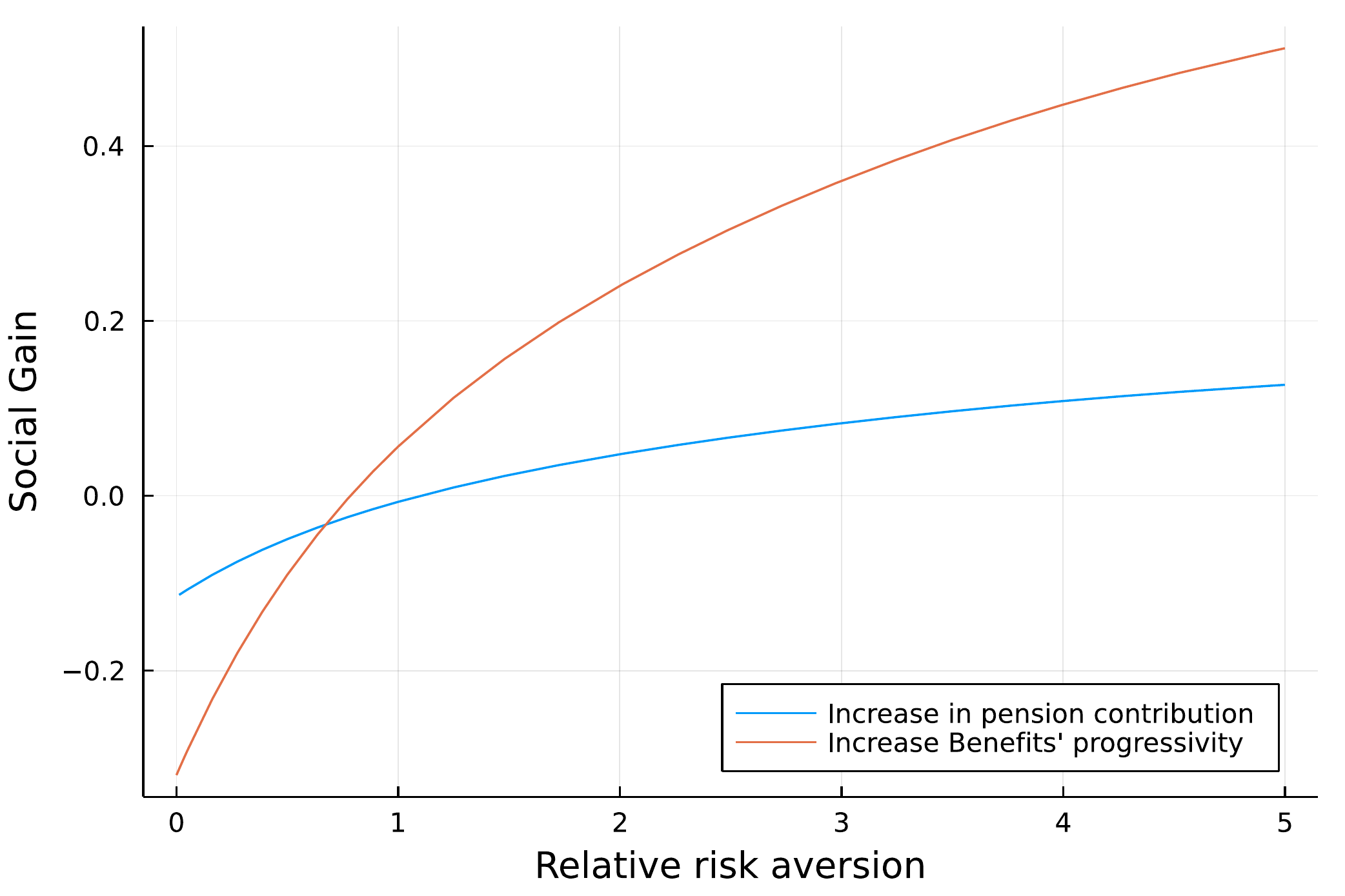}
    \caption{Relation between social gains of reforms and relative risk aversion}
    \label{comp_stat_rra}
\end{figure}

\begin{figure}
    \centering
    \includegraphics[scale=.4]{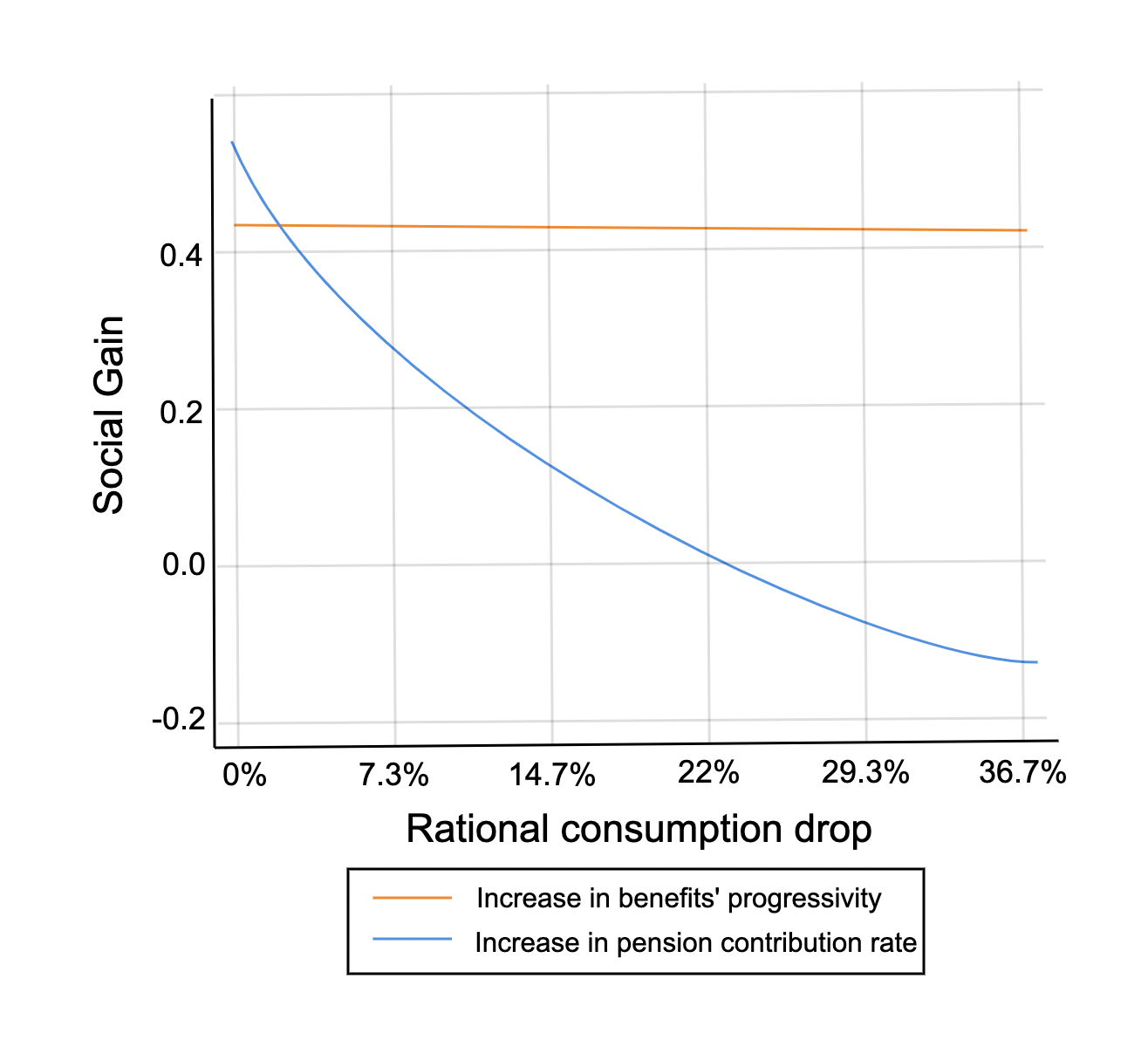}
    \caption{Relation between social gains of reforms and rational consumption drop at retirement}
    \label{comp_stat_consdrop}
\end{figure}

\begin{figure}
    \centering
    \includegraphics[scale=.4]{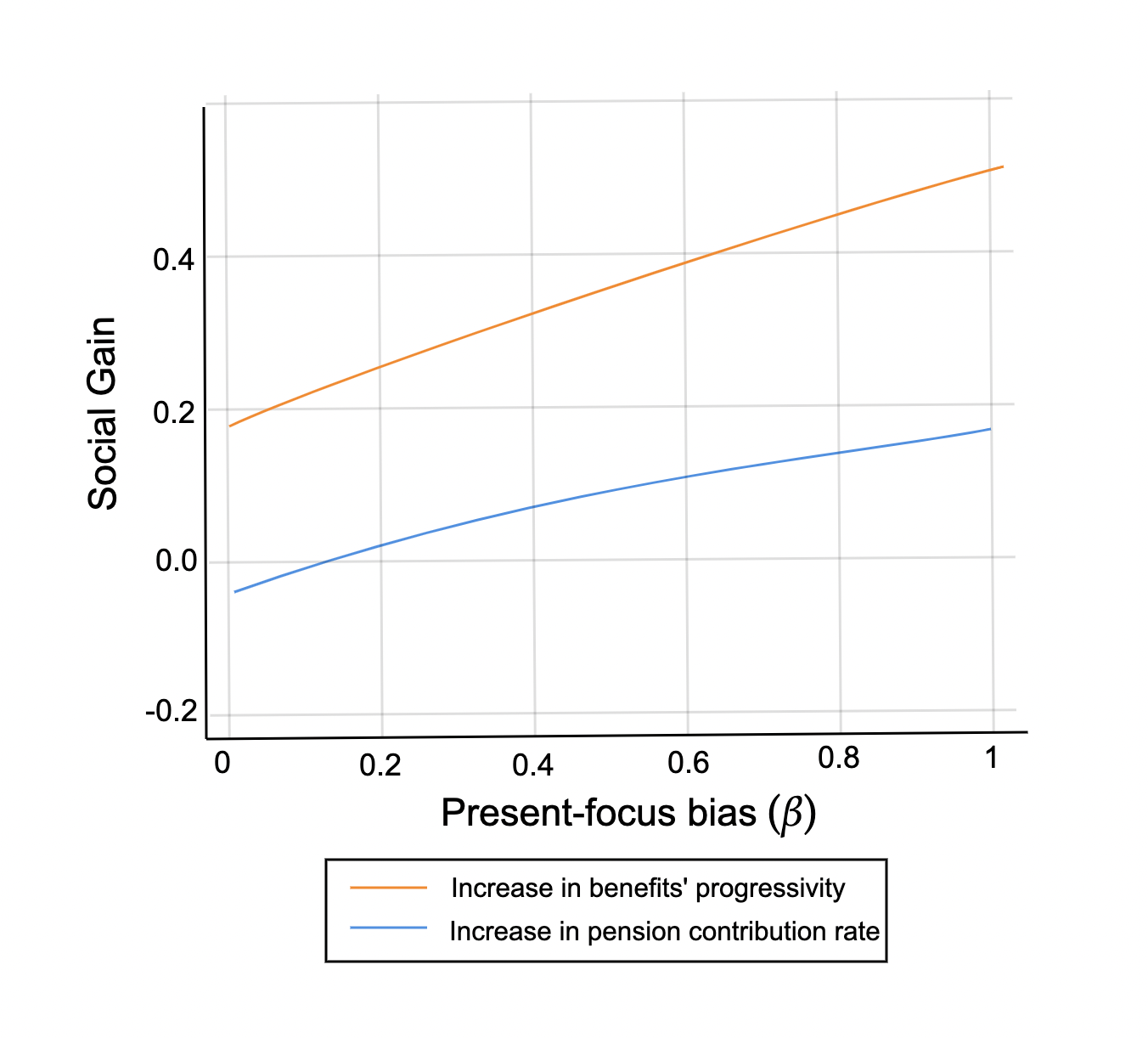}
    \caption{Relation between social gains of reforms and rational consumption drop at retirement}
    \label{comp_stat_bias}
\end{figure}

\begin{figure}
    \centering
    \includegraphics[scale=.6]{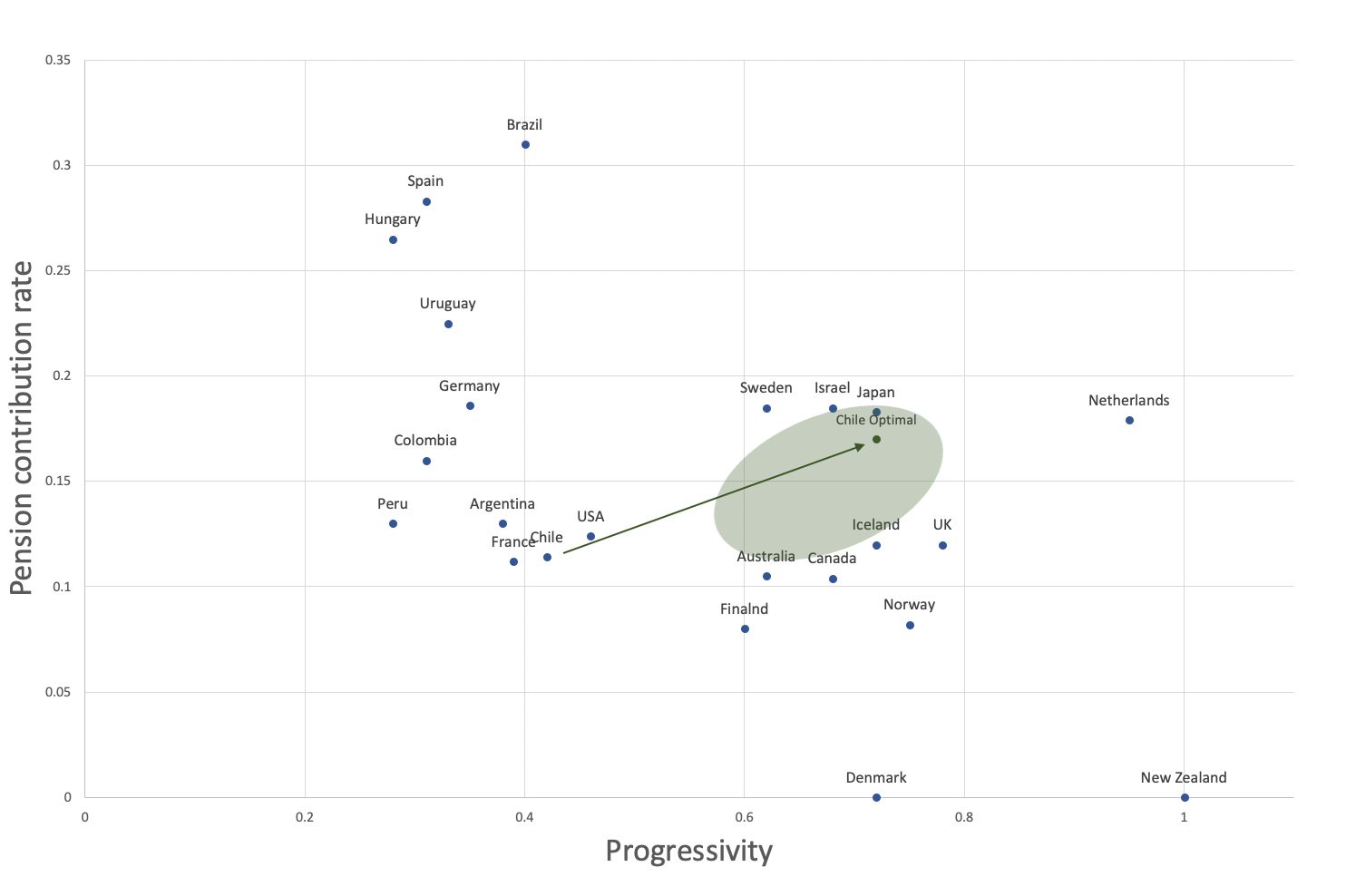}
    \caption{Pension contribution rate and progressivity of 23 public pension systems}
    \label{world_progr_contribution}
\end{figure}

\newpage

\FloatBarrier

\section*{Tables}

\begin{table}[h!]
    \centering
    \includegraphics[scale=0.5]{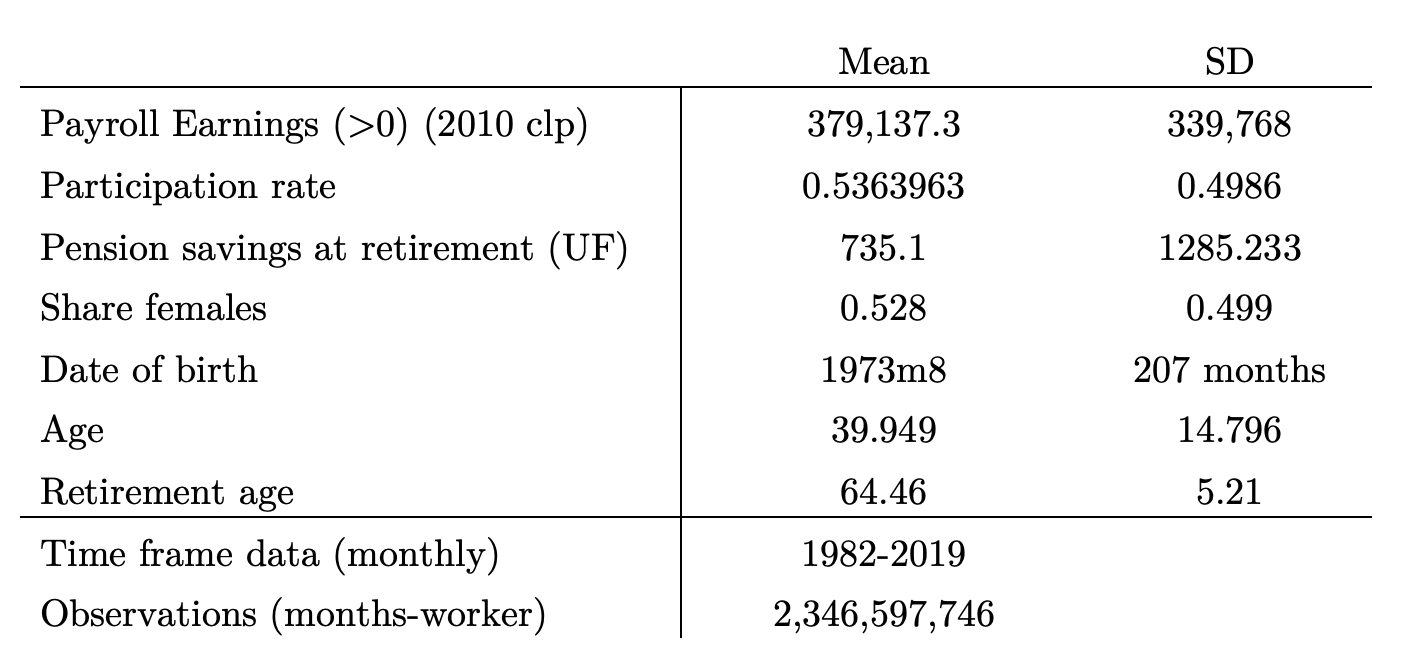}
    \caption{Summary statistics of pension system administrative data}
    \label{ETI_table}
    \medskip 
    \begin{minipage}{0.96\textwidth} 
    {\footnotesize \textit{Notes:}  \par} 
    \end{minipage}
    \label{sum_admin}
\end{table}

\begin{table}[h!]
    \centering
    \includegraphics[scale=0.5]{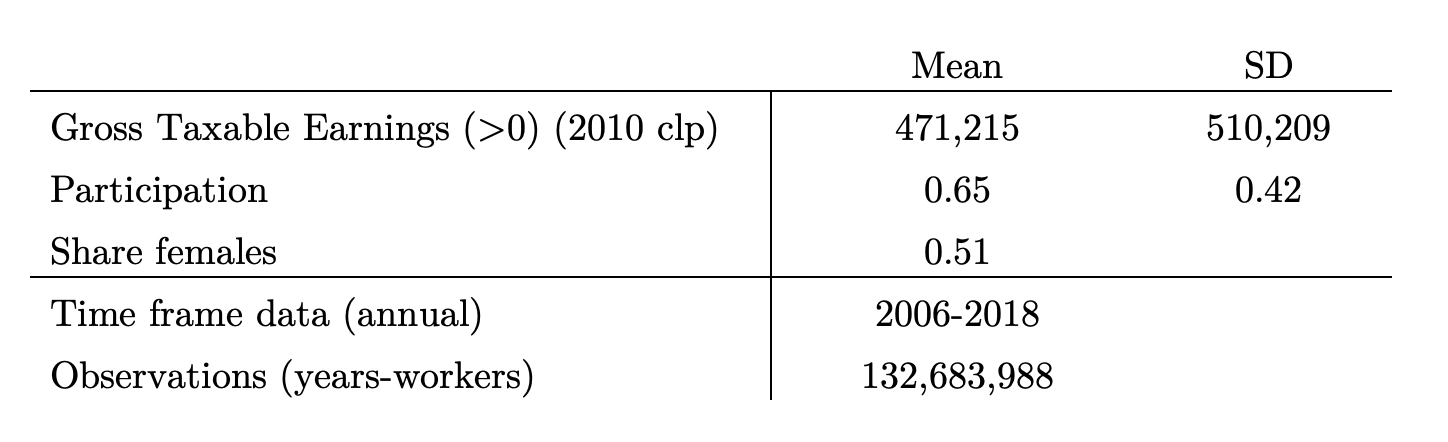}
    \caption{Summary statistics of admin tax data}
    \label{ETI_table}
    \medskip 
    \begin{minipage}{0.96\textwidth} 
    {\footnotesize \textit{Notes:}  \par} 
    \end{minipage}
   \label{sum_admin_tax} 
\end{table}

\begin{table}[h!]
    \centering
    \includegraphics[scale=0.5]{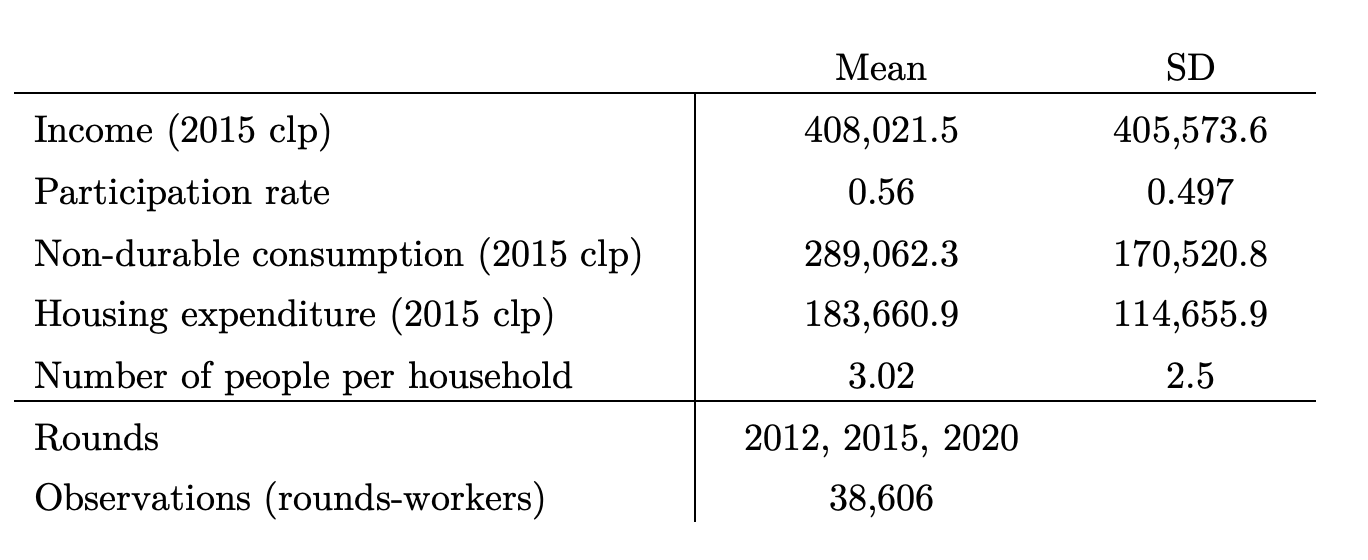}
    \caption{Summary statistics of survey data}
    \label{ETI_table}
    \medskip 
    \begin{minipage}{0.96\textwidth} 
    {\footnotesize \textit{Notes:}  \par} 
    \end{minipage}
    \label{sum_survey}
\end{table}


\begin{table}[h!]
    \centering
    \begin{tabular}{lccc} \hline
 & (1) & (2) & (3) \\
VARIABLES & log(pension) & log(pension) & log(pension) \\ \hline
 &  &  &  \\
log(pension savings GFC shock) & 0.552*** & 0.543*** & 0.528*** \\
 & [0.004] & [0.004] & [0.002] \\
 &  &  &  \\
Observations & 108,129 & 108,129 & 48,026 \\
R-squared & 0.684 & 0.602 & 0.645 \\
Age Controls & No & Yes & Yes \\
Other Controls & No & Yes & Yes \\
Crisis Definition & Short & Short & Long \\
Sample Definition & Wide & Wide & Narrow \\
 SE & Robust & Robust & Robust \\ \hline
\end{tabular}

    \caption{First-stage: Global Financial Crisis shock to future pension benefit}
    \label{first_stage_FC_shock}
    \medskip 
    \begin{minipage}{\textwidth} 
    {\footnotesize \textit{Notes:}  *, **, *** indicates significance at 90\%, 95\% and 99\% of confidence, respectively. Standard errors are inside [ ]. Estimation follows equation (\ref{first_stage_eq}). The dependent variable is the log of pension benefit at retirement adjusted by CPI. Controls include: age fixed-effects, gender and pre-GFC pension savings. \par} 
    \end{minipage}
\end{table}

\begin{table}[h!]
    \centering
    \begin{tabular}{lccc} \hline
 & (1) & (2) & (3) \\
VARIABLES & log(Income) & log(Income) & log(Income) \\ \hline
 &  &  &  \\
\% Shock to Pension & -0.155*** & -0.115*** & -0.121*** \\
 & [0.051] & [0.023] & [0.036] \\
 &  &  &  \\
Observations & 216,200 & 111,738 & 281,227 \\
R-squared & 0.550 & 0.579 & 0.562 \\
Controls & Yes & Yes & Yes \\
Crisis Definition & Long & Long & Short \\
Sample Definition & Wide & Narrow & Wide \\
 Clustered SE & Sex-Dob & Sex-Dob & Sex-Dob \\ \hline
\end{tabular}

    \caption{Taxable Income and future pension benefit.}
    \label{pen_shock}
    \medskip 
    \begin{minipage}{\textwidth} 
    {\footnotesize \textit{Notes:}  *, **, *** indicates significance at 90\%, 95\% and 99\% of confidence, respectively. Standard errors are inside [ ]. Estimation follows (\ref{second_stage_eq}). The dependent variable is the log of annual earnings adjusted by CPI. Controls include: age and year fixed-effects, gender, and pre-GFC pension saving. \par} 
    \end{minipage}
\end{table}

\begin{table}[h!]
    \centering
    \begin{tabular}{lccc} \hline
 & (1) & (2) & (3) \\
VARIABLES & Future Pension & Future Pension & Future Pension \\ \hline
 &  &  &  \\
GFC shock to pension savings   & 0.903*** & 0.899*** & 0.856*** \\
 & [0.008] & [0.008] & [0.004] \\
 &  &  &  \\
Observations & 108,129 & 108,129 & 48,026 \\
R-squared & 0.584 & 0.591 & 0.557 \\
Age Controls & No & Yes & Yes \\
Other Controls & No & Yes & Yes \\
Crisis Definition & Short & Short & Long \\
Sample Definition & Narrow & Narrow & Narrow \\
 SE & Robust & Robust & Robust \\ \hline
\end{tabular}

    \caption{First-stage: Global Financial Crisis shock to future pension benefit}
    \label{first_stage_FC_shock}
    \medskip 
    \begin{minipage}{0.96\textwidth} 
    {\footnotesize \textit{Notes:}  *, **, *** indicates significance at 90\%, 95\% and 99\% of confidence, respectively. Standard errors are inside [ ]. Estimation follows equation (\ref{first_stage_eq_MPC}). The dependent variable is pension benefit at retirement adjusted by CPI. Controls include: age fixed-effects, gender and pre-GFC pension savings. \par} 
    \end{minipage}
\end{table}

\begin{table}[h!]
    \centering
    \footnotesize
    \begin{tabular}{lccccc} \hline
 & (1) & (2) & (3) & (4) & (5) \\
VARIABLES & consumption & consumption & consumption &  consumption &  consumption \\ \hline
 &  &  &  &  &  \\
Pension benefit & 0.843*** & 0.798*** & 0.798*** & 0.798*** & 0.693** \\
 & [0.092] & [0.093] & [0.202] & [0.201] & [0.325] \\
 &  &  &  &  &  \\
Observations & 2,173 & 2,173 & 2,173 & 2,173 & 960 \\
Controls & No & Age-Gender & Age-Gender & Age-Gender & Age-Gender \\
SE & None & None & Robust & Cluster DoB & Cluster DoB \\
 Crisis definition & Long & Long & Long & Long & Short \\ \hline
\end{tabular}

    \caption{Marginal propensity to consume from pension benefits}
    \label{eld_disp_income_pension}
    \medskip 
    \begin{minipage}{0.96\textwidth} 
    {\footnotesize \textit{Notes:}  *, **, *** indicates significance at 90\%, 95\% and 99\% of confidence, respectively. Standard errors are inside [ ]. Estimation follows (\ref{second_stage_eq_mpc}). The dependent variable is the non-durable consumption in retirement. Controls include: age and year fixed-effects, gender, and pre-GFC pension saving. \par} 
    \end{minipage}
\end{table}

\begin{table}[h!]
    \centering
    \begin{tabular}{lccc} \hline
 & (1) & (2) & (3) \\
VARIABLES & Recipient & Recipient & Recipient  \\ \hline
 &  &  &    \\
Always receiver ($I_s=1$) & 0.581*** & 0.582*** & 0.563*** \\
 & [0.007] & [0.008] & [0.007]  \\
 &  &  &    \\
Observations & 46,851 & 46,851 & 22,151\\
R-squared & 0.632 & 0.693 & 0.602 \\
Controls & No & Yes & Yes \\
Bandwidth & 10\% & 10\% & 5\% \\
Clustered SE & DoB & DoB & DoB \\ \hline
\end{tabular}

    \caption{First-stage: Probability of receiving subsidy based on pre-subsidy Savings}
    \label{first_stage_receiver}
    \medskip 
    \begin{minipage}{\textwidth} 
    {\footnotesize \textit{Notes:}  *, **, *** indicates significance at 90\%, 95\% and 99\% of confidence, respectively. Standard errors are inside [ ]. Estimation follows equation (\ref{first_receiver}). The dependent variable is a dummy variable with value one for workers that receive subsidy at retirement. Always receiver ($I_s=1$) are those workers with pre-subsidy savings below lower bound ($\underline{a_i}$), while never receivers ($I_s=0$) are those with pre-subsidy savings above upper threshold ($\overline{a_i}$). Workers in between thresholds are not part in the sample. Bandwidth defines the distance of pre-subsidy savings to the threshold -- to lower for always receviers and to upper for always receiver-- that are used in the estimation. Only workers that are less than 6 years away to retirement at the subsidy introduction are used in the estimation. Controls are date of birth fixed effect; gender; and a 4th-degree polynomial of pre-retirement savings.
    \par} 
\end{minipage}
\end{table}

\begin{table}[h!]
    \centering
    \resizebox{12.2cm}{!}{
    \begin{tabular}{lccc} \hline
 & (1) & (2) & (3) \\
VARIABLES & log(earnings) & log(earnings) & log(earnings)  \\ \hline
 &  &  &    \\
 Post $\cdot$ Recipient$(=1)$& -0.0181*** & -0.0145*** & -0.0138***  \\
 & [0.002] & [0.002]  & [0.002]   \\
 &  &  &    \\
Observations & 281,106 &  281,106  & 140,553\\
R-squared & 0.053 & 0.182 & 0.192 \\
Controls & No & Yes & Yes \\
Bandwidth & 10\% & 10\% & 5\% \\
Clustered SE & DoB & DoB & DoB \\ \hline
\end{tabular}

  }
    \caption{Second-stage: pre-retirement earnings response to future subsidy}
    \label{second_receiver_table}
    \medskip 
    \begin{minipage}{\textwidth} 
    {\footnotesize \textit{Notes:}  *, **, *** indicates significance at 90\%, 95\% and 99\% of confidence, respectively. Standard errors are inside [ ]. Estimation follows equation (\ref{second_receiver}). Recipient comes from the first-stage estimated in Table \ref{first_stage_receiver}. Workers in between thresholds are not part in the sample. Bandwidth defines the distance of pre-subsidy savings to the threshold -- to lower for always receviers and to upper for always receiver-- that are used in the estimation. Only workers that are less than 6 years away to retirement at the subsidy introduction are used in the estimation. Controls are year, worker, and year fixed effects.    \par}  
    \end{minipage}
\end{table}

\begin{table}[h!]
    \centering
    \resizebox{12.2cm}{!}{
    \begin{tabular}{lcc} \hline
 & (1) & (2) \\
VARIABLES & Placebo lower & Placebo higher \\ \hline
 &  &    \\
Always receiver ($I_s=1$) & 0.000 & 0.000  \\
 & [0.007] & [0.008]   \\
 &  &     \\
Observations & 65,732 & 34,153 \\
R-squared & 0.321 & 0.345 \\
Placebo PMAS & 0.8$\cdot$ PMAS & 1.2$\cdot$ PMAS  \\
Controls & Yes & Yes  \\
Bandwidth & 10\% & 10\% \\
Clustered SE & DoB & DoB\\ \hline
\end{tabular}

}
    \caption{Placebo first-stage: Probability of receiving subsidy based on pre-subsidy savings}
    \label{first_stage_receiver_placebo}
    \medskip 
    \begin{minipage}{\textwidth} 
    {\footnotesize \textit{Notes:}  *, **, *** indicates significance at 90\%, 95\% and 99\% of confidence, respectively. Standard errors are inside [ ]. Estimation follows equation (\ref{first_receiver}). The dependent variable is a dummy variable with value one for workers that receive subsidy at retirement. Column (1) shows the first-stage for a placebo PMAS of 0.8 times the actual PMAS, while column (2) does it for a placebo PMAS of 1.2 the actual one. Always below ($I_s=1$) are those workers with pre-subsidy savings below lower bound ($\underline{a_i}$), while always above ($I_s=0$) are those with pre-subsidy savings above upper threshold ($\overline{a_i}$), where the bounds are defined using the placebo PMAS. Workers between thresholds are not part in the sample. Bandwidth defines the distance of pre-subsidy savings to the threshold -- to lower for always receviers and to upper for always receiver-- that are used in the estimation. Only workers that are less than 6 years away to retirement at the subsidy introduction are used in the estimation. Controls are date of birth fixed effect; gender; and a 4th-degree polynomial of pre-retirement savings.
    \par} 
\end{minipage}
\end{table}

\begin{table}[h!]
    \centering
    \resizebox{12.2cm}{!}{
    \begin{tabular}{lcc} \hline
 & (1) & (2) \\
VARIABLES & log(earnings) & log(earnings)  \\ \hline
 &  &     \\
 Post $\cdot$ Placebo Recipient & -0.001 & -0.003   \\
 & [0.002] & [0.008]    \\
 &  &     \\
Observations & 394,392 &  204,918  \\
R-squared &  0.152 & 0.142 \\
Controls & Yes & Yes \\
Bandwidth & 10\% & 10\% \\
Clustered SE & DoB & DoB \\ \hline
\end{tabular}

}
    \caption{Placebo second-stage: Probability of receiving subsidy based on pre-subsidy savings}
    \label{second_stage_receiver_placebo}
    \medskip 
 \begin{minipage}{\textwidth} 
    {\footnotesize \textit{Notes:}  *, **, *** indicates significance at 90\%, 95\% and 99\% of confidence, respectively. Standard errors are inside [ ]. Estimation follows equation (\ref{second_receiver}). Placebo recipient comes from the first-stage estimated in Table \ref{first_stage_receiver}, that is, I assume that worker below the lower bound have a 58\% larger probability of receiving the subsidy at retirement. The sample include workers with pre-subsidy earnings larger than $0.9 \cdot \underline{a}_i$ and smaller than  $1.1 \cdot \overline{a}_i$,  not in between thresholds, and are less than 6 years away to retirement age at the subsidy introduction. Controls are year, worker and age fixed effects.    \par}  
    \end{minipage}
\end{table}

\begin{table}[h!]
    \centering
    \begin{tabular}{lcccc} \hline
 & (1) & (2) & (3) & (4) \\
\hline
 &  &  &  &  \\
ETI & 0.421*** & 0.383*** & 0.383*** & 0.383** \\
 & [0.066] & [0.043] & [0.039] & [0.16] \\
 &  &  &  &  \\
Observations & 1,298,805 & 1,298,802 & 1,298,802 & 1,298,802 \\
R-squared & 0.012 & 0.572 & 0.572 & 0.572 \\
Switchers & No & No & No & No \\
Controls & No & Yes & Yes & Yes \\
PFA Trend & No & Yes & Yes & Yes \\
 Clustered SE & None & None & PFAxDate & Worker \\ \hline
\end{tabular}

    \caption{Elasticity of Taxable Income with respect to Net-of-Tax Rate.}
    \label{ETI_table}
    \medskip 
    \begin{minipage}{\textwidth} 
    {\footnotesize \textit{Notes:}  *, **, *** indicates significance at 90\%, 95\% and 99\% of confidence, respectively. Standard errors are inside [ ]. Estimation follows equation (\ref{dif_dif}). The dependent variable is the log of earnings. Controls are time and pension fixed effects, age, gender, earnings and participation before treatment, and pension savings before treatment. Workers that have ever switched between pension fund (7.8\% of the sample). \par} 
    \end{minipage}
\end{table}

\begin{table}[h!]
\centering
\begin{tabular}{ |p{5cm}||p{1.7cm}|p{2.7cm}|p{5.6cm} |} \hline  
 & Model's Notation & Value & Source \\
 \hline
  Taxable earnings elasticity to payroll tax  & \centering $\varepsilon_{(1-\tau)\overline{z}}$  & \centering 0.38 \ \ \ $[0.27, 0.49]$ & Reduction on Planvital's pension fund management fee in 2014  \\
   \hline
     Taxable earnings elasticity to benefits progressivity tax  & \centering $\varepsilon_{(1-\phi)\overline{z}}$  & \centering 0.22  \ \ \ \ \  $[0.10, 0.32]$ & Pension subsidy introduction in 2008  \\
   \hline
    Taxable earnings elasticity to pension benefits  & \centering $\varepsilon_{b\overline{z}}$  &  \centering 0.11 \ \ \ \ $[0.08, 0.14]$ & Pension savings' exposition to Global financial crisis returns  \\
    \hline
    Marginal propensity to consume from pension benefit  & \centering $\mu$  &  \centering 0.79 \ \ \ \ $[0.55, 1.2]$ & Pension savings' exposition to Global financial crisis returns  \\
   \hline
    Relative risk aversion  & \centering $\gamma$  & \centering $4$ & \cite{landais2021value} \\
   \hline
       Retirement consumption preferences  & \centering $\theta$  & \centering $\{0.62,0.81\}$  & \cite{battistin2009retirement} \\
   \hline
             Lifetime payroll earnings  & \centering $z_{i}$  &  & Admin data \\
   \hline
          Active period non-durable consumption   & \centering $c_{i1}$  &  & Survey data \\
   \hline
        Retired period non-durable consumption   & \centering $c_{i2}$  &  & Survey data \\
   \hline
           Pension investment return   & \centering $R$  &  \centering 1.042 & Historical 10-year geometric average return \\
   \hline
            Discount factor   & \centering $\delta$  & \centering $R^{-1}=0.96$   & Assumption \\
    \hline
            Hyperbolic discounting   & \centering $\beta$  & \centering $0.82$ \\ $[0.74, 0.9]$   & \cite{hyperbolic_meta} \\
   \hline
\end{tabular}
    
        \medskip 
    \begin{minipage}{0.96\textwidth} 
    {\footnotesize \textit{Notes:} 95\% of confidence interval is inside brackets ($[ \ ]$). \par} 
    \end{minipage}
    \caption{Model parametrization}
        \label{model_parameters}
\end{table}


\newpage

\FloatBarrier

\section*{Appendix}

\renewcommand{\thesection}{A}
\section{Proofs}

\begin{proof}[\textbf{Proof Lemma 1}]

In an interior solution with preferences separable between consumption and taxable earnings, the optimal taxable earning ($z$) is given by:

$$w'(z)+(1-\kappa-\tau)u'+\kappa R (1-\phi)v'=0$$

Differentiating with respect to $\tau, \kappa, \phi$ we get:

\begin{equation} \label{tax}
    \frac{\partial z}{\partial \tau}=\frac{z(1-\kappa-\tau)u''+u'-((1-\kappa-\tau)u''g'(\chi)+\kappa R (1-\phi)v'' h'(x))\frac{\partial \chi}{\partial \tau}}{w''(z)+(1-\kappa-\tau)^2 u''+\kappa^2 R^2 (1-\phi)^2 v''}
\end{equation}

\begin{equation} \label{contribution}
    \frac{\partial z}{\partial \kappa}=\frac{z(1-\kappa-\tau)u''+u'+\kappa R^2(1-\phi)^2 v'' z-R(1-\phi)v'-((1-\kappa-\tau)u''g'(\chi)+\kappa R (1-\phi)v'' h'(\chi))\frac{\partial \chi}{\partial \kappa}}{w''(z)+(1-\kappa-\tau)^2 u''+\kappa^2 R^2 (1-\phi)^2 v''} 
\end{equation}

\begin{equation} \label{progressivity}
    \frac{\partial z}{\partial (1-\phi)}=\frac{-\kappa^2 R^2(1-\phi) v'' z+\kappa Rv'-((1-\kappa-\tau)u''g'(\chi)+\kappa R (1-\phi)v'' h'(\chi))\frac{\partial \chi}{\partial (1-\phi)}}{w''(z)+(1-\kappa-\tau)^2 u''+\kappa^2 R^2 (1-\phi)^2 v''} 
\end{equation}

Conjecture: 

\begin{align}
    \frac{\partial z}{\partial \kappa}\kappa-\frac{\partial z}{\partial (1-\phi)}=\frac{\partial z}{\partial \tau}\kappa \notag \\
    \frac{\partial \chi}{\partial \kappa}\kappa-\frac{\partial \chi}{\partial (1-\phi)}=\frac{\partial \chi}{\partial \tau}\kappa  \notag
\end{align}

Using the conjecture in equations (\ref{tax}), (\ref{contribution}) and (\ref{progressivity}) I show that they are satisfied. 
\end{proof}

The government's problem first order conditions are:

\begin{align}
\label{foc_pensiontax}
    \frac{dW}{d\phi}= &\underbrace{\kappa R^P \int_i U_{c_{i2}} (z_i-\overline{z}) di}_\text{Inter-worker redistribution} - \underbrace{ \overline{ U_{c_{i2}}}\left(\kappa R^P+\frac{\tau}{\phi}\right)\frac{d \overline{z}}{d\phi} }_\text{Moral Hazard externality} \notag\\ 
    &\underbrace{+\int_i\left(U_{c_{i1}}-\widehat{U}_{c_{i1}}\right) (1-\kappa -\tau)\frac{\partial z_i}{\partial \phi}di +\int_i\left(U_{c_{i2}}-\widehat{U}_{c_{i2}}\right) (1-\phi)\kappa R \frac{\partial z_i}{\partial \phi}di}_\text{Labor supply internality} \notag \\ 
     &\underbrace{+\int_i\left(U_{c_{i1}}-\widehat{U}_{c_{i1}}\right) \frac{\partial g_i}{\partial \chi_i}\frac{\partial \chi_i}{\partial \phi}di +\int_i\left(U_{c_{i2}}-\widehat{U}_{c_{i2}}\right) \frac{\partial h_i}{\partial \chi_i}\frac{\partial \chi_i}{\partial \phi}di}_\text{Consumption timing internality}=0
\end{align}

\begin{align} \label{foc_contribution_na}
\frac{dW}{d\kappa}= &\underbrace{\int_i (U_{c_{i2}}R^P-U_{c_{i1}})z_i d\mu}_\text{Intra-worker redistribution} - \underbrace{R^P\phi \int_{i} U_{c_{i2}} (z_i-\overline{z}) d\mu}_\text{Inter-worker redistribution}-\underbrace{\overline{ U_{c_{2}}}\left(R^P \phi+\frac{\tau}{\kappa}\right)\left(\varepsilon_{(1-\kappa)\overline{z}}+\varepsilon_{b\overline{z}}\right) R^P \overline{z}}_\text{Moral hazard externality} \notag \\
&\underbrace{+\int_i\left(U_{c_{i1}}-\widehat{U}_{c_{i1}}\right) (1-\kappa -\tau)\frac{\partial z_i}{\partial \kappa}di +\int_i\left(U_{c_{i2}}-\widehat{U}_{c_{i2}}\right) (1-\phi)\kappa R \frac{\partial z_i}{\partial \kappa}di}_\text{Labor supply internality} \notag  \\ 
     &\underbrace{+\int_i\left(U_{c_{i1}}-\widehat{U}_{c_{i1}}\right) \frac{\partial g_i}{\partial \chi_i}\frac{\partial \chi_i}{\partial \kappa}di +\int_i\left(U_{c_{i2}}-\widehat{U}_{c_{i2}}\right) \frac{\partial h_i}{\partial \chi_i}\frac{\partial \chi_i}{\partial \kappa}di}_\text{Consumption timing internality}=0
\end{align}

where an overline variable indicates the population average.   \\ 

\noindent\textbf{Discussion-.} The optimality conditions have five elements. The first two elements are the social welfare gains of doing income redistribution across time, from active to retired periods, and across workers at retirement. These elements are captured by the intra-worker and inter-worker redistribution terms, respectively. The third elements is given by the externality on the fiscal budget that is generated by the behavioral response to the reforms. This is captured by of the moral hazard externality term. The last two elements are the cost on workers welfare (internalities) generated by the behavioral response to the reforms, given that workers choose taxable earnings and consumption allocation unoptimally. The first internality is generated by workers misoptimal taxable earnings generation, and the second one by the misoptimal consumption timing. These two internalities are captured by the labor supply internality and consumption timing internality terms, respectively. 

The intra-worker redistribution is an extension of \cite{baily1978some} and \cite{chetty2006general} social insurance formula, where the extension is to incorporate  heterogeneity on workers' productivity and ability to prepare for retirement. In the other hand, the inter-worker redistribution is a variation of the linear-income tax, which has been extensively studied by the literature.\footnote{\cite{sheshinski1972optimal}, \cite{atkinson1995public}, \cite{itsumi1974distributional}, \cite{stern1976specification}, \cite{dixit1977some}, \cite{helpman1978optimal}, \cite{deaton1983explicit}, and \cite{tuomala1985simplified}, among others.} The variation is that the redistribution of the taxed income is given back at the retirement period, therefore, the retirement preparedness and it's relation with taxable earnings plays a role. As larger is the relationship between retirement preparation and lifetime taxable earnings, larger is the inter-worker redistribution value of pensions. I show the equivalence between the inter-temporal and inter-workers redistribution of equations (\ref{foc_pensiontax}) and (\ref{foc_contribution_rate}) with the social insurance and optimal linear-income tax literature in the Appendix \ref{appendix_B}.

The taxable earning response to the reforms generates externalities, through the fiscal budget, and internalities, through worker welfare. These behavioral responses to the reforms are driven by two forces. First, the reforms have a direct effect on the marginal value of taxable earnings. An increase in the pension contribution reduces the net-of-tax rate while active and increases the effect of taxable earnings on future pension payment. Similarly, an increase in the benefits' progressivity, reduces the benefit-contribution link of the future pension payment. The second force is given by the income effect generated by the lump-sum transfer at retirement. An increase in pension contribution rate or in the benefits' progressivity rate ($\phi$), increases the lump-sum transfer at retirement. The welfare cost is driven by the overall response of taxable earnings to the reforms. 

In sum, the social desirability of an increase in pension progressivity is a trade-off between the inter-temporal and inter-worker income redistribution, and the externalities and internalities generated by the behavioral responses to reform.

\subsection{Framework to Data}

This subsection connects the elements of the model with moments estimable in the data. Behavioral responses and marginal value of consumption. 

I first show, in lemma 1, that the response of taxable earnings to an increase in the pension contribution rate can be decomposed on how taxable earnings respond to payroll taxes and the benefits progressivity. I have variation to identify these two elasticities, and therefore the response to pension contribution rate. 

\begin{lemma}
Let $m=\kappa(1-\phi)$ be the slope of the benefits and lifetime earnings function. Then, 
the behavioral response of earnings to a change in pension contribution rate ($\kappa$) is given by:
$$\frac{dz}{d \kappa}=\underbrace{\frac{\partial z}{\partial \tau}}_\text{$\Delta$ Net-of-tax rate}\underbrace{+\frac{\partial z}{\partial m}(1-\phi)}_\text{$\Delta$ Slope}+\underbrace{\frac{\partial z}{\partial b}\phi \overline{z}R}_\text{$\Delta$ Intercept}$$
and to a change in benefits' progressivity ($\phi$) is: 

$$\frac{dz}{d \phi}=\underbrace{-\frac{\partial z}{\partial m}\kappa}_\text{$\Delta$ Slope}+\underbrace{\frac{\partial z}{\partial b}\kappa \overline{z}R}_\text{$\Delta$ Intercept}$$
\end{lemma}

\begin{proof}
In the appendix.
\end{proof}

The intuition of this result is on Figure \ref{graph_reform}. A change in benefits' progressivity has two effects: it changes the lump-sum transfer $b$ (intercept in Figure \ref{graph_reform} panel (a)) and the relationship between lifetime earnings and benefit $m$ (slope in Figure \ref{graph_reform} panel (a)). Therefore, the response to this reforms depends on how earnings respond to these two changes and the magnitude of the changes generated by the reform. On the other hand, a change in pension contribution rate also changes the lump-sum transfer and the slope of benefit-earnings relationship, although in different magnitude than benefits' progressivtiy. Additionally, a change in the pension contribution rate also generates a change in the net-of-tax rate when worker is active, which is is equivalent to a change in the payroll tax.





I make two assumption over workers preferences in order to simplify the optimality conditions. I assume that workers preferences are separable, and that the behavioral bias has the particular form of present-focused bias. With this assumption I can build a lower bound for the social gains of reforming pension contribution rate and benefits' progressivity.  

\begin{assumption}[\textbf{Separable preferences}]
Preferences are separable:
$\frac{\partial^2 U}{\partial k \partial l}=0 \ \ for \ k,l \in \{c_1,c_2,z,\chi\} \ and \  k\neq l$
\end{assumption}

This assumption is standard in the social insurance literature. 

\begin{assumption}
Workers have present-bias preferences with respect to consumption. While they are active, they discount retired period consumption by the factor $\delta$: $$\frac{\partial \widehat{U}}{\partial c_2}= \beta \frac{\partial U}{\partial c_2}$$
where $\widehat{U}$ is worker perceived preferences and $U$ is central planner preference.
\end{assumption}

These two expressions are

Finally, I make two assumptions over workers' valuation of consumption. These assumption are standard in the social insurance literature and are aim to simplify the connection between the marginal utility of consumption and consumption data. 

I do this by doing a Taylor expansion around retired consumption.
\begin{lemma}
Let $\overline{c}$ be the consumption such that $U_{c_2}(\overline{c})=\int_i U_{c_2}(c_{i2}) di$. Then:
\begin{align*}
    &Cov\left[\frac{d_i}{\overline{U_{c_2}}},zi \right] \approx Cov\left[\gamma \frac{(\theta c_{i1}-c_{i2})}{\overline{c}}, z_i \right]\\
    &Cov\left[\frac{U_{c_2}(c_{i2})}{\overline{U_{c_2}}},z_i\right] \approx Cov\left[\gamma \frac{c_{i2}}{\overline{c}}, z_i \right]= Cov\left[\gamma \frac{c_{i1}}{\overline{c}}, z_i \right] +  Cov\left[\gamma \frac{c_{i2}-c_{i1}}{\overline{c}}, z_i \right]
\end{align*}
\end{lemma}
\begin{proof}
In the appendix.
\end{proof}

This lemma captures that reforms' transfer receiver are determined by the lifetime earnings. With an increase in pension contribution rate, workers are forced to save for retirement. The amount forced to save is given by payroll earnings, while the welfare value of that transfer is given by the distance to the Euler equation. Therefore, the inter-temporal value is inversely proportional to $Cov[]$. If low income workers are those that further away from the Euler equation, then forced savings will be bringing income to retirement to those that need it less. In the other hand, the amount transfer between workers is given by the distance to mean earnings and the value of that transfer is given by the marginal utility of consumption at retirement. The retirement consumption can be mechanically decomposed in consumption while active and the consumption drop at retirement. Therefore, the inter-worker value is directly proportional to the covariance between consumption drop and lifetime earnings.

\renewcommand{\thesection}{B}

\section{Decomposition of reforms' redistribution value} \label{appendix_B}

\noindent \textbf{Inter-temporal.-} The inter-temporal redistribution value is an extension of the \cite{baily1978some}-\cite{chetty2006general} formula. To see that, let's start with the case where there are no heterogeneity between workers. Then, the social gains of increasing pension contribution rate are given by:

\begin{align*}
   \frac{dW}{d\kappa} = U_{c_2}\frac{(U_{c_2}-U_{c_1})}{U_{c_2}} z -U_{c_2}\frac{\tau}{\kappa}(\varepsilon_{(1-\kappa)z})z R
\end{align*}

Expressing it as retirement money-metric, and using assumptions 1-4 in a second degree Taylor approxiamtion we get: 

\begin{align*}
   \frac{dW}{d\kappa} \approx \left(\gamma \frac{\Delta c}{c_2}+\theta\right) z+  \frac{\tau}{\kappa}\varepsilon_{(1-\kappa)z} z R
\end{align*}

which is equivalent to the formula under state-dependent utility case in \cite{chetty2013social}.

By adding heterogeneity on productivity, we extend this formula and get:

\begin{align*}
   \frac{dW}{d\kappa} \approx \left(\gamma \frac{\Delta c}{c_2}+\theta\right) \int_i \frac{U_{c_{i2}}}{\overline{U_{c_2}}} z_i di+  \frac{\tau}{\kappa}\varepsilon_{(1-\kappa)z} z R
\end{align*}

The intuition of the inter-temporal redistribution is the following. There is a positive value of doing inter-temporal redistribution if workers are not preparing well for retirement, i.e., $\gamma \frac{\Delta c}{c_2}+\theta>0$. The total value of improving retirement preparation is the sum of each worker's social value and the amount in which the preparation increases. The former is given by their consumption at retirement relative to the average, while the latter is given by the taxable earnings, which define the amount the worker is forced to save. 

By adding heterogeneity on the preparation for retirement, we get:

\begin{align*}
   \frac{dW}{d\kappa} \approx  \int_i  \left(\gamma \frac{\Delta c_i}{c_{i2}}+\theta\right)\frac{U_{c_{i2}}}{\overline{U_{c_2}}} z_i di+  \frac{\tau}{\kappa}\varepsilon_{(1-\kappa)z} z R
\end{align*}

Now, there is one extra component to the the inter-temporal redistribution value. The value depends on how well the worker that receive the transfer is prepared for retirement, the amount of the transfer that he receives, and the social value of that worker.\\ 

\noindent \textbf{Inter-worker.-} The inter-worker redistribution is equivalent to that of an linear-income tax with the difference that the redistribution of the taxed amount is done at retirement.  This difference afects the social value of the inter-worker redistribution by adding retirement preparation to it. To see this, let's start with the case where workers are perfectly prepared for retirement. Then, the social gains of increasing benefits' progressivity are given by:  

\begin{align*}
    \frac{dW}{d\phi} = \kappa \int_i  \mu U_{c_{i1}} (z_i-\overline{z}) d\mu -  \overline{g U_{c_{i2}}}(\kappa \phi+\tau)\left(\varepsilon_{(1-\phi)\overline{z}}-\varepsilon_{b\overline{z}} \right) \overline{z} \\ = \kappa COV(U_{c_1}(c_{i1}),z_i)+ \overline{g U_{c_{i2}}}(\kappa \phi+\tau)\left(\varepsilon_{(1-\phi)\overline{z}}-\varepsilon_{b\overline{z}} \right) \overline{z}
\end{align*}

which is equivalent to \cite{sheshinski1972optimal}. The transfer generated by an increase in $\phi$ is decreasing on active life earnings, becoming negative at the average earnings ($\overline{z}$). The social value of the transfer receiver is given by the marginal utility of his consumption, given by $U_{c_{1}}(c_{i1})$. Therefore, the value of the inter-worker redistribution is proportional to $-COV(U_{c_{1}}(c_{i1}), z_i)$. 

By adding unoptimal preparation for retirement, we get:

\begin{align*}
    \frac{dW}{d\phi} = \kappa \int_i \mu U_{c_{i1}}\theta \left(1+\gamma\frac{(c_{i1}-c_{i2})}{c_{i1}}\right) (z_i-\overline{z}) d\mu -  \overline{g U_{c_{i2}}}(\kappa \phi+\tau)\left(\varepsilon_{(1-\phi)\overline{z}}-\varepsilon_{b\overline{z}} \right) \overline{z} 
\end{align*}

When we introduce unoptimal preparation for retirement, there is an extra component for the value of inter-worker redistribution has an extra component: retirement preparedness. The transfer is given at retirement, therefore the social value is larger for workers that are less prepare for retirement  the transfer increases more for those worker

\begin{proof}[\textbf{Proof of Lemma 1}]
Doing a Taylor expansion of $\partial U/ \partial c_2$ around active consumption we get:

\begin{align}
    \frac{\partial U}{\partial c_2} (c_2) \approx&  \frac{\partial U}{\partial c_2}(c_1) -  \frac{\partial^2 U}{\partial c_2^2} (c_1)[c_2-c_1] \notag \\
    =& \frac{\partial U}{\partial c_2}(c_1) \left( 1- \frac{c_1 \frac{\partial^2 U}{\partial c_2^2}}{ c_1 \frac{ \partial U}{\partial c_2}} [c_2-c_1] \right) \notag \\
    =& \frac{\partial U}{\partial c_2} \left( 1- \gamma(c_1) \frac{[c_2-c_1]}{c_1} \right) \label{taylor_expansion_c1}
\end{align}

where $\gamma$ is the relative risk aversion. 

Under state dependent preferences:

\begin{align}
    \frac{\partial U}{\partial c_2} (c_2) \approx \beta \theta \frac{\partial U}{\partial c_1r}(c_1) \left( 1 -\gamma(c_1) \frac{[c_2-c_1]}{c_1} \right) \label{taylor_expansion_c1}
\end{align}

Replacing (\ref{taylor_expansion_c1}) in the first order conditions (\ref{foc_pensiontax}) and (\ref{foc_contribution}), and using the assumption that $\beta=R^{P-1}$ we get the expression of lemma 2. 

\end{proof}

\renewcommand{\thesection}{C}
\setcounter{figure}{0}
\renewcommand\thefigure{C.\arabic{figure}}
\section{Retirement preparation}
\label{appendix_c}

In this appendix, I use the survey to study the mechanism behind the consumption drop at retirement and its relationship with lifetime earnings. I find that low-income workers are income drops more at retirement, that they have less savings and depend more on government support, and they are more likely to retire because an unexpected health or unemployment shock.  

\subsection{Figures}

\begin{figure}[h!]
    \centering
    \begin{subfigure}[b]{\textwidth}
        \centering
        \includegraphics[width=0.35\linewidth]{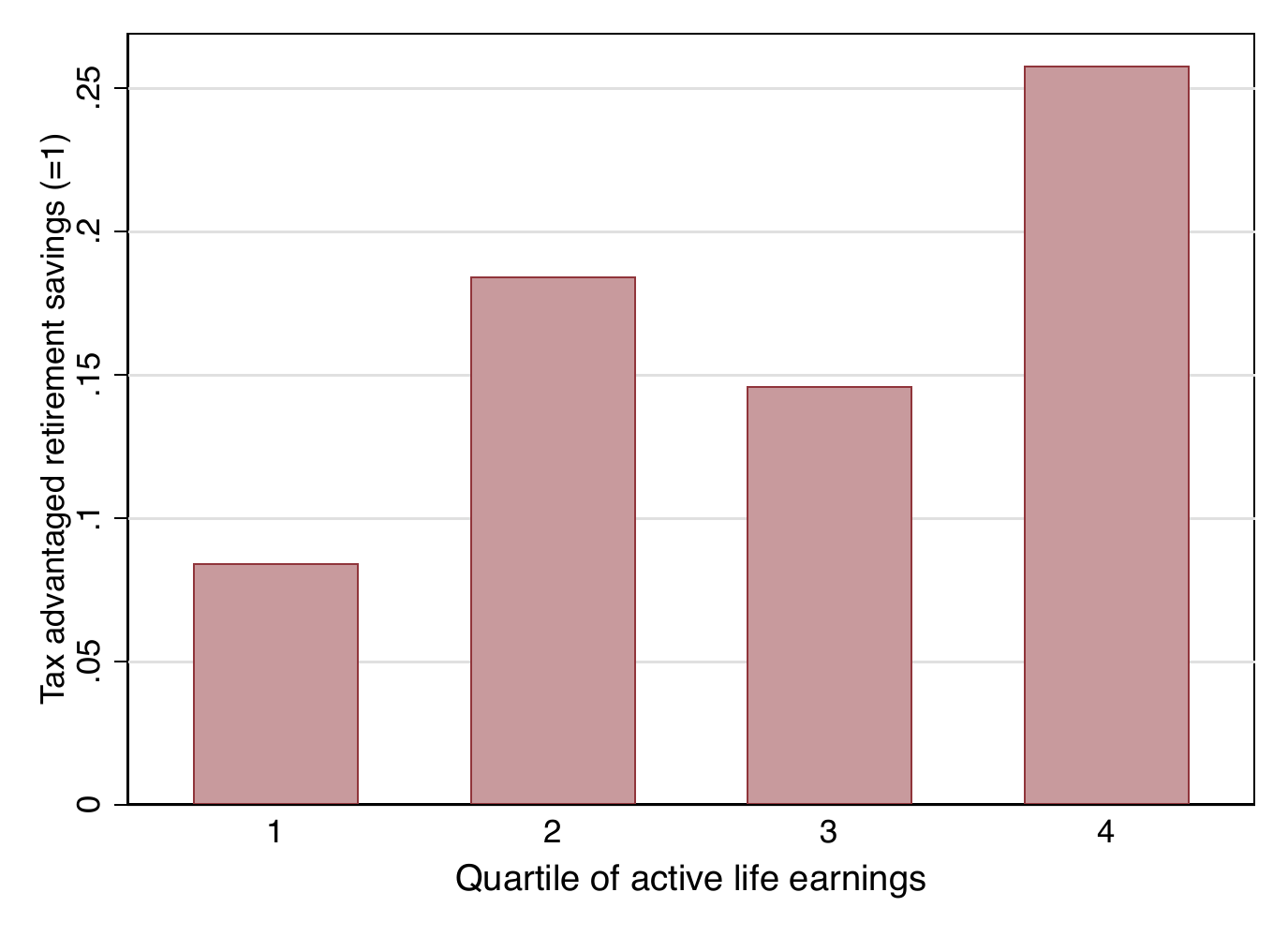}%
        \hfill
        \includegraphics[width=0.35\linewidth]{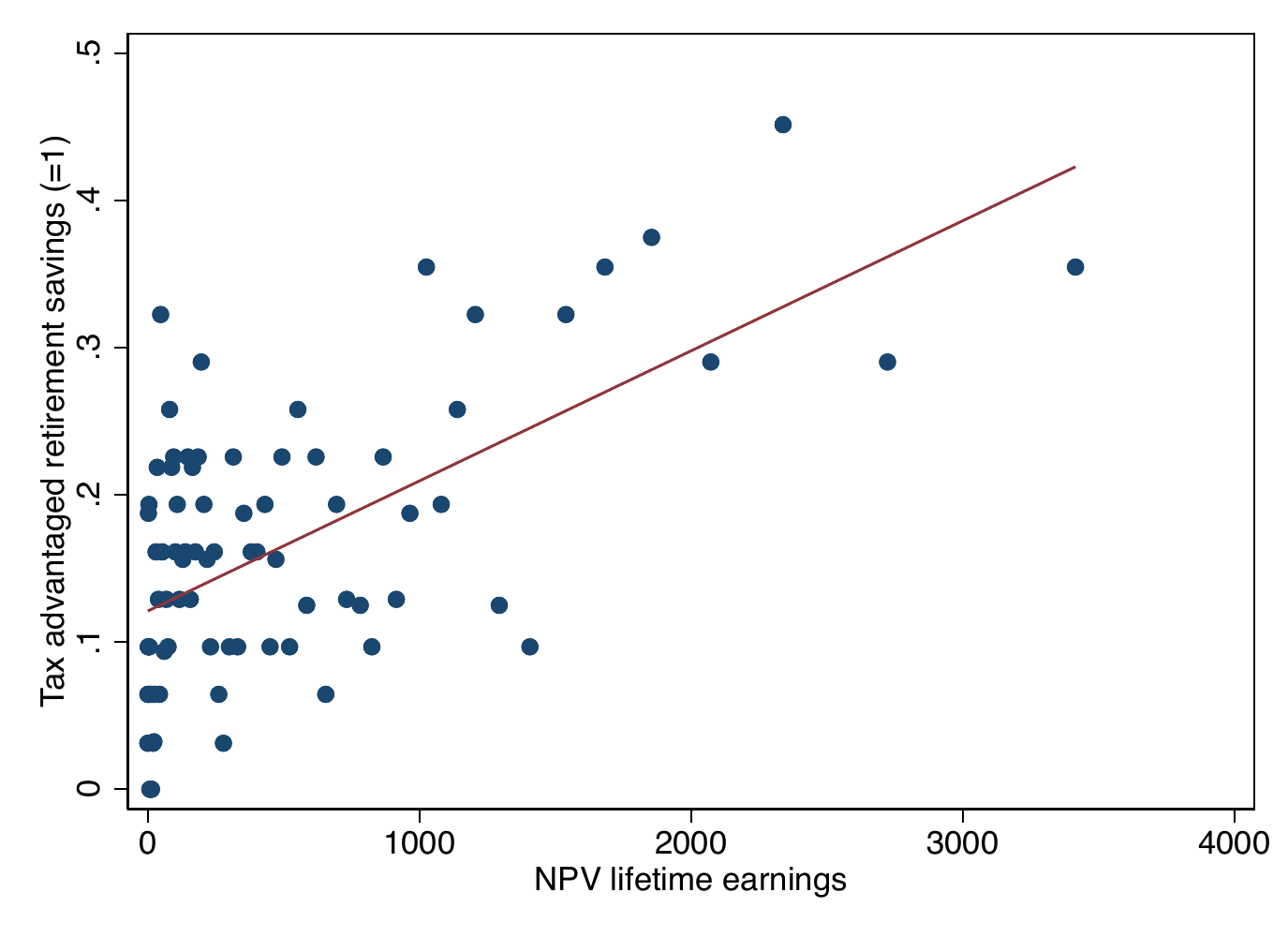}
        \caption{Tax advantaged savings for retirement}
    \end{subfigure}
    \vskip\baselineskip
    \begin{subfigure}[b]{\textwidth}
        \centering
        \includegraphics[width=0.35\linewidth]{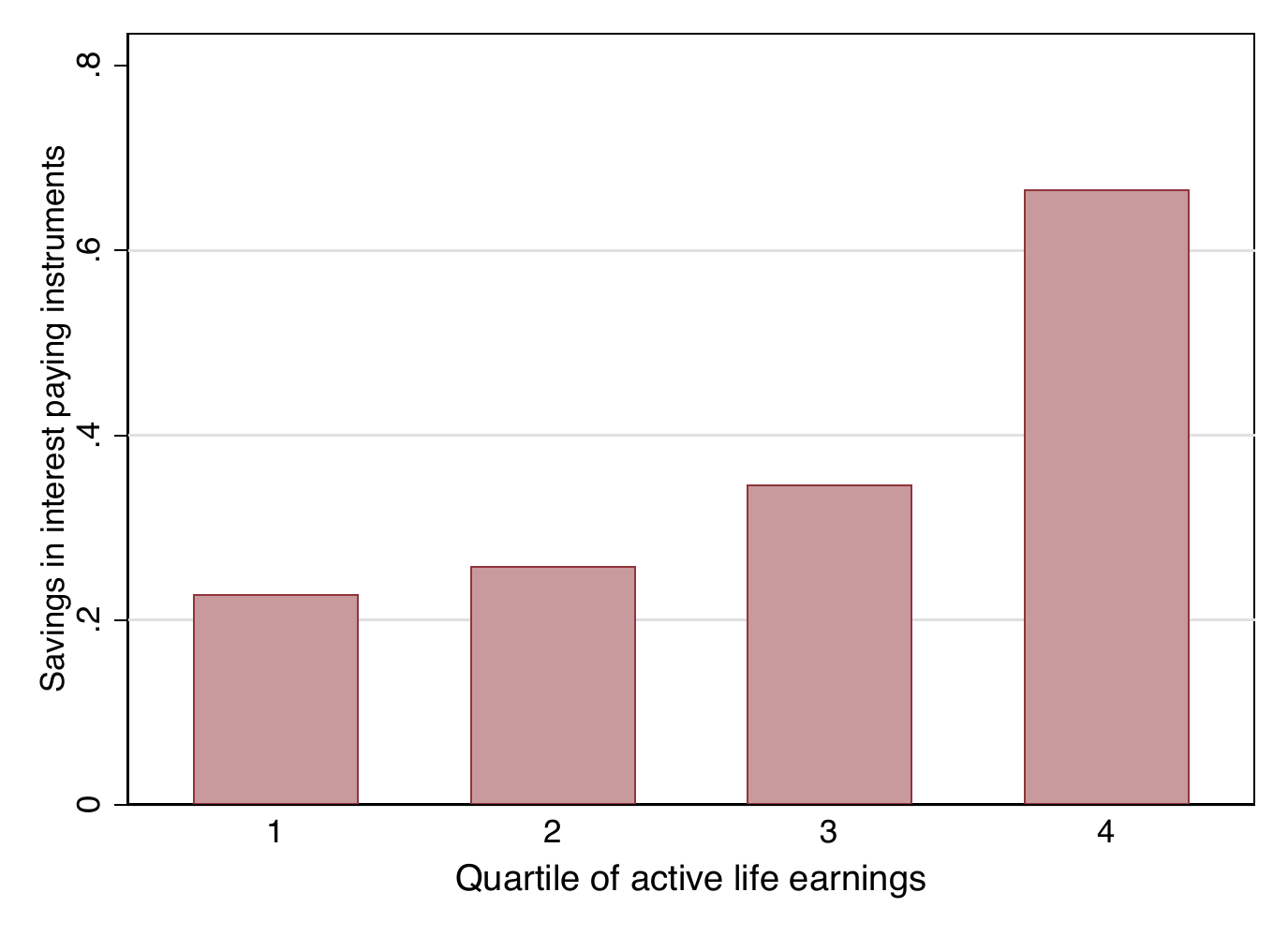}%
        \hfill
        \includegraphics[width=0.35\linewidth]{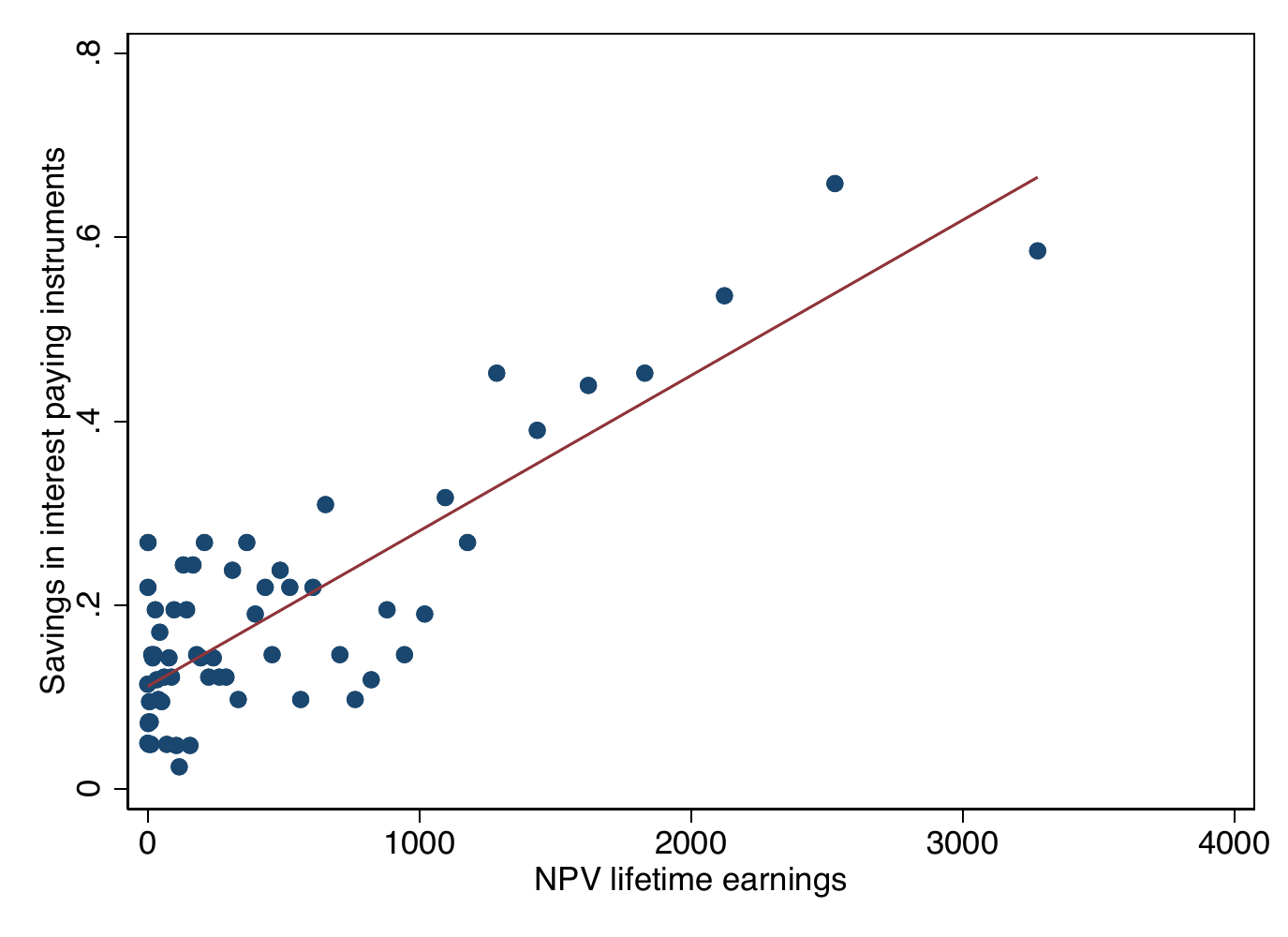}
        \caption{Reports access to interrest payer instruments}
    \end{subfigure}
        \vskip\baselineskip
    \begin{subfigure}[b]{\textwidth}
        \centering
        \includegraphics[width=0.35\linewidth]{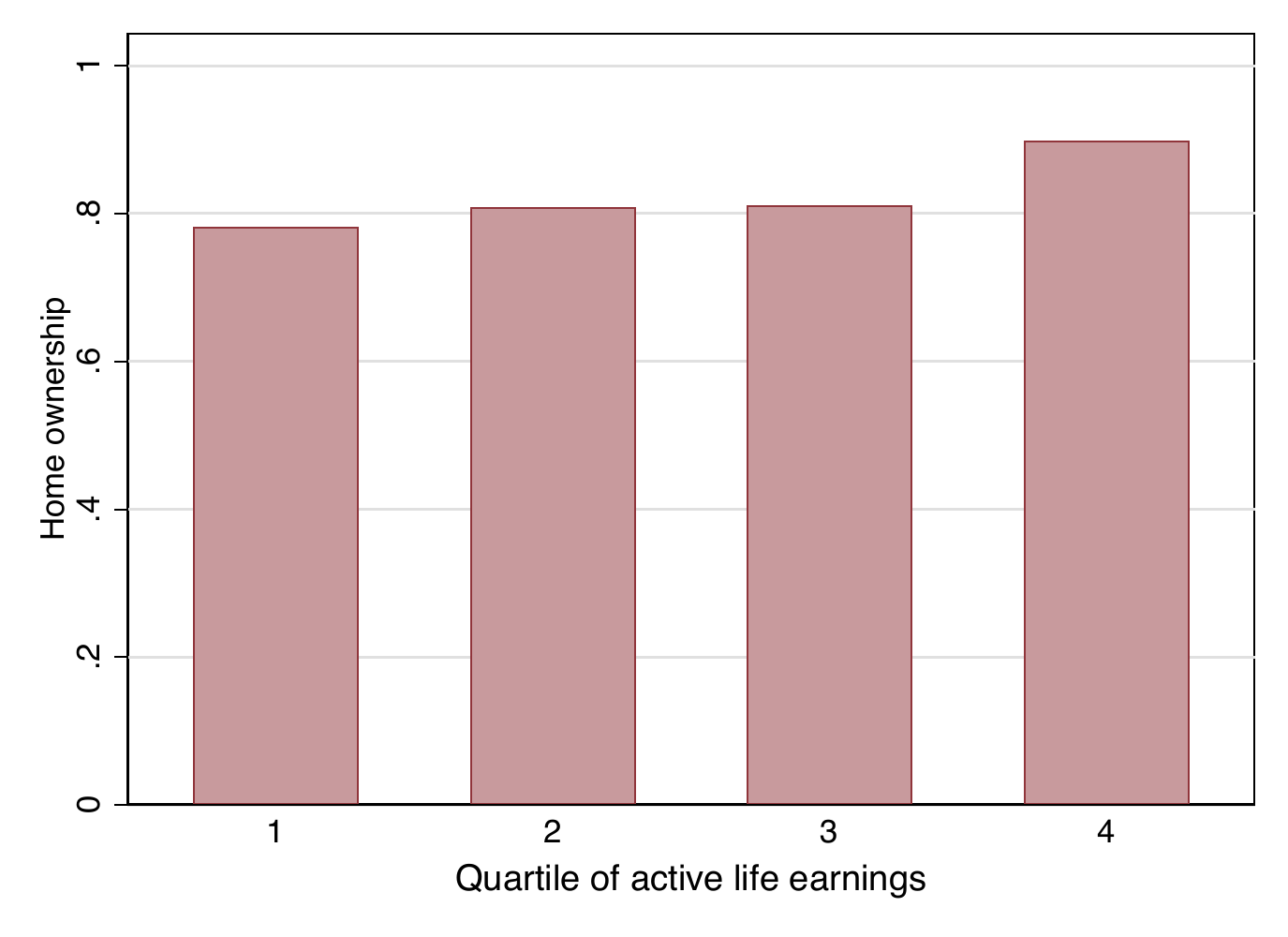}%
        \hfill
        \includegraphics[width=0.35\linewidth]{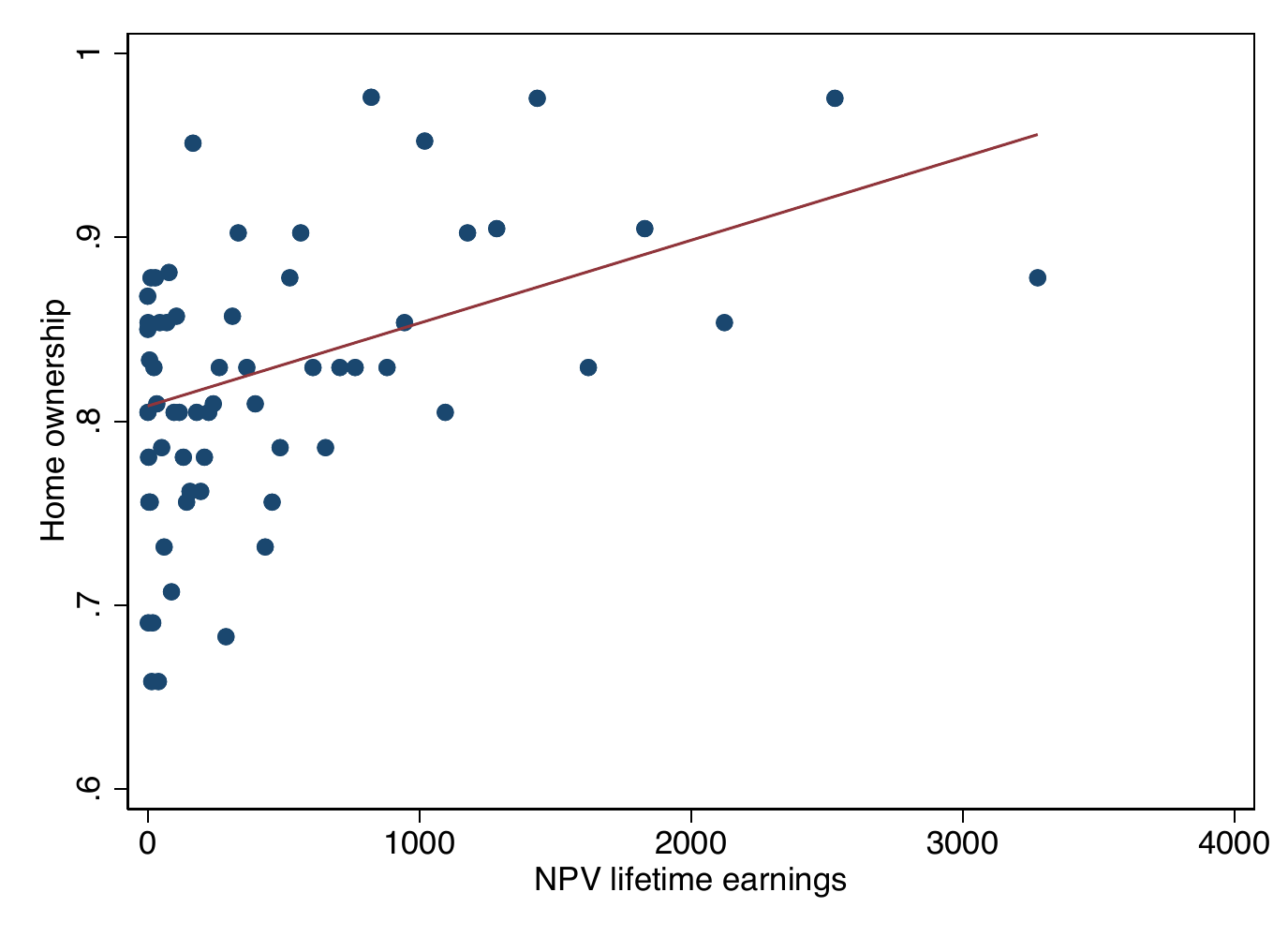}
        \caption{House ownership}
    \end{subfigure}
    \caption{Savings for retirement}
                \label{first_ret_preparedness}
\end{figure}

\begin{figure}[htb]
    \centering
    \begin{subfigure}[b]{\textwidth}
        \centering
        \includegraphics[width=0.475\linewidth]{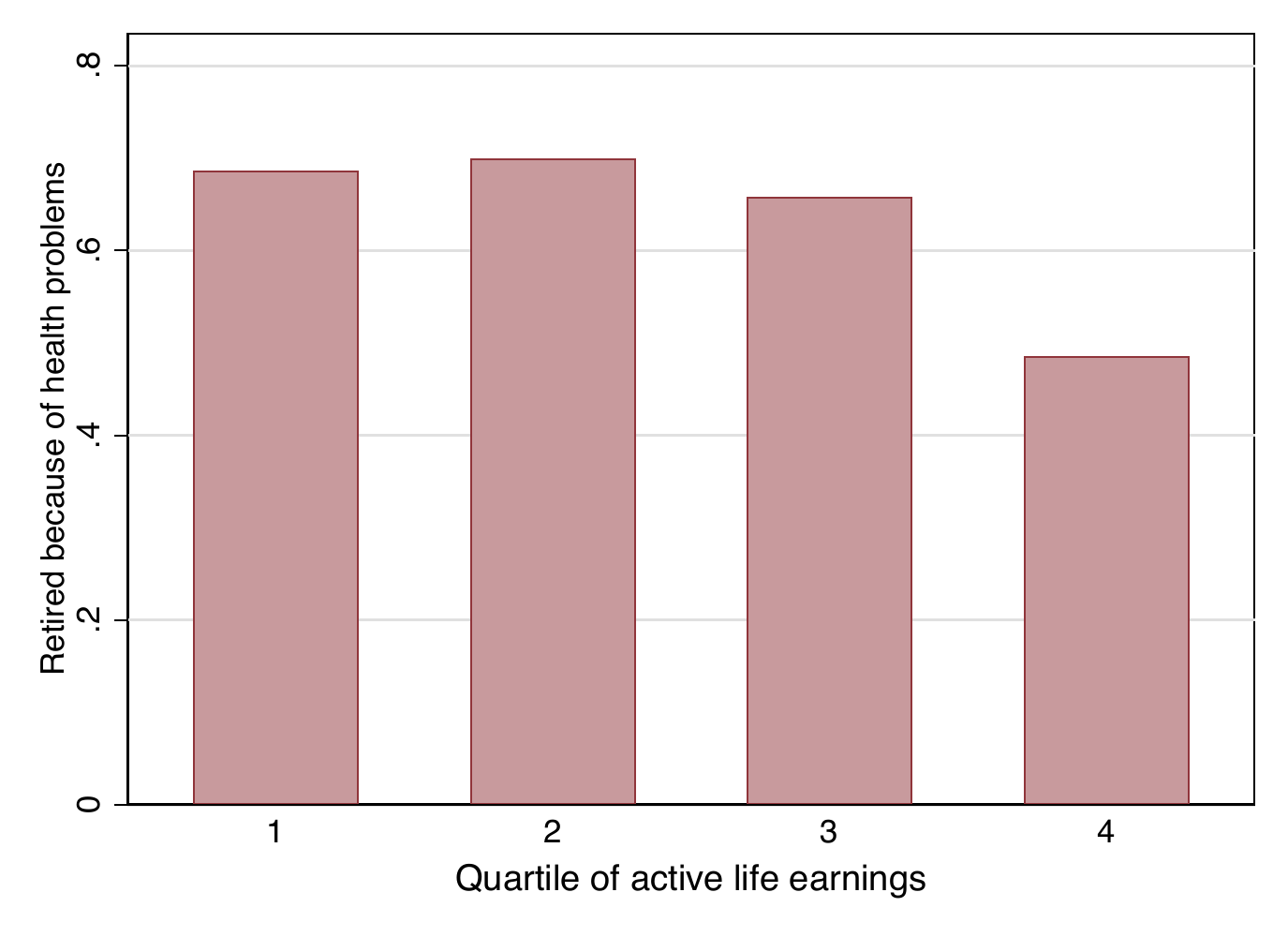}%
        \hfill
        \includegraphics[width=0.475\linewidth]{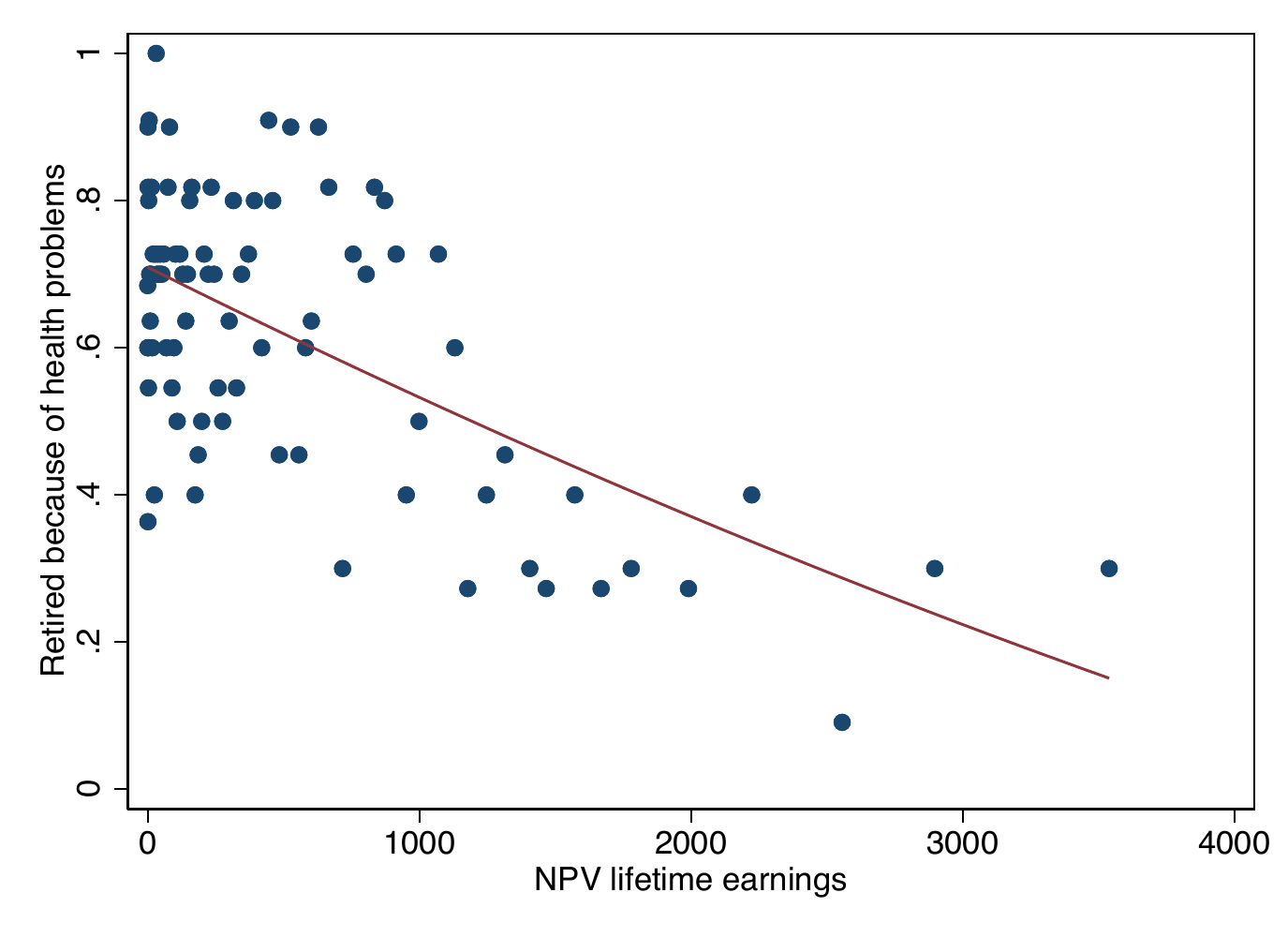}
        \caption{Stopped working because of health problems}
    \end{subfigure}
    \vskip\baselineskip
    \begin{subfigure}[b]{\textwidth}
        \centering
        \includegraphics[width=0.475\linewidth]{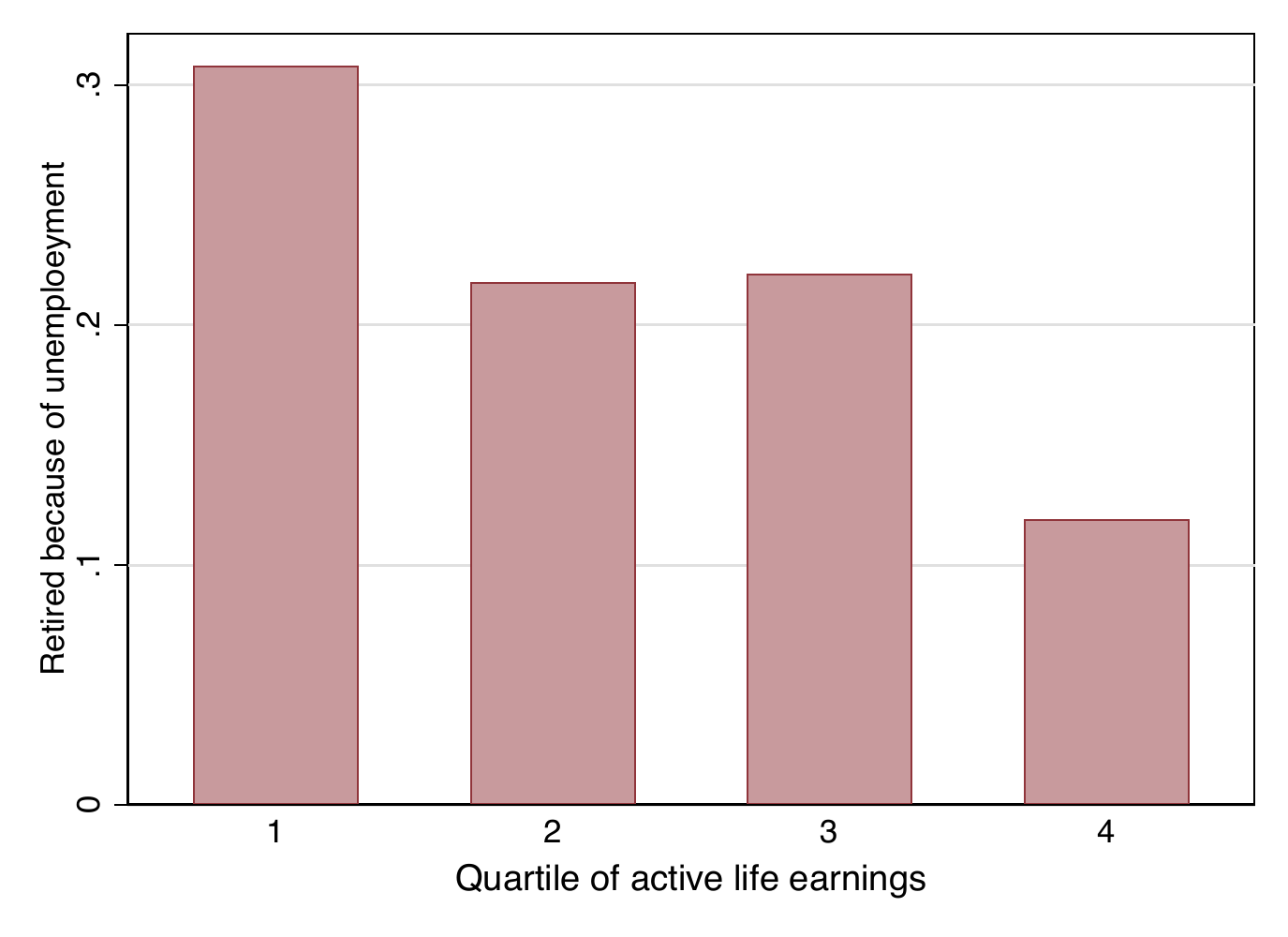}%
        \hfill
        \includegraphics[width=0.475\linewidth]{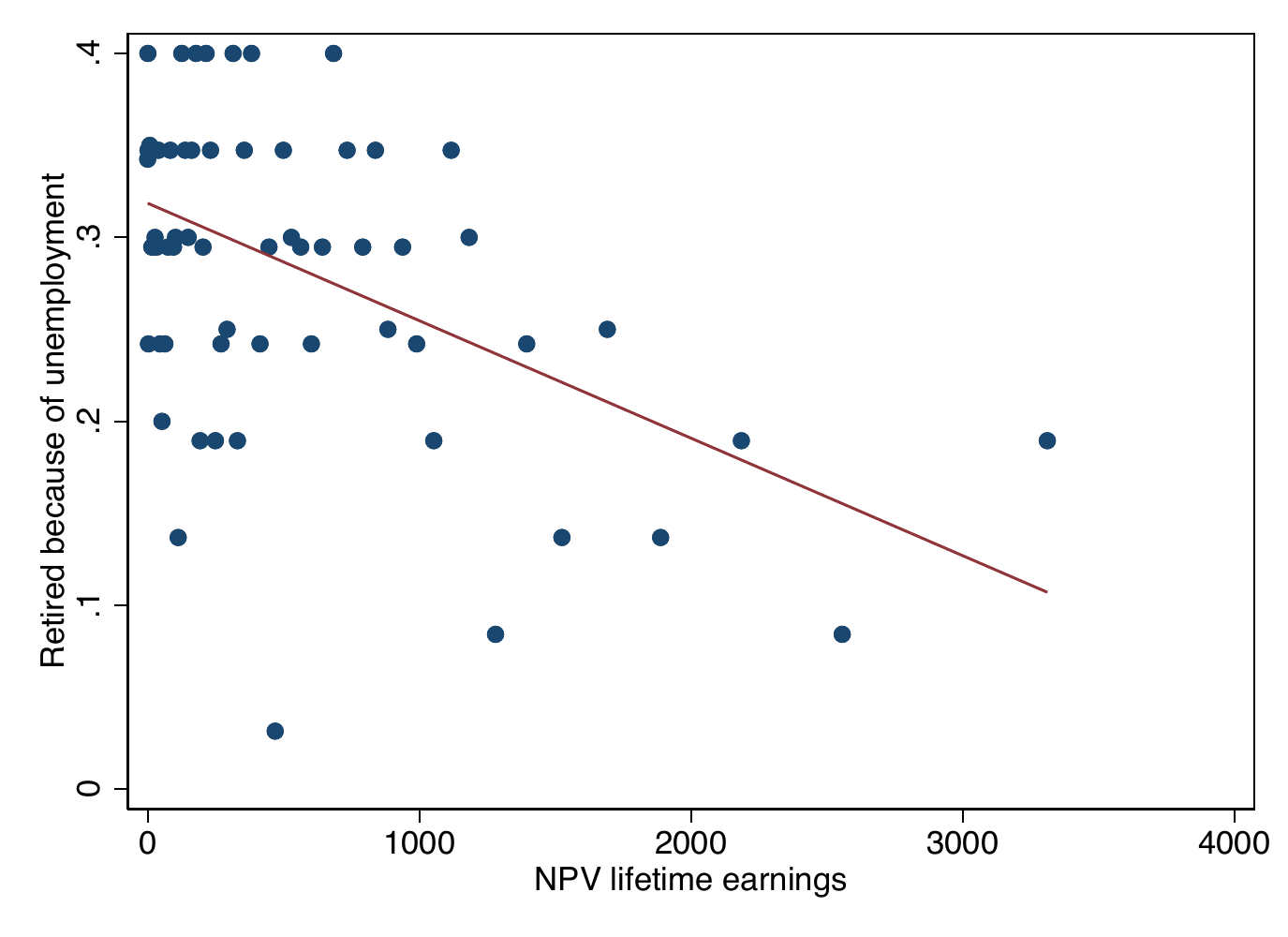}
        \caption{Stopped working because of end of labor relation}
    \end{subfigure}
    \caption{Reason for retirement}
\end{figure}

\begin{figure}[htb]
    \centering
    \begin{subfigure}[b]{\textwidth}
        \centering
        \includegraphics[width=0.475\linewidth]{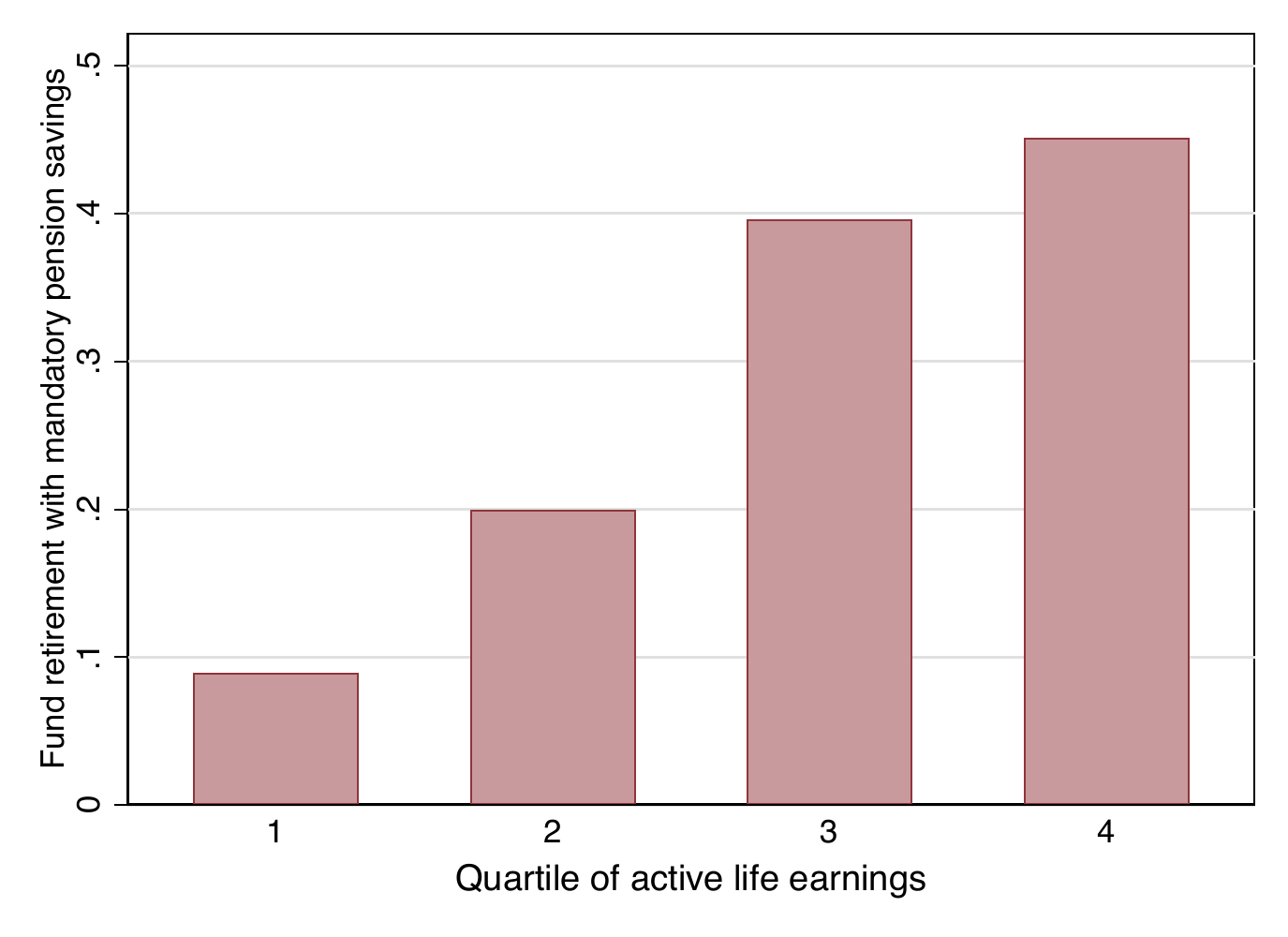}%
        \hfill
        \includegraphics[width=0.475\linewidth]{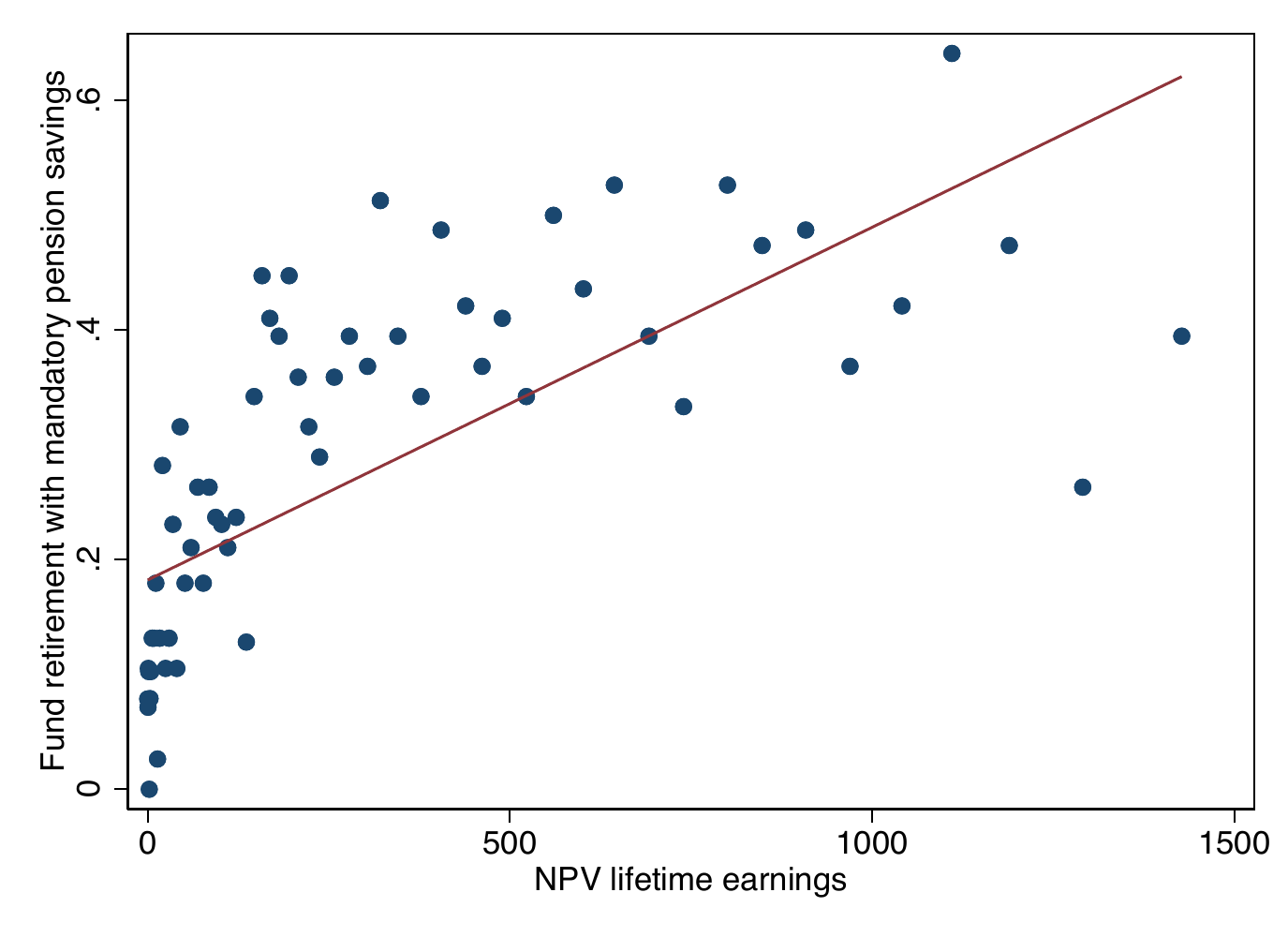}
        \caption{Fund retirement with mandatory pension savings}
    \end{subfigure}
    \vskip\baselineskip
    \begin{subfigure}[b]{\textwidth}
        \centering
        \includegraphics[width=0.475\linewidth]{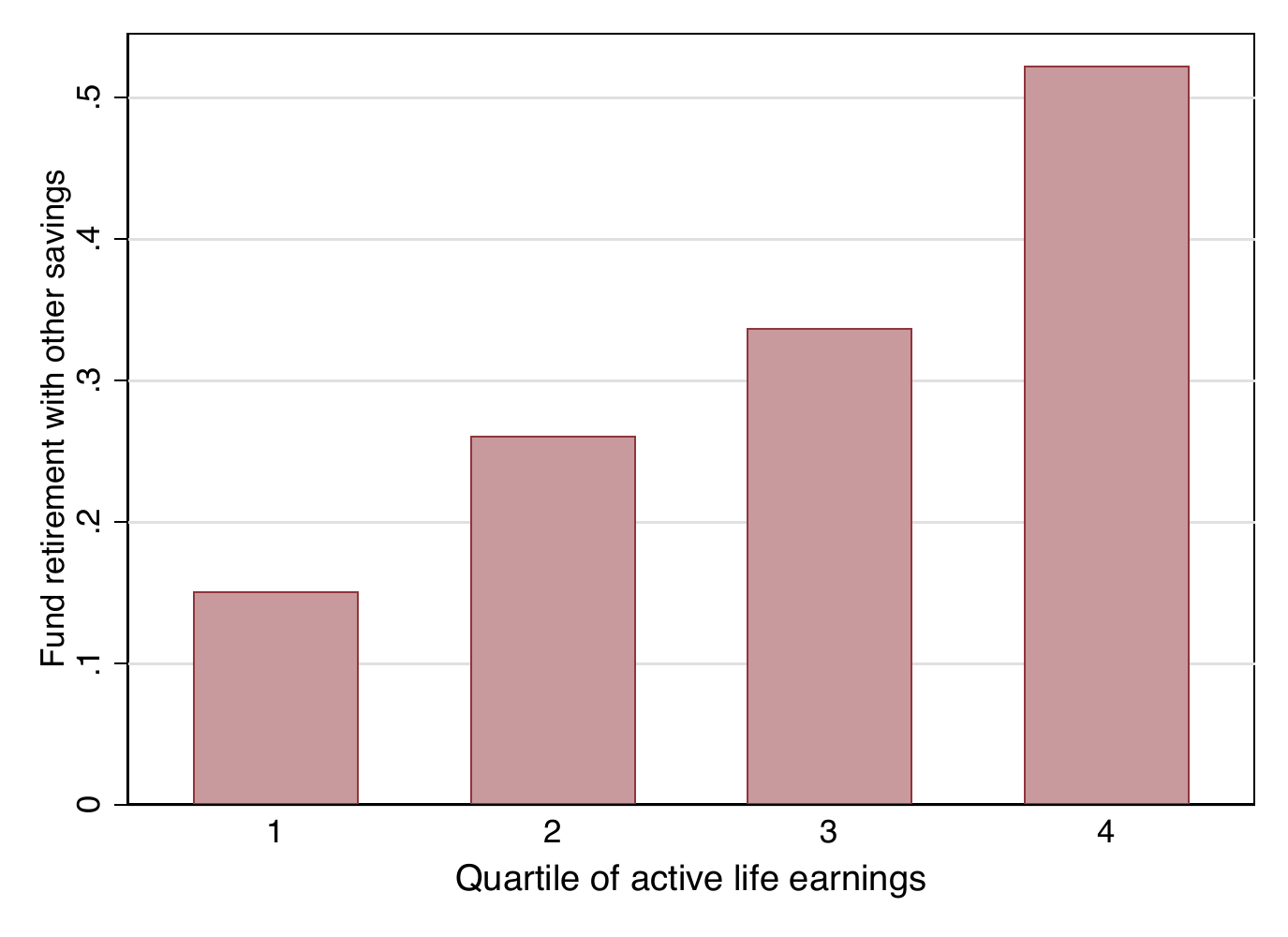}%
        \hfill
        \includegraphics[width=0.475\linewidth]{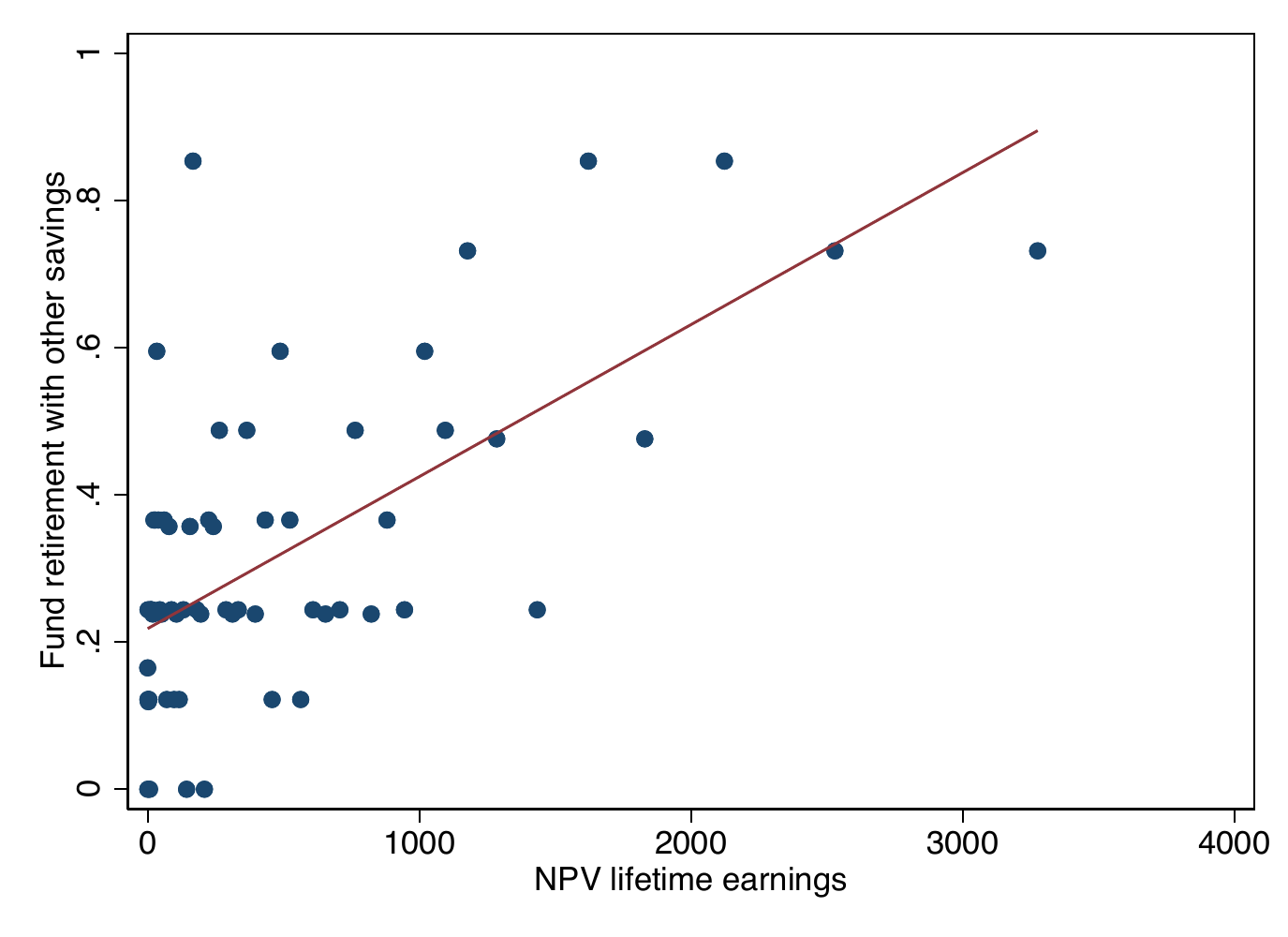}
        \caption{Fund retirement with other (non-mandatory) savings}
    \end{subfigure}
        \vskip\baselineskip
    \begin{subfigure}[b]{\textwidth}
        \centering
        \includegraphics[width=0.475\linewidth]{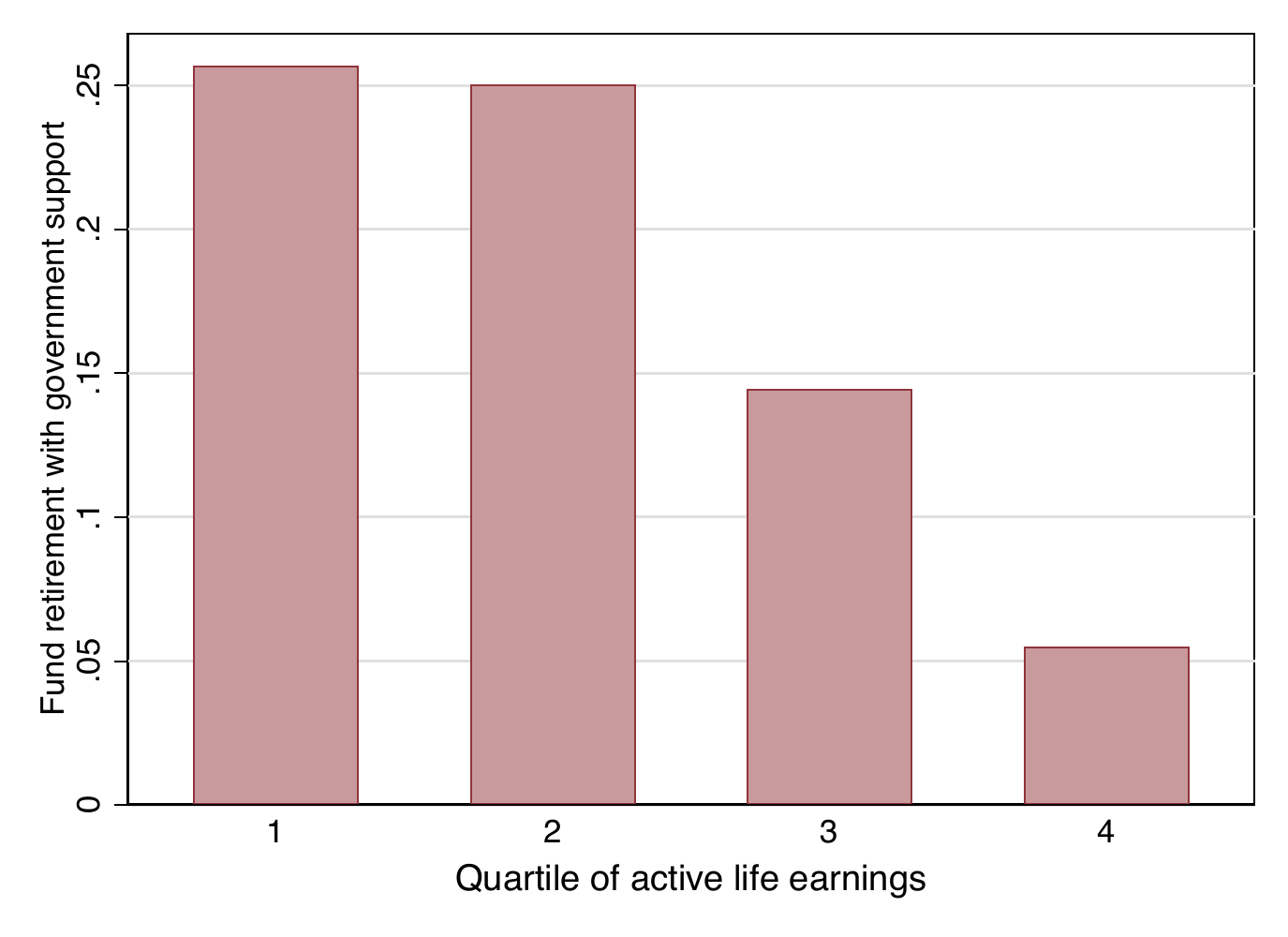}%
        \hfill
        \includegraphics[width=0.475\linewidth]{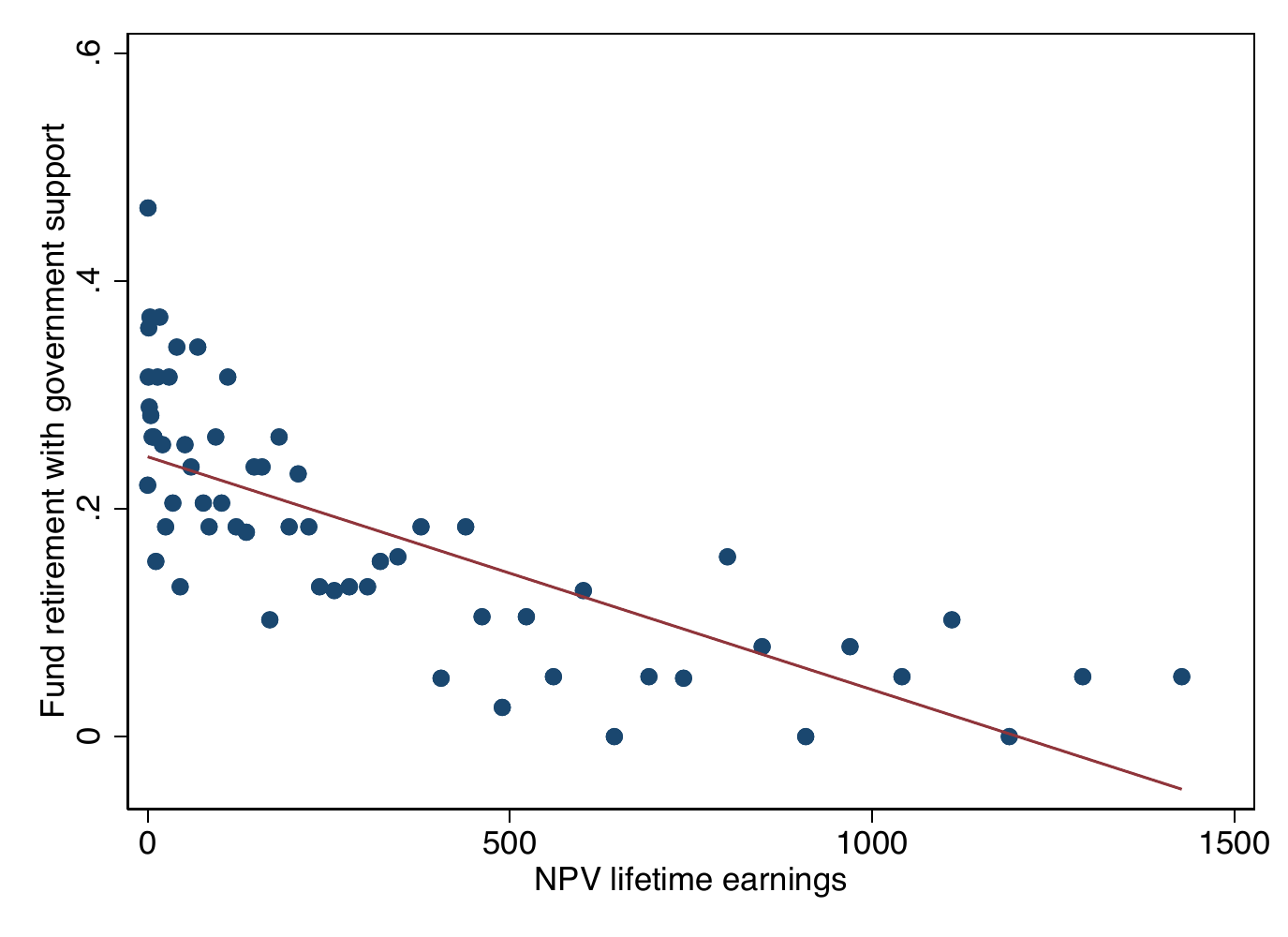}
        \caption{Fund retirement with government support}
    \end{subfigure}
    \caption{Source of income to support household at retirement}
\end{figure}

\begin{figure}[htb]
    \centering
    \begin{subfigure}[b]{\textwidth}
        \centering
        \includegraphics[width=0.475\linewidth]{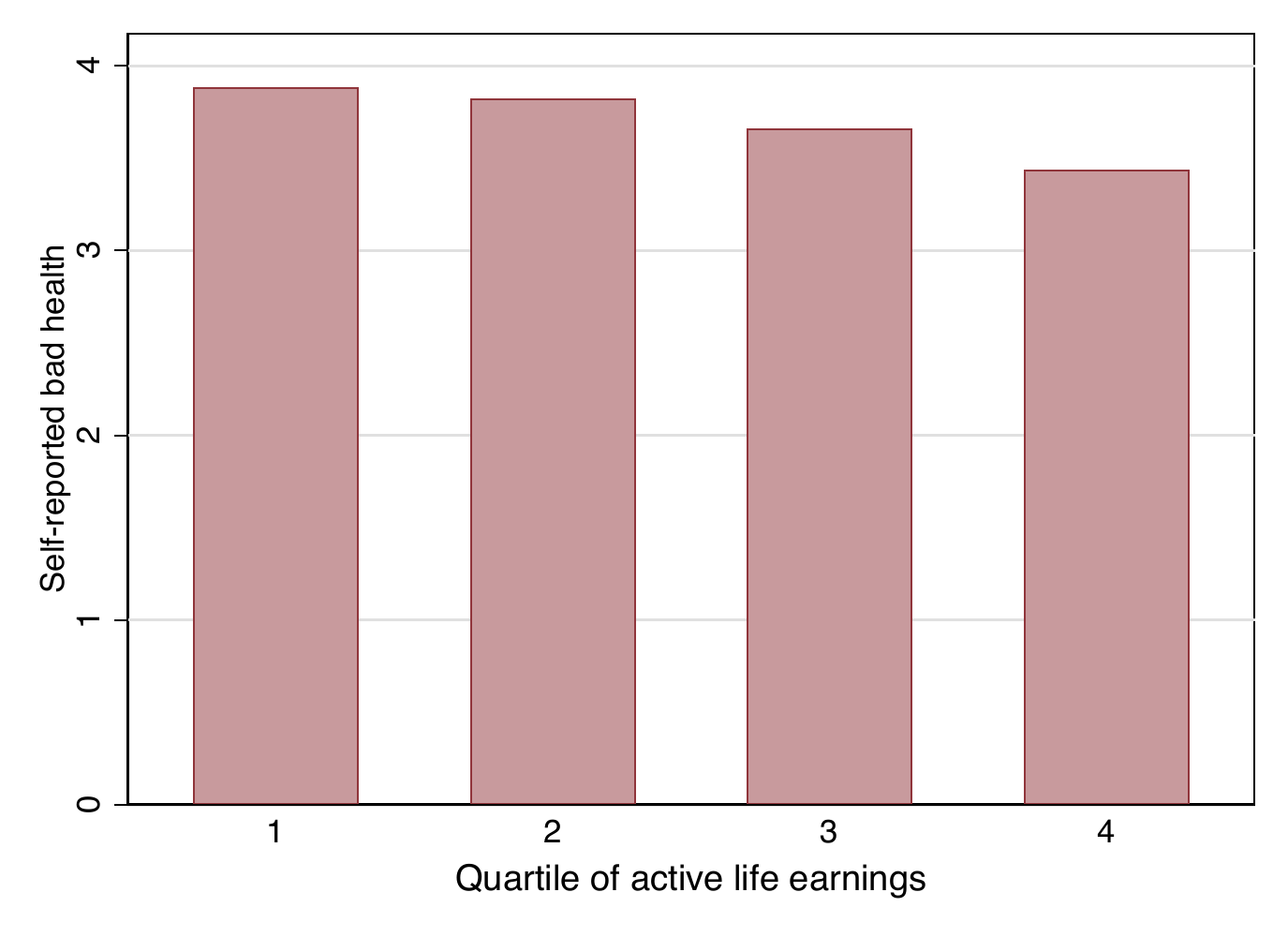}%
        \hfill
        \includegraphics[width=0.475\linewidth]{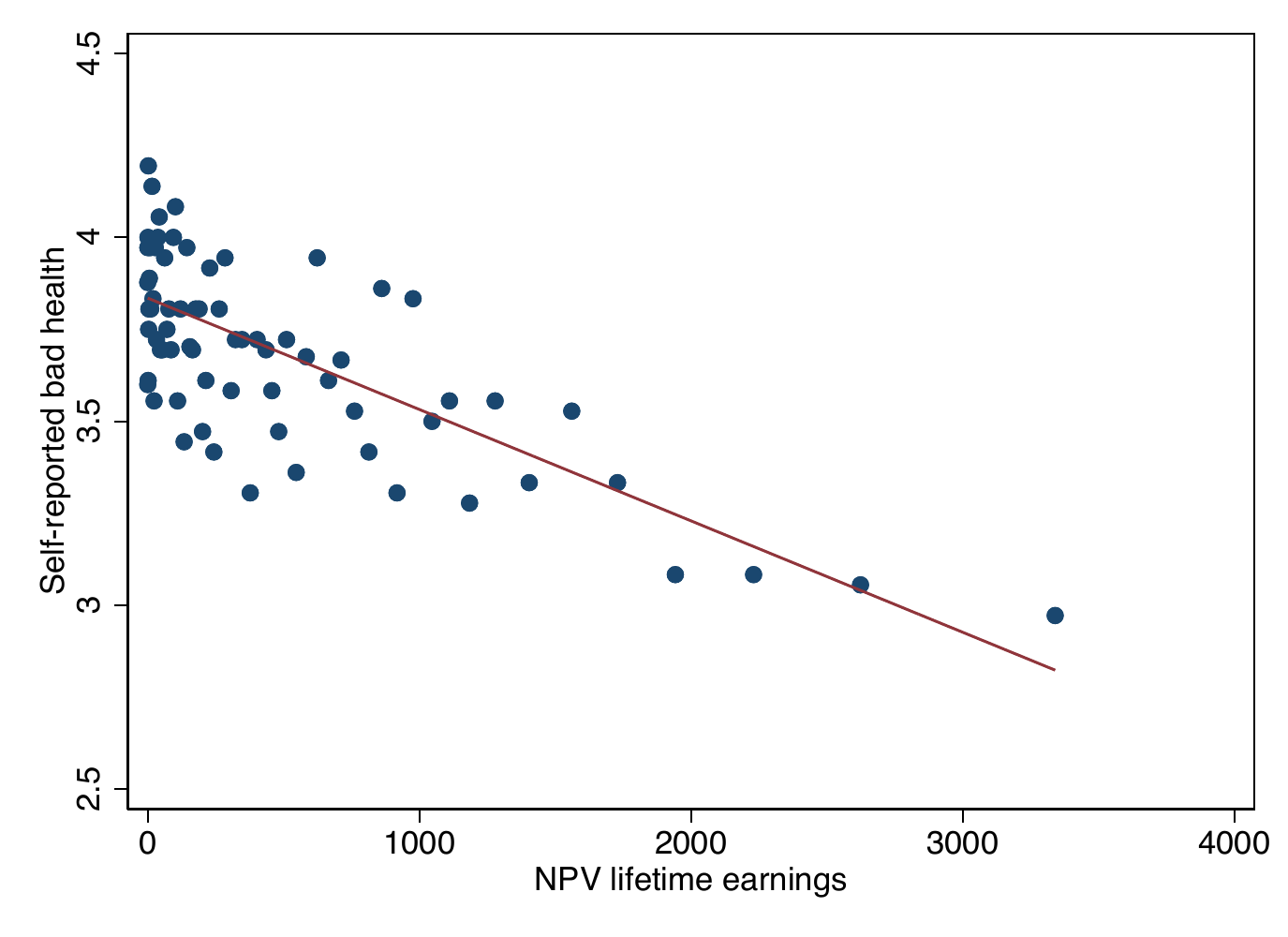}
        \caption{Self-reported health}
    \end{subfigure}
    \vskip\baselineskip
    \begin{subfigure}[b]{\textwidth}
        \centering
        \includegraphics[width=0.475\linewidth]{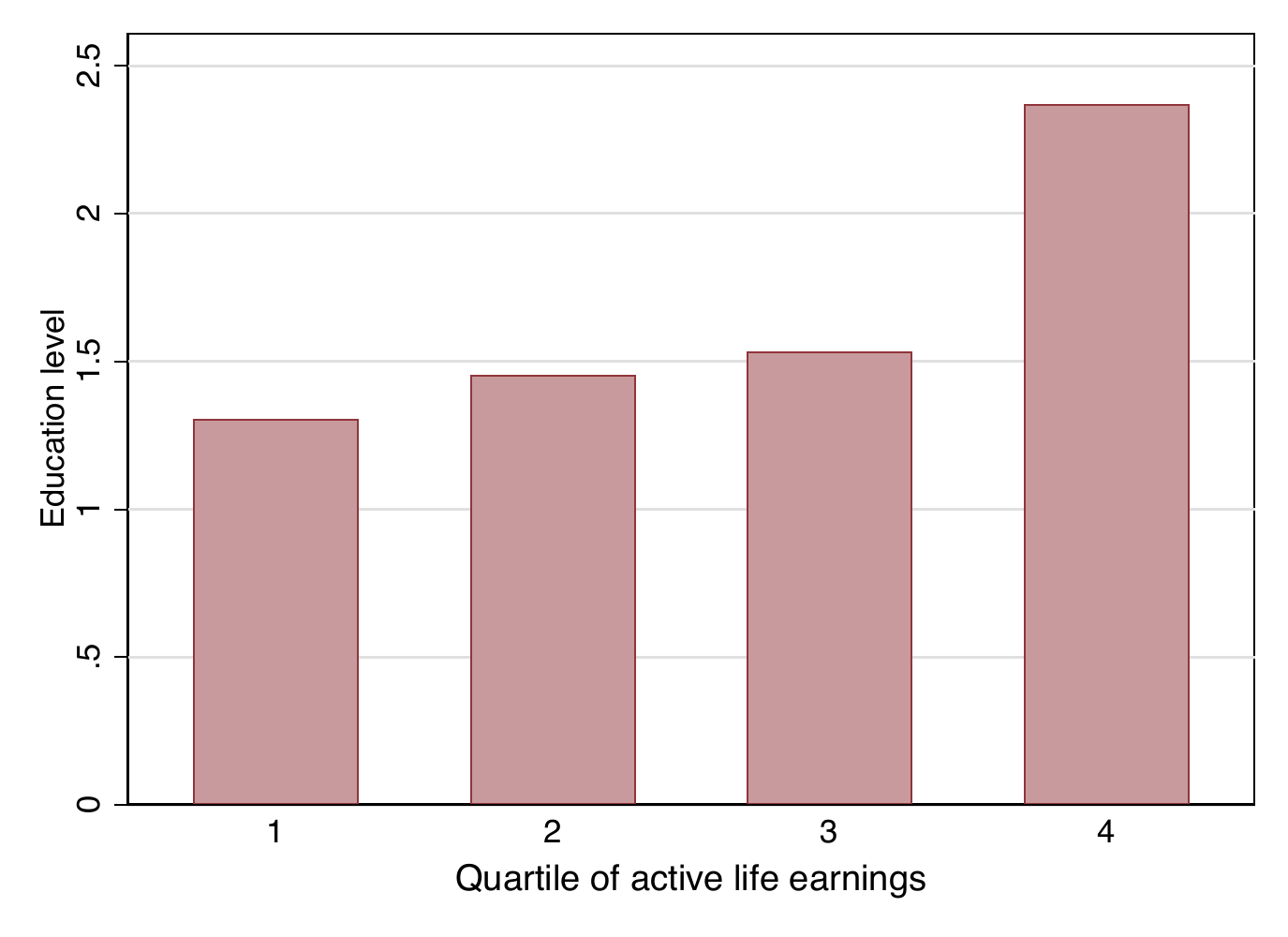}%
        \hfill
        \includegraphics[width=0.475\linewidth]{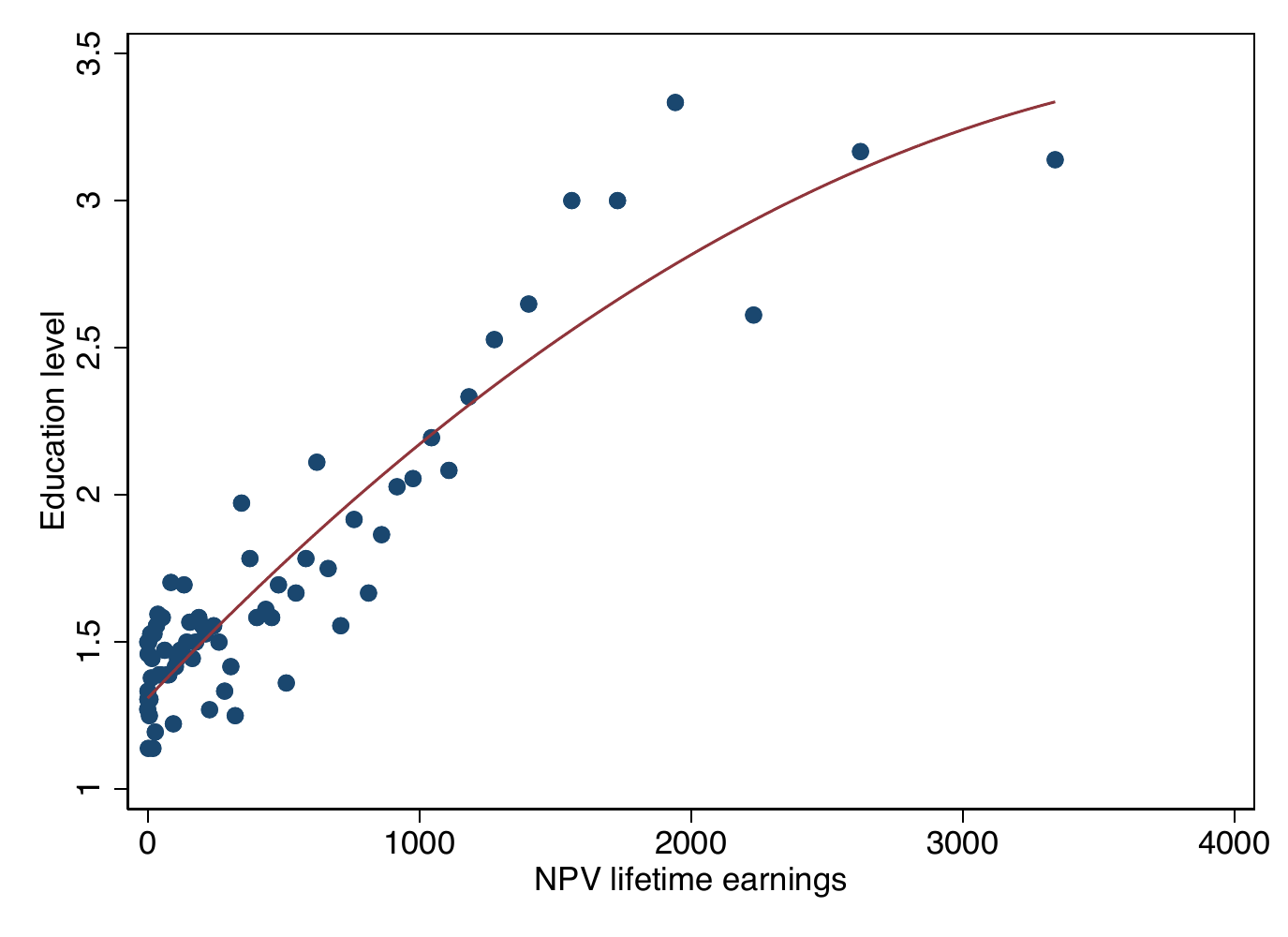}
        \caption{Education level}
    \end{subfigure}
        \caption{Health and education}
\end{figure}

\begin{figure}[htb]
    \centering
    \begin{subfigure}[b]{\textwidth}
        \centering
        \includegraphics[width=0.4\linewidth]{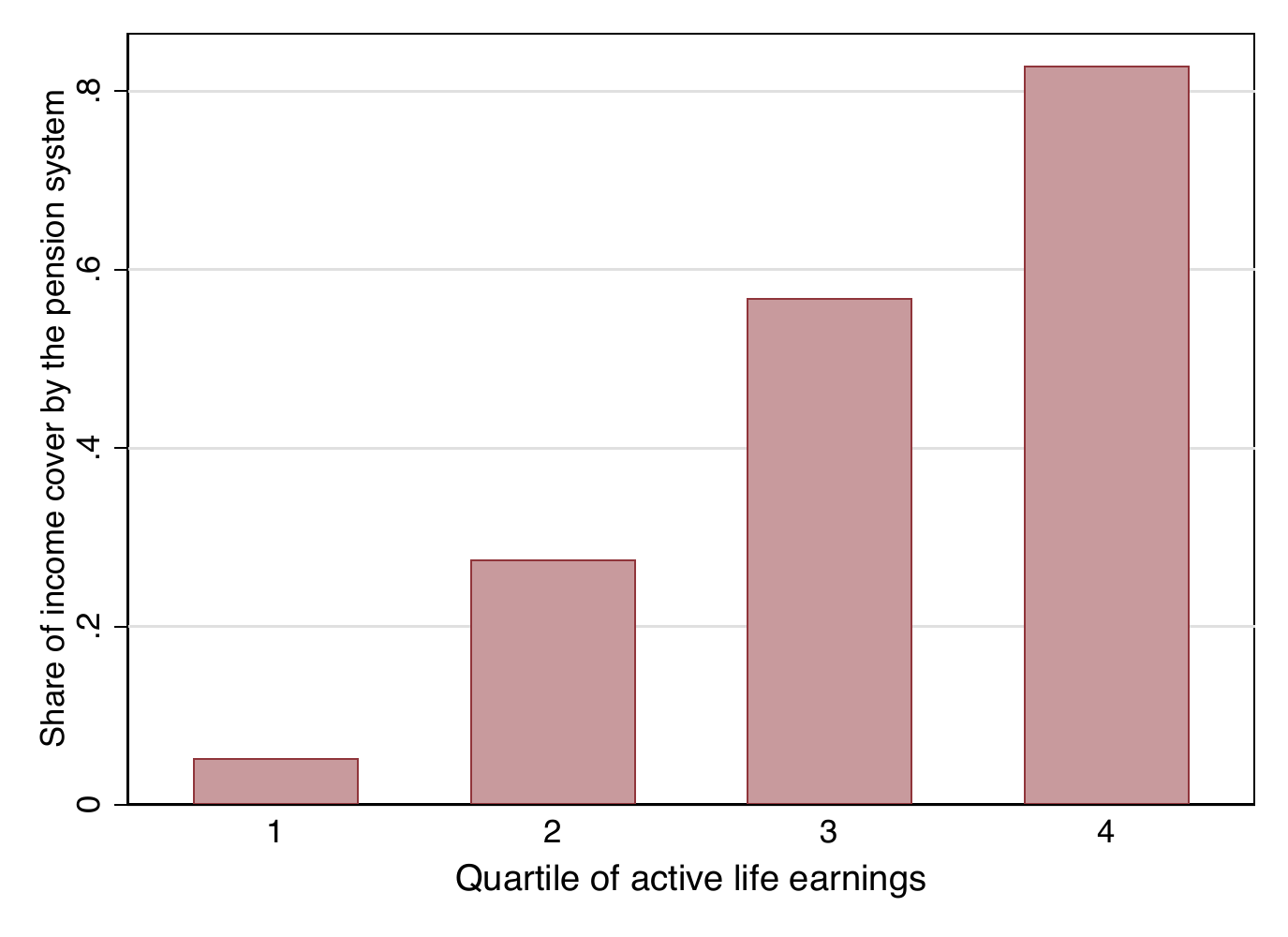}%
        \hfill
        \includegraphics[width=0.4\linewidth]{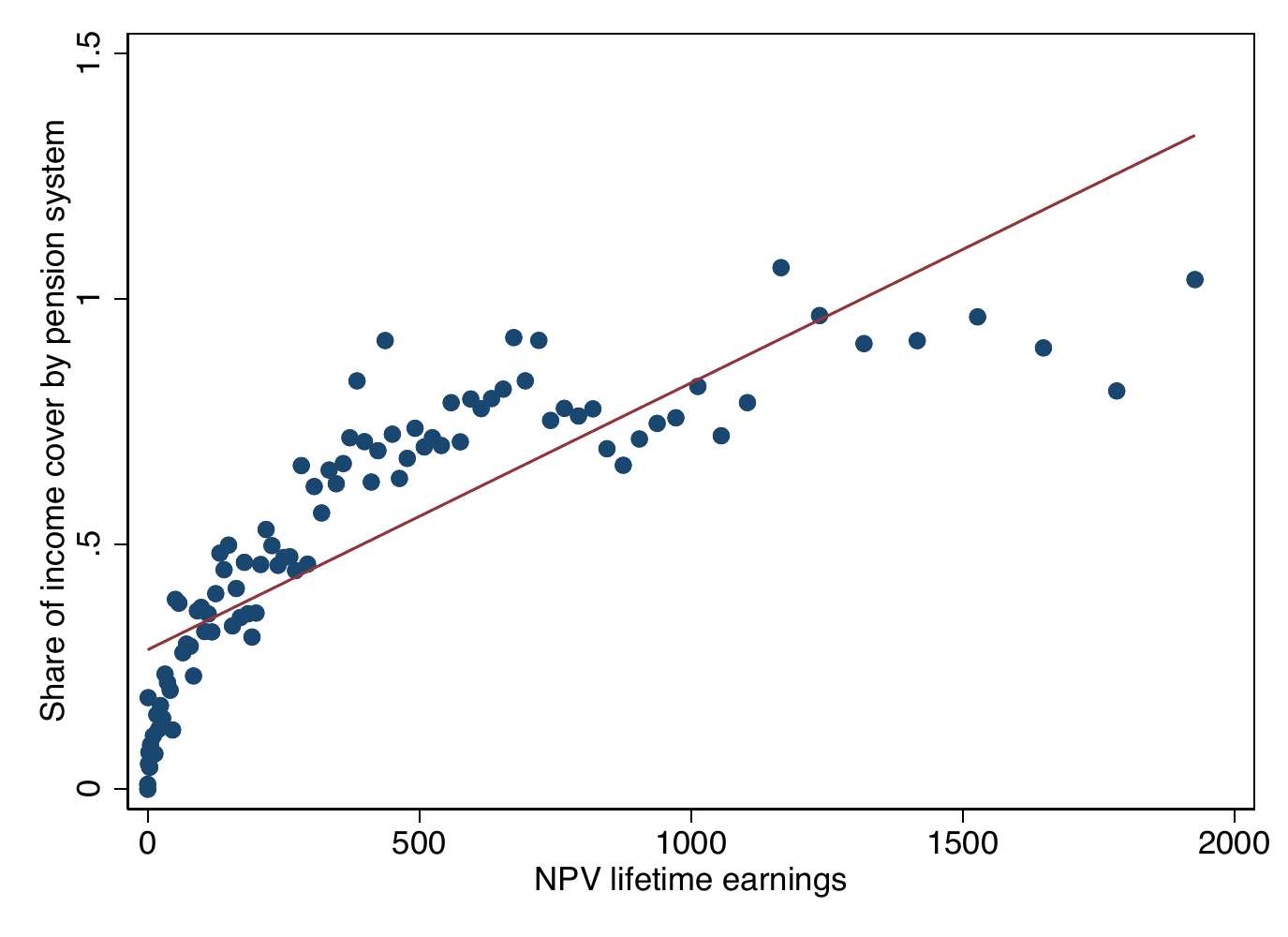}
        \caption{Share of self-reported income covered by the pension system}
    \end{subfigure}
    \vskip\baselineskip
    \begin{subfigure}[b]{\textwidth}
        \centering
        \includegraphics[width=0.4\linewidth]{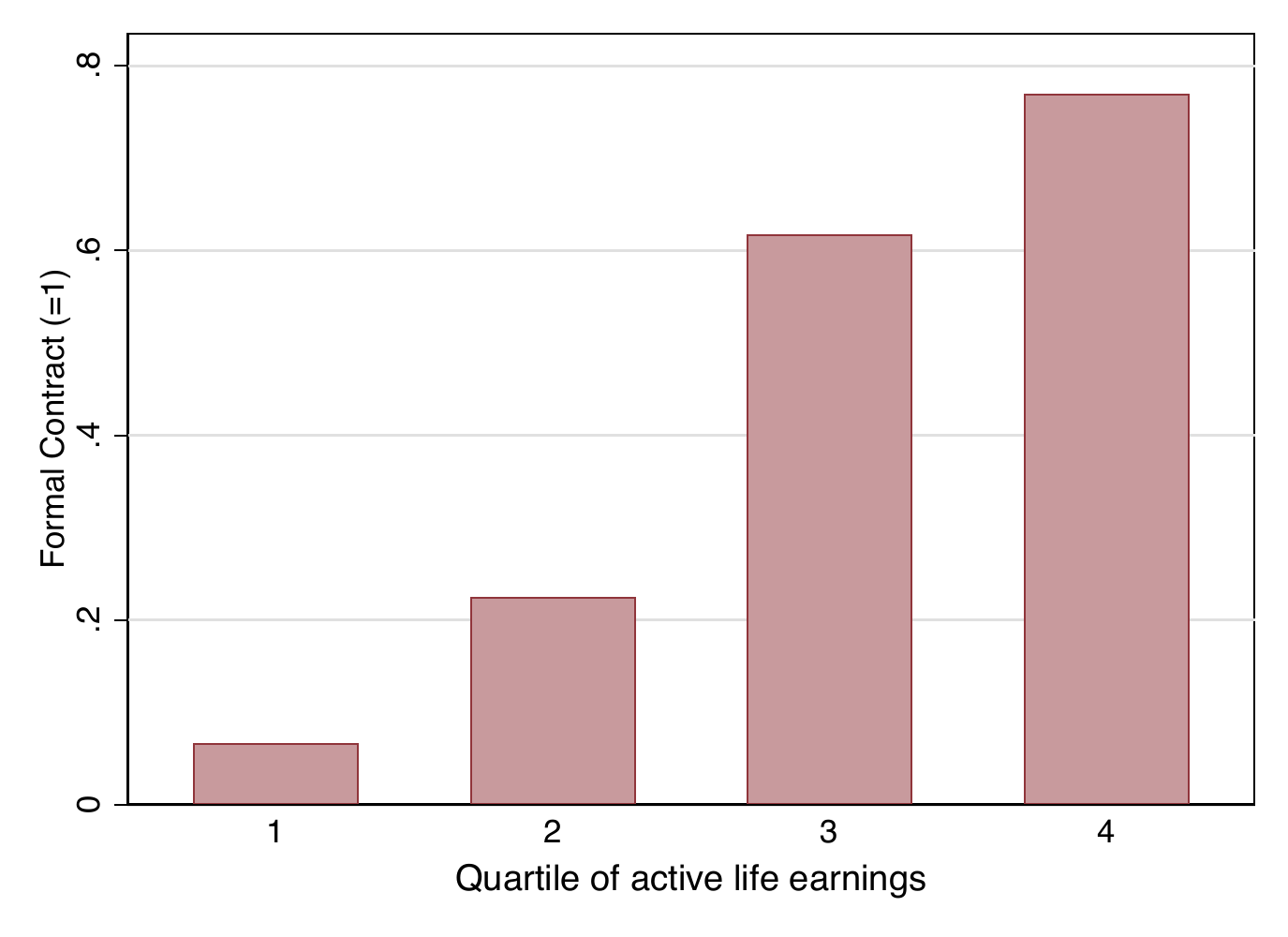}%
        \hfill
        \includegraphics[width=0.4\linewidth]{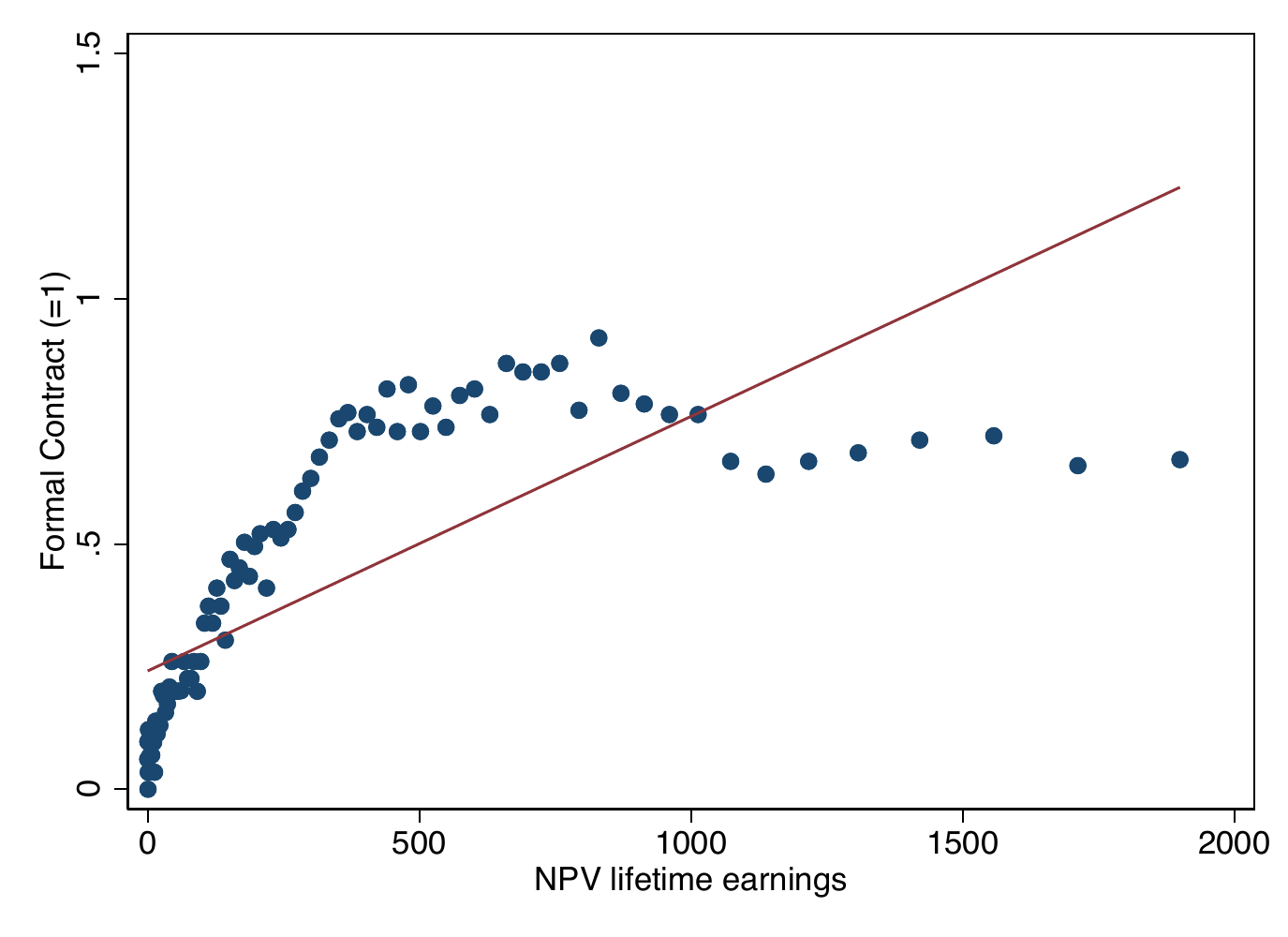}
        \caption{Formal contract}
    \end{subfigure}
            \vskip\baselineskip
    \begin{subfigure}[b]{\textwidth}
        \centering
        \includegraphics[width=0.4\linewidth]{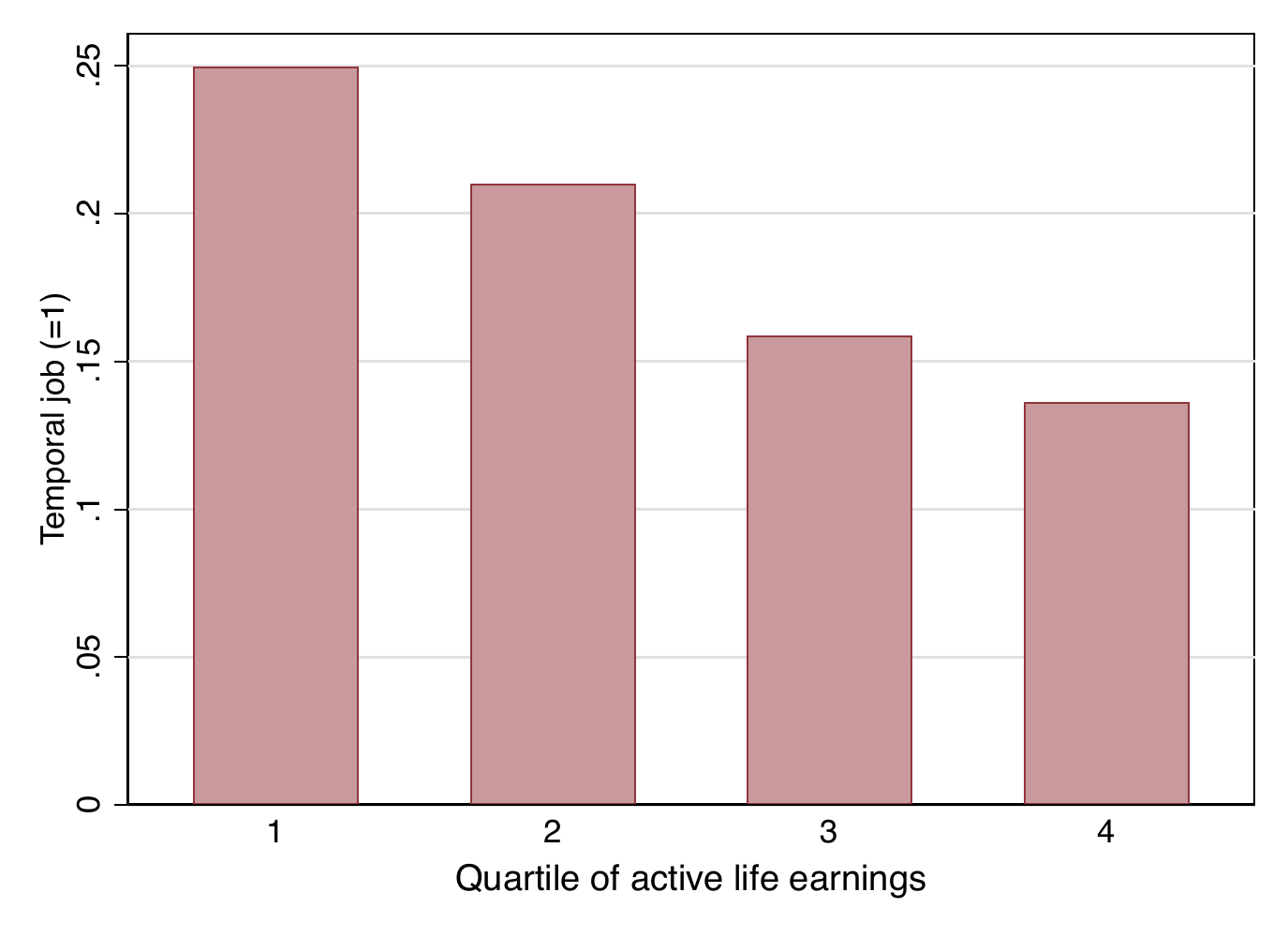}%
        \hfill
        \includegraphics[width=0.4\linewidth]{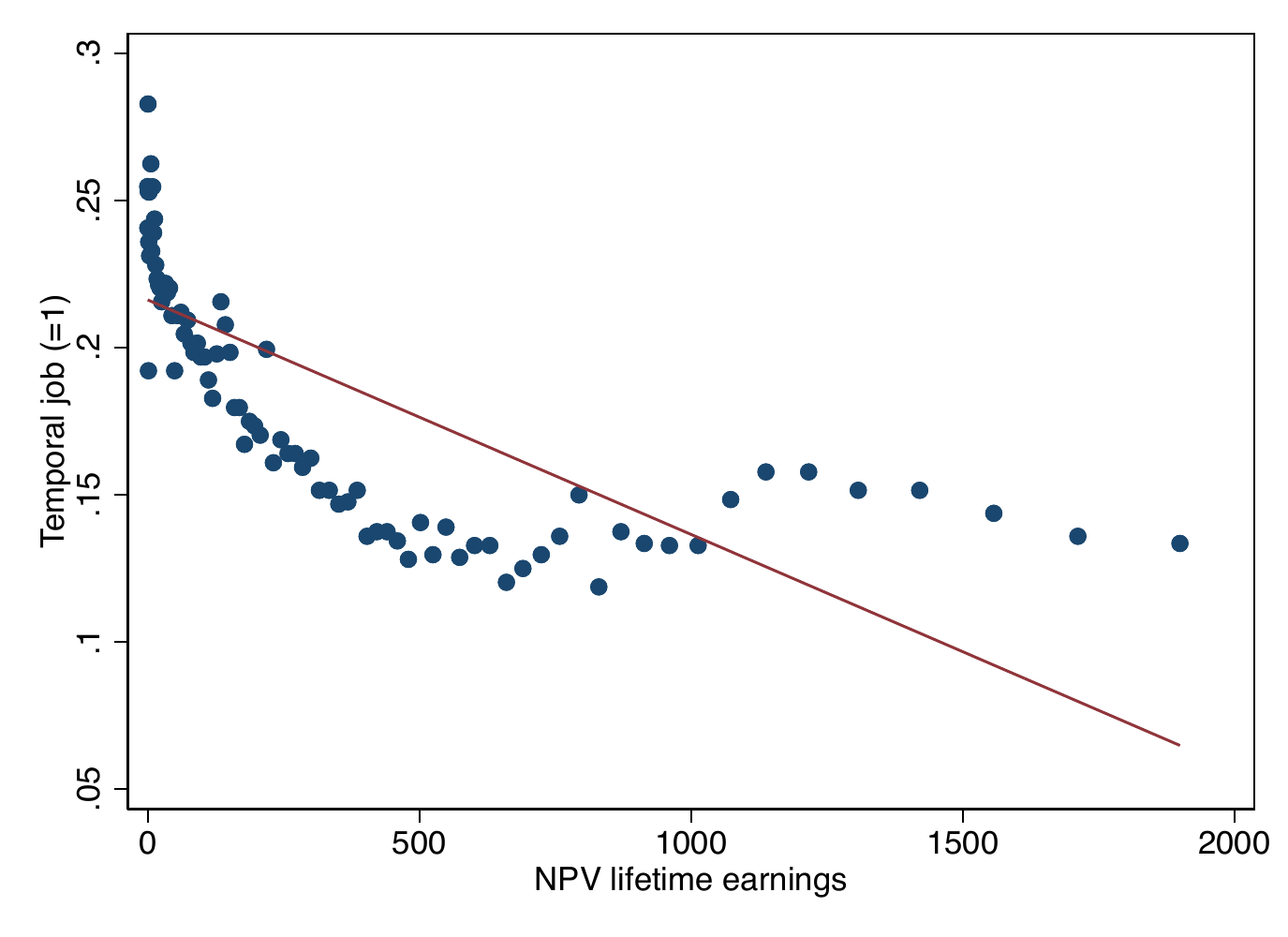}
        \caption{Temporal Job}
    \end{subfigure}
            \vskip\baselineskip
    \begin{subfigure}[b]{\textwidth}
        \centering
        \includegraphics[width=0.4\linewidth]{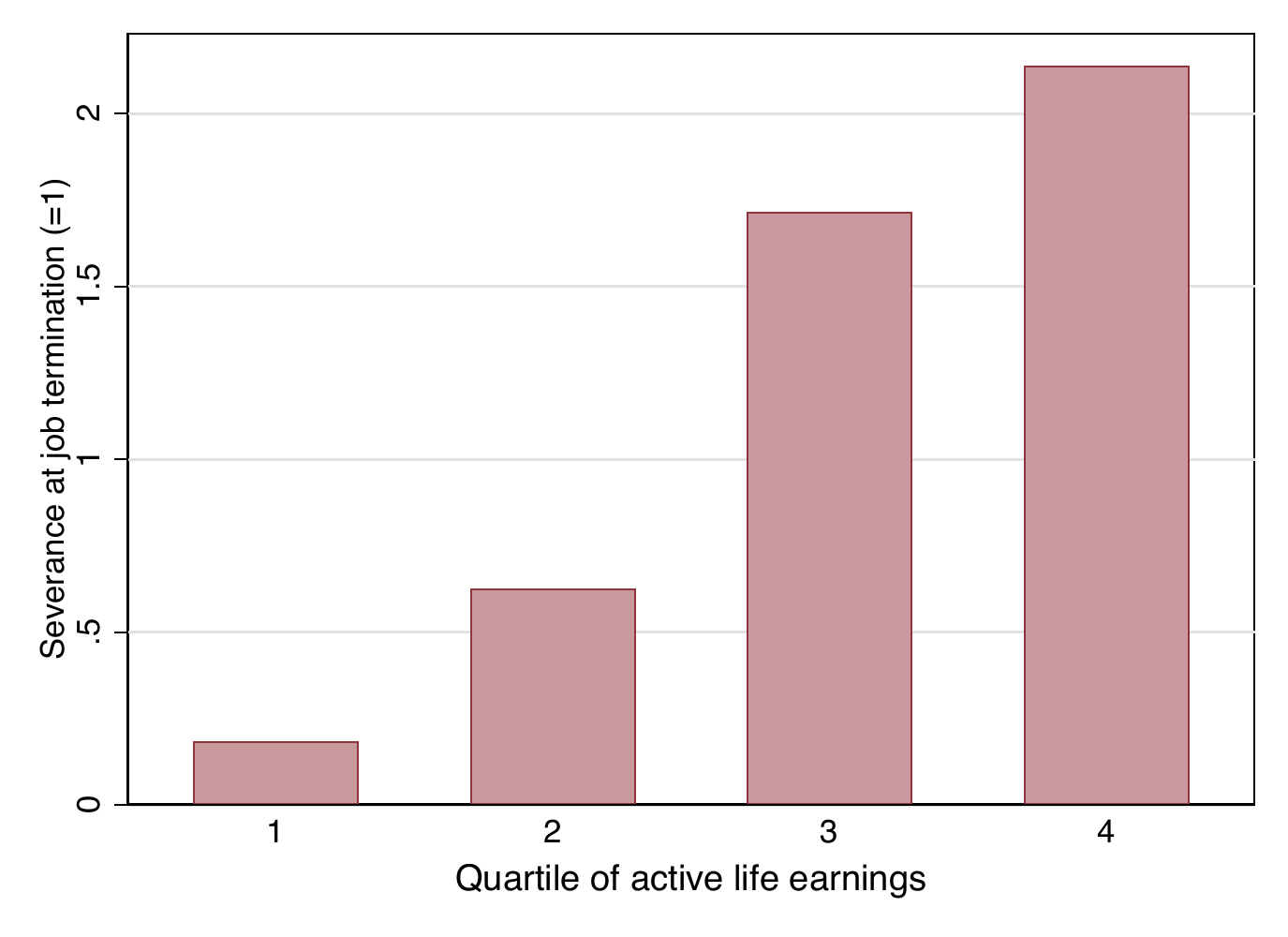}%
        \hfill
        \includegraphics[width=0.4\linewidth]{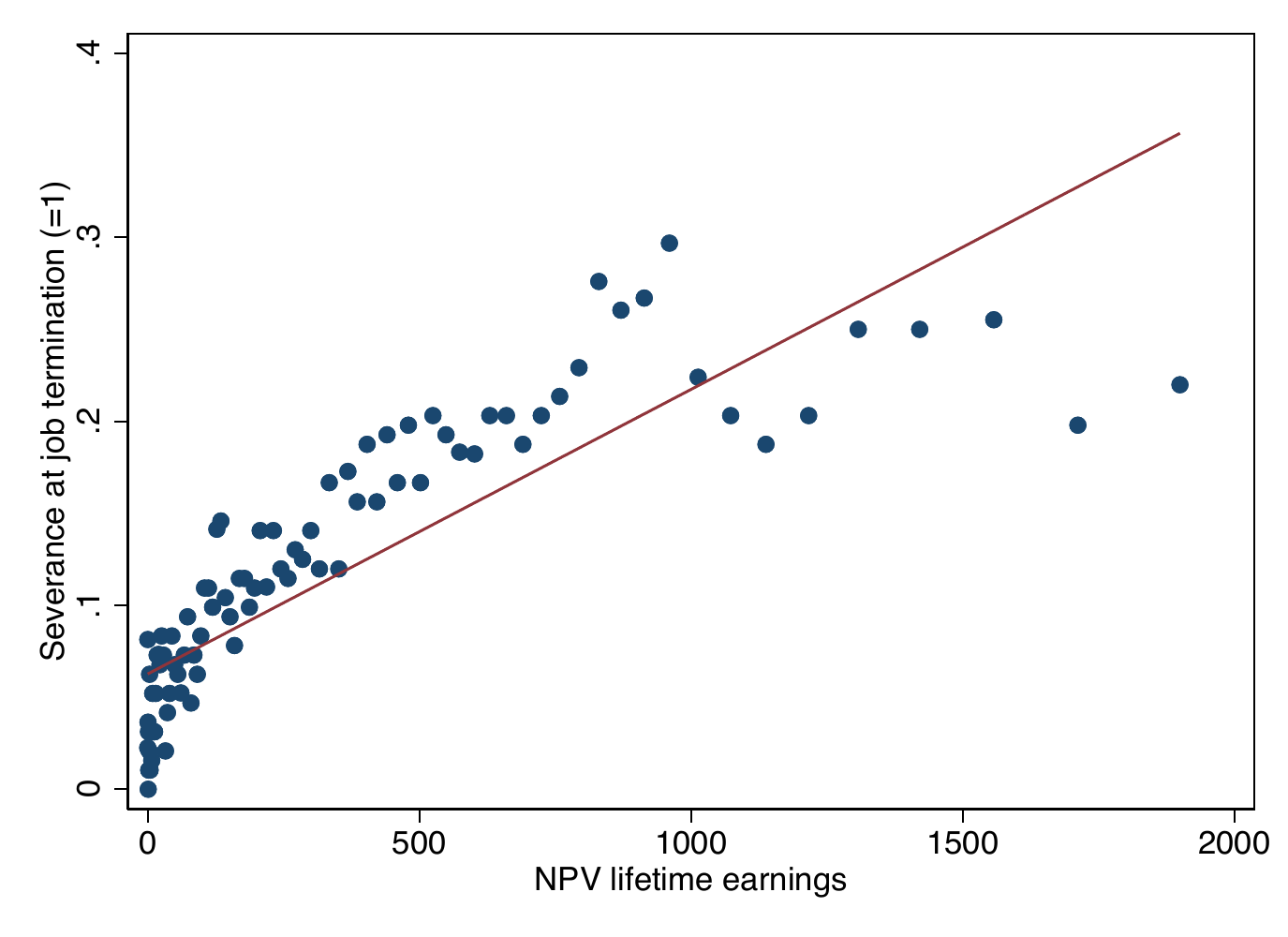}
        \caption{Severance payment at job termination}
    \end{subfigure}
        \caption{Labor characteristics prior retirement}
\end{figure}

\begin{figure}[htb]
    \centering
    \begin{subfigure}[b]{\textwidth}
        \centering
        \includegraphics[width=0.4\linewidth]{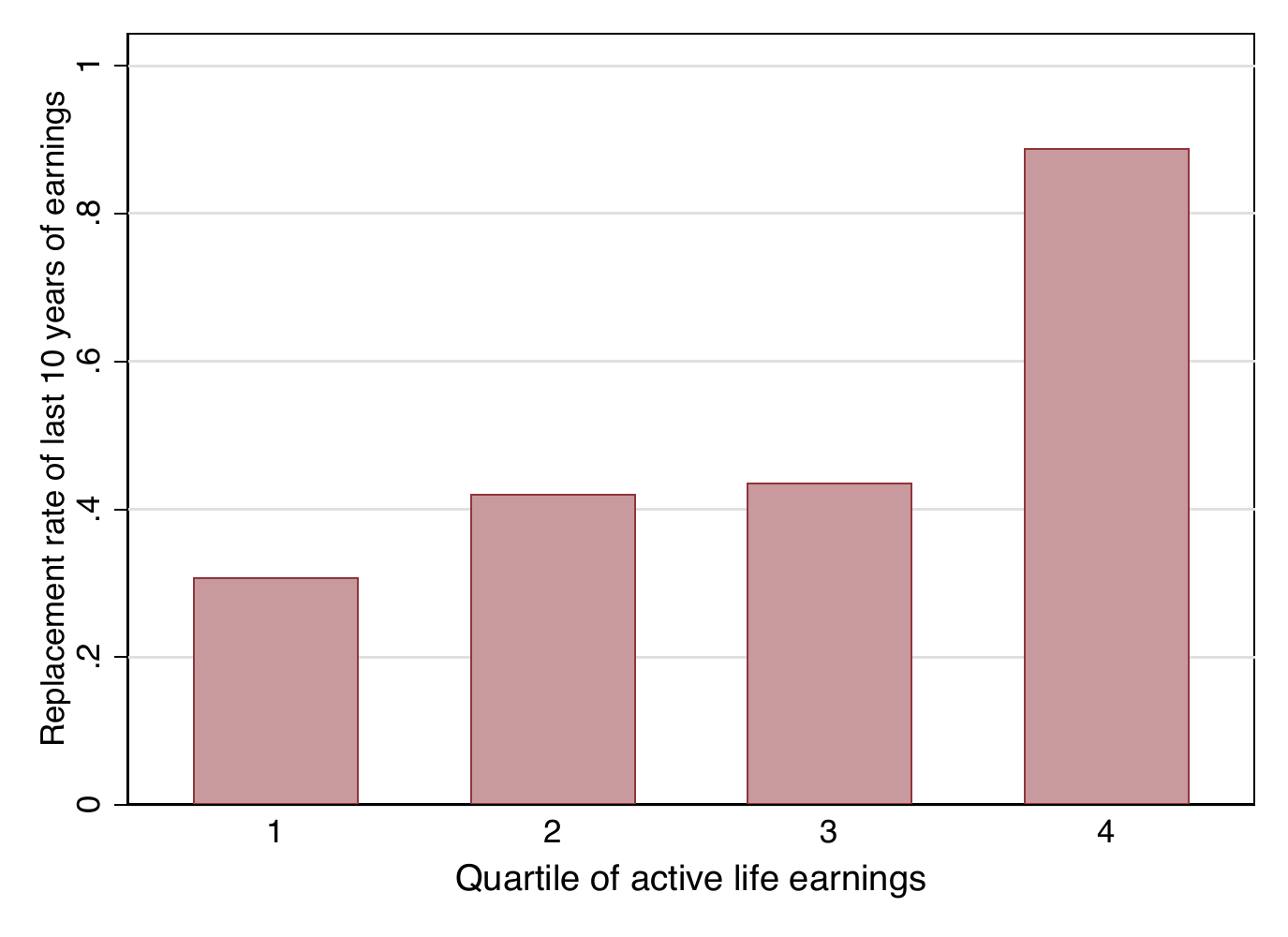}%
        \hfill
        \includegraphics[width=0.4\linewidth]{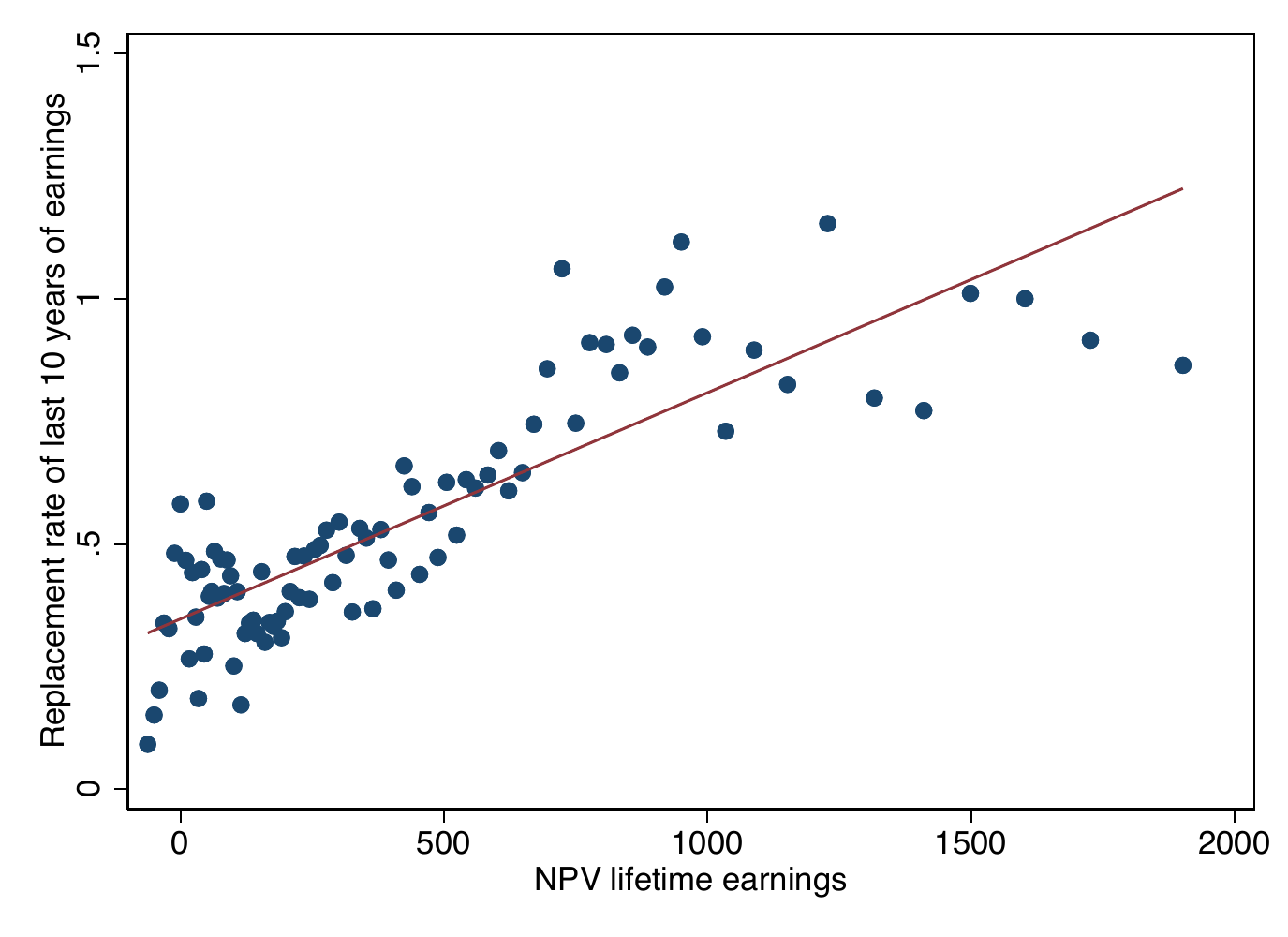}
        \caption{Pension benefit replacement rate of last 10 years of self-reported income}
    \end{subfigure} 
    \vskip\baselineskip
    \begin{subfigure}[b]{\textwidth}
        \centering
        \includegraphics[width=0.4\linewidth]{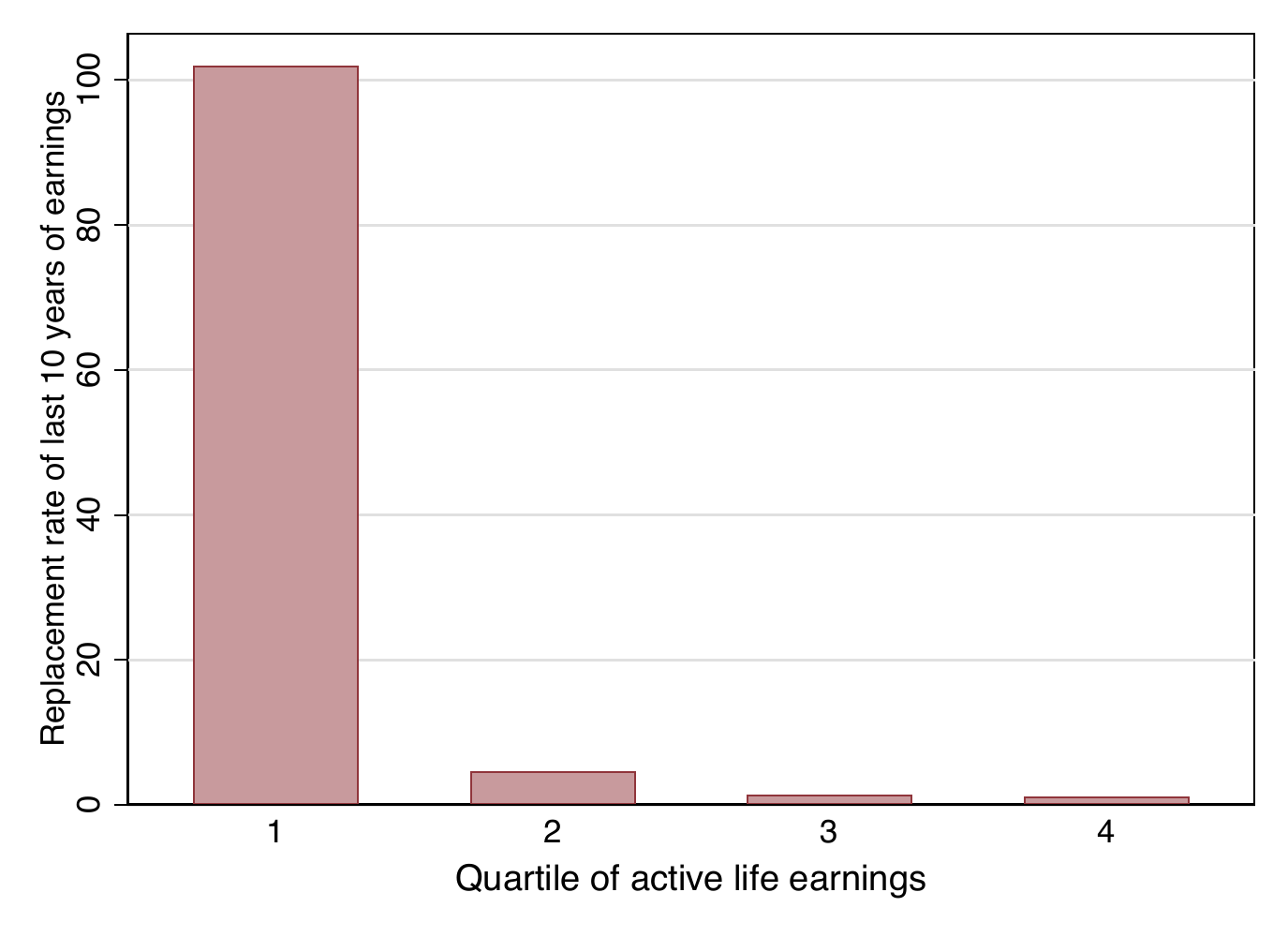}%
        \hfill
        \includegraphics[width=0.4\linewidth]{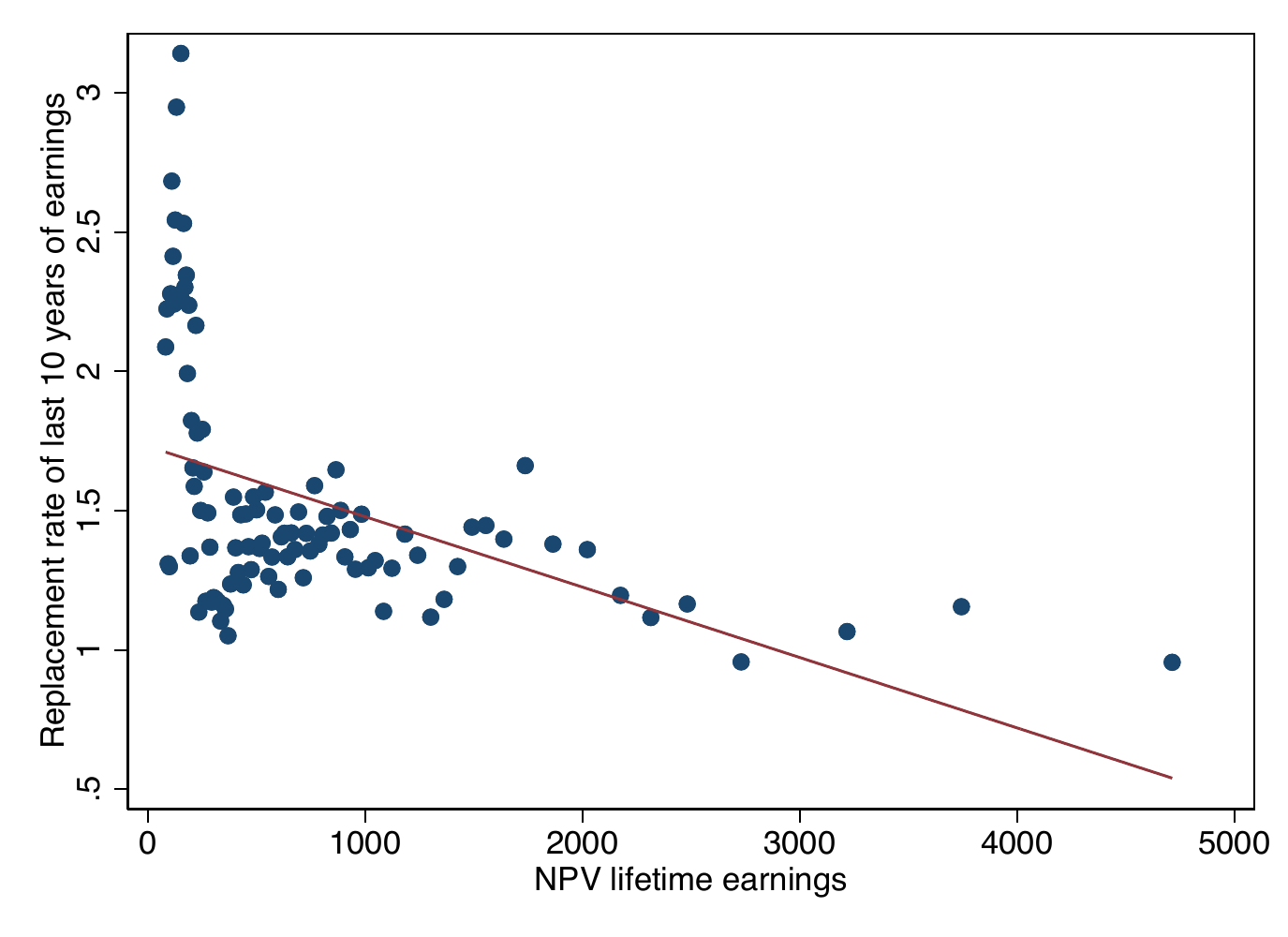}
        \caption{Return rate on pension contribution}
    \end{subfigure}
        \caption{Pension system return on contribution and replacement rate}
             \label{last_ret_preparedness}
\end{figure}

\renewcommand{\thesection}{D}
\setcounter{figure}{0}
\renewcommand\thefigure{D.\arabic{figure}}
\setcounter{table}{0}
\renewcommand\thetable{D.\arabic{figure}}
\section{Heterogeneous life expectancy}
\label{appendix_d_LE}

Life expectancy is not constant across income (see \cite{elo2009social} for a review). I find that workers with a lower present value of lifetime payroll income have a shorter life expectancy. 

To compute the life expectancy by income level I follow \cite{chetty2016association} methodology. Using the survey and administrative data, I restrict the sample to all of the workers that reach retirement age after 2004: 4,114 workers and of those, 631 died during the period 2005-2020. Then, I use the admin data to split the sample of retired workers in those with lifetime payroll earnings above and below the median. To make payroll earnings comparable across time and gender, I control lifetime payroll earnings on retirement date and gender in order to assign workers to each income category. Then, I build empirical survival rates for each group for the 12 years after retirement. Using the fact that mortality rates are log linear (see Chetty et. al. (2016)) I do a income specific log linear approximation of mortality rates to extrapolate survival rates for ages 77 to 86. The fit is almost perfect, with a R-squared of above 97\%. Table \ref{mortality_rate} shows the regressions between the log of mortality rate and age. Note that this extrapolation is conditional on income. Then, I use reports of the Chilean statistical institute of mortality rates for the whole population to extrapolate survival rates for the years 86-94. 

Figure \ref{survival_rate} shows the survival rate for each income group. Workers with lifetime payroll income above the median live longer in average. The life expectancy conditional on reaching retirement age is 86.4 years for those above the median and 83.1 for those below the median. This means that in average those with payroll income above the median have a 18.2\% larger retirement time than those below the median.

\begin{figure}[h!]
    \centering
    \includegraphics[scale=.8]{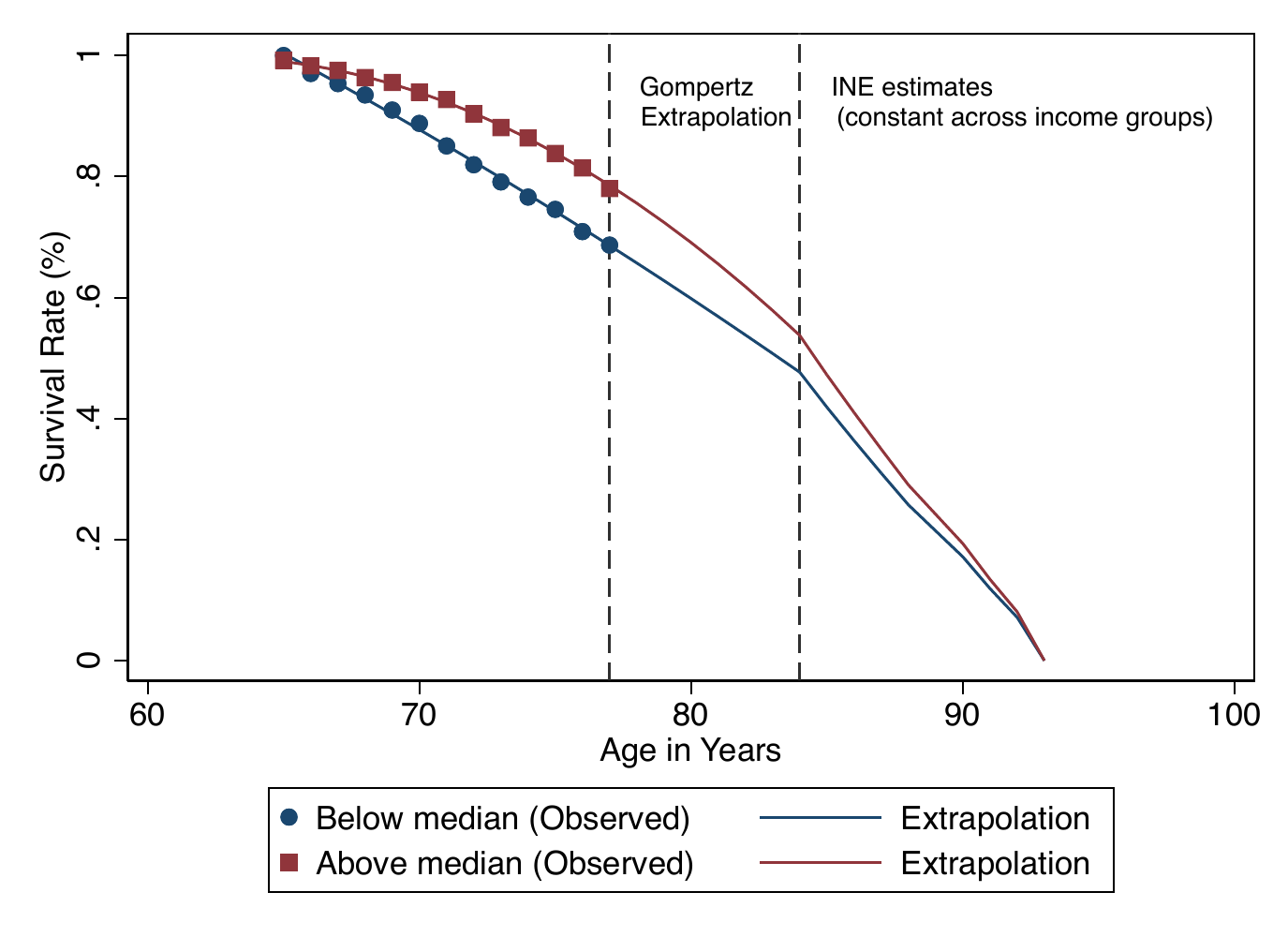}
    \caption{Survival rate and lifetime payroll income}
          \label{survival_rate}
\end{figure}

\begin{table}[h!]
    \centering
    \begin{tabular}{lcc} \hline
 & (1) & (2) \\
 & log(mortality rate) & log(mortality rate) \\ \hline
 &  &  \\
Age & 0.027*** & 0.019*** \\
 & [0.000] & [0.001] \\
Constant & -1.746*** & -1.258*** \\
 & [0.029] & [0.074] \\
 &  &  \\
Observations & 15 & 15 \\
R-squared & 0.997 & 0.963 \\
 Income & Below Median & Above Median \\ \hline
\end{tabular}

    \caption{Mortality rate and age}
    \label{mortality_rate}
    \medskip 
    \begin{minipage}{0.96\textwidth} 
    {\footnotesize \textit{Notes:}  *, **, *** indicates significance at 90\%, 95\% and 99\% of confidence, respectively. Standard errors are inside []. } 
    \end{minipage}
\end{table}

\renewcommand{\thesection}{E}
\setcounter{figure}{0}
\renewcommand\thefigure{E.\arabic{figure}}
\section{Other Figures and Tables}
\label{apendix_E}
\begin{figure}[t]
\centering
   \begin{subfigure}{0.49\linewidth} \centering
     \includegraphics[scale=0.58]{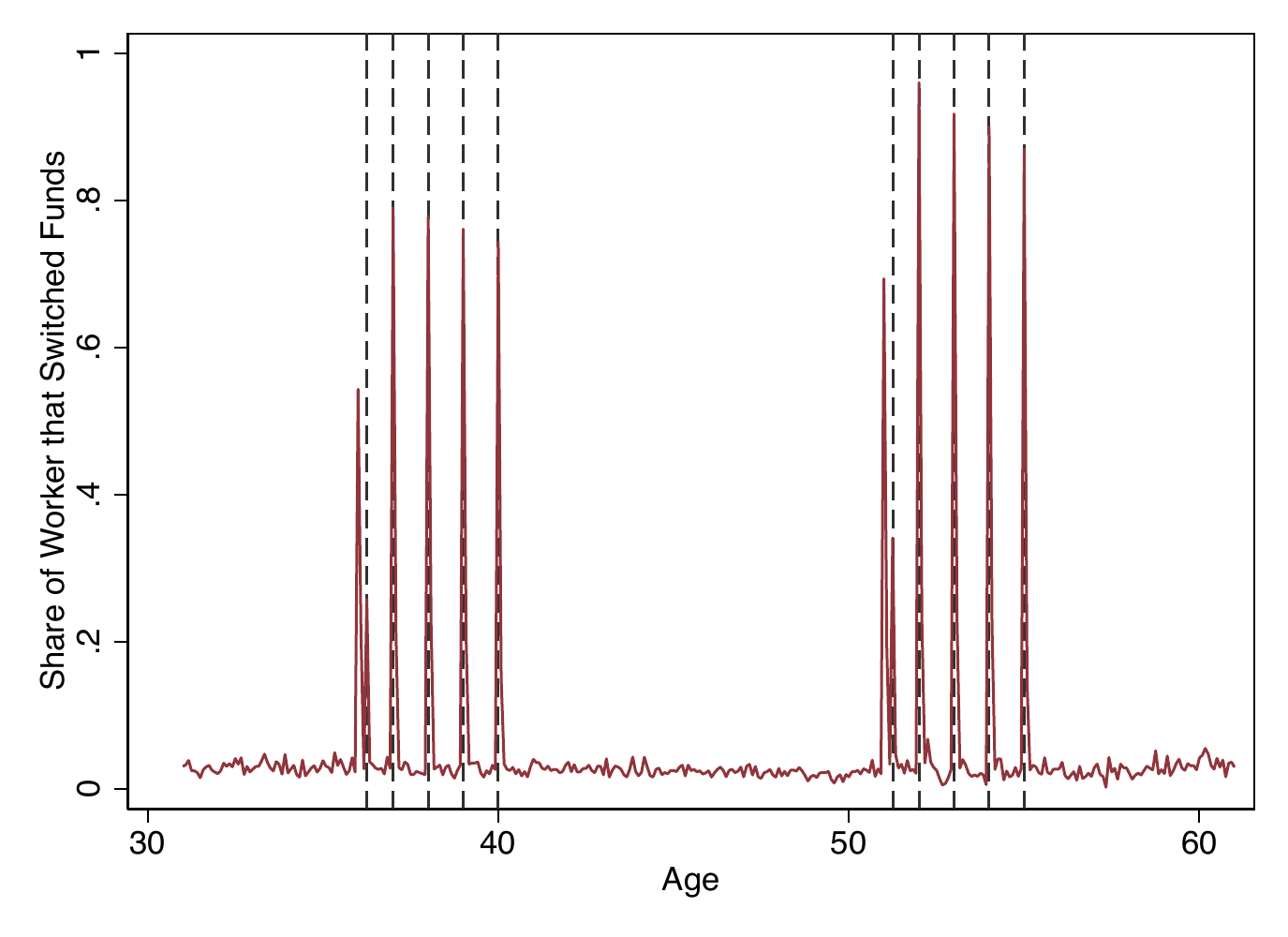}%
     \caption{Females}\label{fig:figA}
   \end{subfigure}
   \begin{subfigure}{0.49\linewidth} \centering
     \includegraphics[scale=0.58]{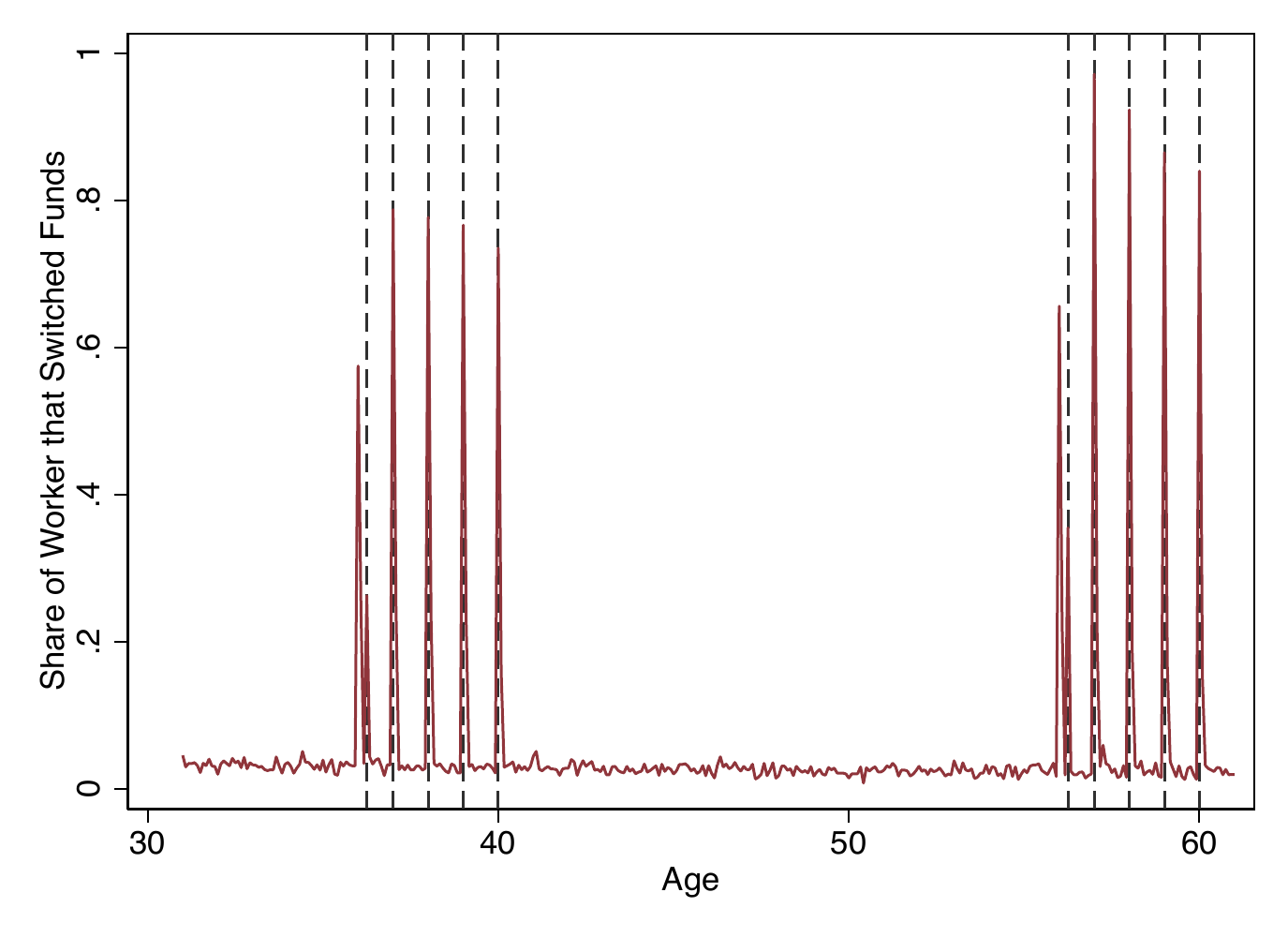}
     \caption{Males}\label{fig:figB}
   \end{subfigure}
    \caption{Share of switchers by age (months) in period 2008-2014}
\end{figure}

\begin{figure}[h!] 
  \begin{subfigure}[b]{0.5\linewidth}
    \centering
    \includegraphics[width=1\linewidth]{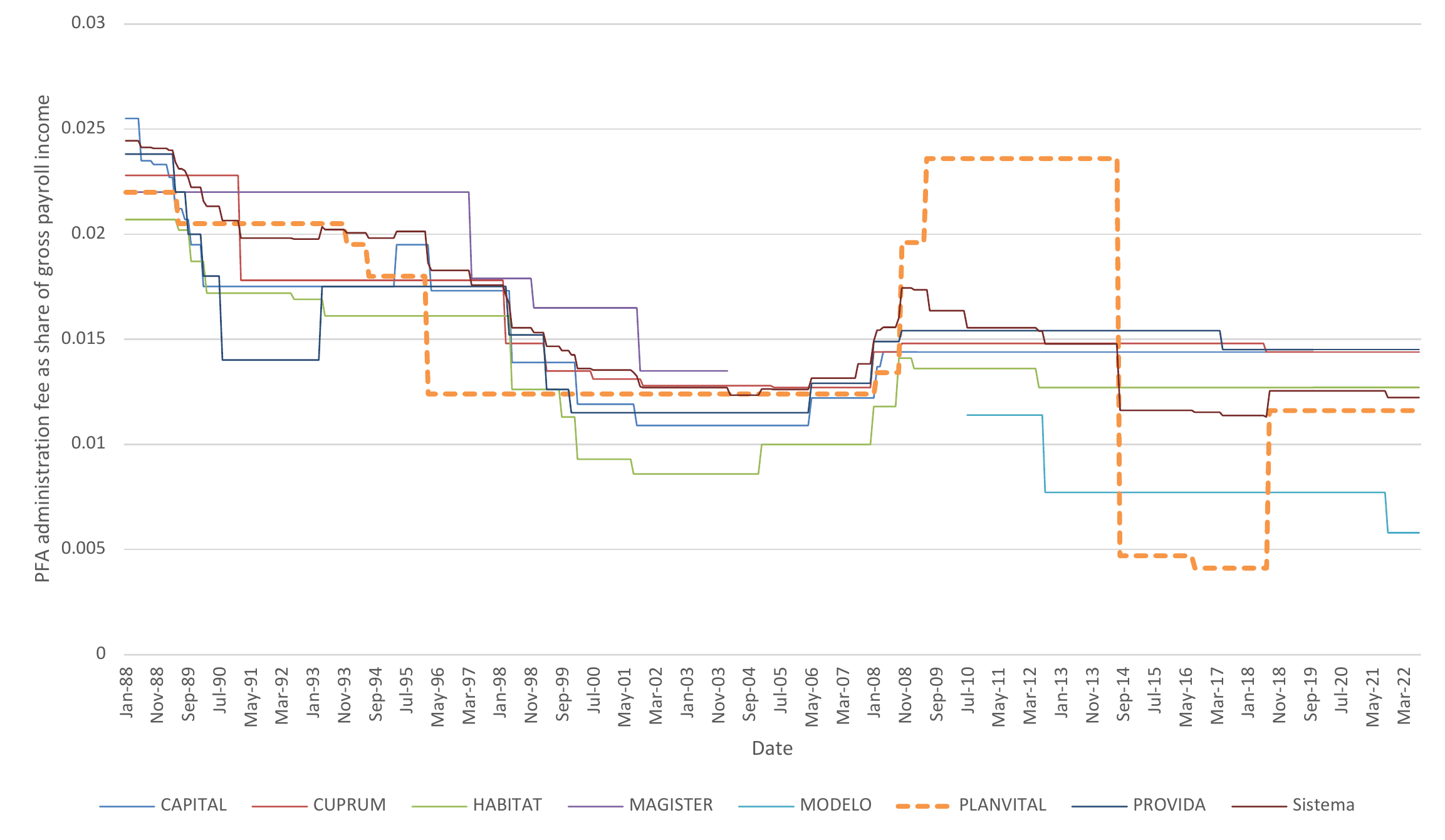} 
    \caption{Management fees since the system inception} 
    \label{fig7:a} 
    \vspace{4ex}
  \end{subfigure}
  \begin{subfigure}[b]{0.5\linewidth}
    \centering
    \includegraphics[width=1\linewidth]{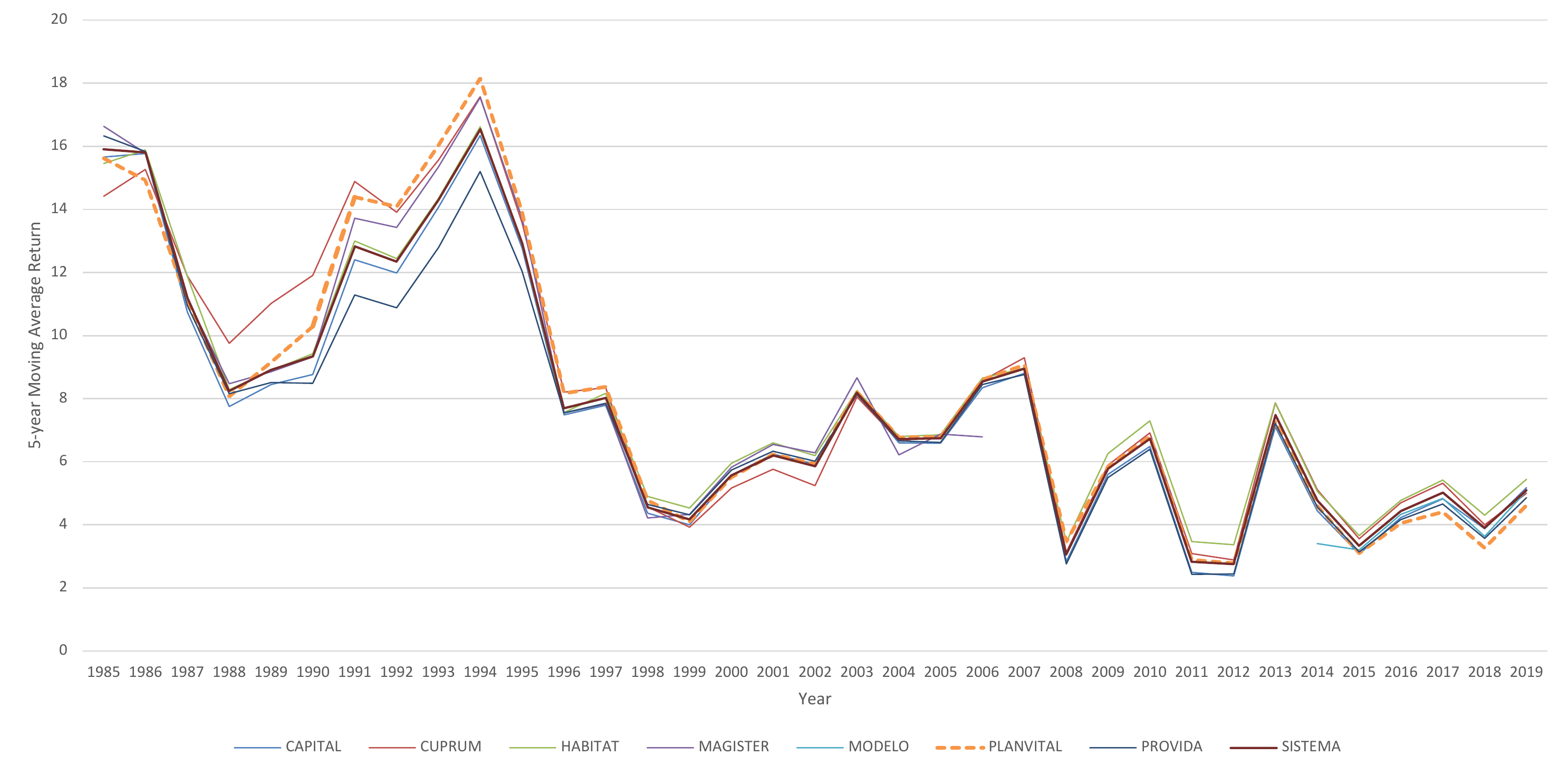} 
    \caption{5-year moving average since the inception of the system for fund C} 
    \label{fig7:b} 
    \vspace{4ex}
  \end{subfigure} 
  \caption{Pension fund administrators management fees and returns}
  \label{historical_fees_return_ts} 
\end{figure}

\end{document}